\documentclass[11pt]{article}
\usepackage[letterpaper,margin=1in]{geometry}
\usepackage{iftex}
\usepackage{authblk}
\ifPDFTeX
  \usepackage[T1]{fontenc}
  \usepackage{lmodern}
\else
  \usepackage[no-math]{fontspec}
\fi

\usepackage{amsmath,amssymb,amsthm,mathtools,bm,mathrsfs}
\usepackage{graphicx,booktabs,array,microtype}
\usepackage{xcolor,url,needspace,placeins}
\usepackage[colorlinks=true,linkcolor=blue!55!black,citecolor=blue!55!black,urlcolor=blue!55!black]{hyperref}
\usepackage[nameinlink,noabbrev]{cleveref}
\newtheorem{theorem}{Theorem}
\newcounter{fact}
\newtheorem{lemma}[theorem]{Lemma}
\newtheorem{proposition}[theorem]{Proposition}
\newtheorem{corollary}[theorem]{Corollary}
\theoremstyle{definition}

\newtheorem{problem}{Problem}
\theoremstyle{remark}
\newtheorem{remark}[theorem]{Remark}
\newcommand{\ket}[1]{\lvert #1\rangle}
\newcommand{\bra}[1]{\langle #1\rvert}
\newcommand{\braket}[2]{\langle #1\mid #2\rangle}
\newcommand{\norm}[1]{\left\lVert #1\right\rVert}

\newcommand{\supp}{\operatorname{supp}}
\newcommand{\diam}{\operatorname{diam}}
\newcommand{\dist}{\operatorname{dist}}

\newcommand{\CX}{\operatorname{CX}}
\newcommand{\CCX}{\operatorname{CCX}}
\newcommand{\SWAP}{\operatorname{SWAP}}
\newcommand{\CSWAP}{\operatorname{CSWAP}}
\newcommand{\Dph}{D_{\mathrm{ph}}}
\newcommand{\calA}{\mathcal A}
\newcommand{\PREPARE}{\operatorname{PREPARE}}
\newcommand{\SELECT}{\operatorname{SELECT}}
\newcommand{\calU}{\mathcal U}
\newcommand{\calT}{\mathcal C_{\mathrm{pre}}}

\newcommand{\ind}[1]{[\,#1\,]}
\newcommand{\reg}[1]{\mathsf{#1}}
\newcommand{\gcnt}{\mathsf{g}}
\newcommand{\dcnt}{\mathsf{d}}
\newcommand{\acnt}{\mathsf{a}}
\definecolor{datacol}{RGB}{45,51,57}
\definecolor{ctrlcol}{RGB}{40,92,148}
\definecolor{selcol}{RGB}{151,60,61}
\definecolor{workcol}{RGB}{116,123,133}

\makeatletter
\let\standardmaketitle\maketitle
\renewcommand{\maketitle}{%
  \ifx\@author\@empty
    \begin{center}{\Large\bfseries\@title\par}\end{center}\vspace{1em}%
  \else\standardmaketitle\fi}
\makeatother

\newcommand{\pku}{Center on Frontiers of Computing Studies, School of Computer Science, Peking University, Beijing 100871, China}
\newcommand{\bnu}{School of Artificial Intelligence, Beijing Normal University, Beijing 100875, China}
\newcommand{\nju}{State Key Laboratory for Novel Software Technology, Nanjing University, Nanjing 210023, China}
\newcommand{\hefei}{Hefei National Laboratory, Hefei 230088, China}
\newcommand{\oxford}{Mathematical Institute, University of Oxford, Woodstock Road, Oxford OX2 6GG, United Kingdom}
\newcommand{\pkuphy}{School of Physics, Peking University, Beijing 100871, China}

\begin{document}
\hypersetup{pdftitle={Toward Optimal Quantum Circuit Depth for Geometrically Local Hamiltonian Simulation}}
\title{Toward an Optimal Quantum Algorithm for\\Geometrically Local Hamiltonian Simulation}

\author[1,2]{Zheyu Shen\thanks{\texttt{zheyushen@stu.pku.edu.cn}}}

\author[3]{Yusen Wu\thanks{\texttt{yusen.wu@bnu.edu.cn}}}

\author[4,5]{Penghui Yao\thanks{\texttt{phyao1985@gmail.com}}}

\author[1]{Xiao Yuan\thanks{\texttt{xiaoyuan@pku.edu.cn}}}

\author[1,6]{Yukun Zhang\thanks{\texttt{yukun.zhang@maths.ox.ac.uk}}}

\affil[1]{\pku}
\affil[2]{\pkuphy}
\affil[3]{\bnu}
\affil[4]{\nju}
\affil[5]{\hefei}
\affil[6]{\oxford}


\date{}

\maketitle
\begin{abstract}
Understanding the dynamics of interacting quantum systems is a central goal of quantum simulation. For geometrically local Hamiltonians, locality suggests that digital simulation should retain much of the parallelism of the underlying dynamics. For a target accuracy $\epsilon$, we construct a nearest-neighbor quantum circuit that simulates the dynamics of an $n$-qubit, time-independent, finite-range Hamiltonian with normalized local interaction strength on a fixed-dimensional lattice for time $T\ge1$, with circuit depth $\mathcal O(T\log(nT/\epsilon))$ and gate count $\mathcal O(nT\log(nT/\epsilon))$. This reduces the polylogarithmic overhead of Haah, Hastings, Kothari, and Low~\cite{HHKL} to a single logarithm. On the logarithmic-scale blocks produced by the Lieb--Robinson decomposition, a shallow but highly accurate implementation is essential for reducing the circuit depth. Our key observation is that spatial localization also controls the cost of high precision. In each block, the error of a constant-order product formula is compensated for by a correction, yielding the required exponential accuracy. The correction can be implemented with depth proportional to the block diameter, introducing no additional asymptotic depth scale beyond that required by spatial localization.

Furthermore, we prove an unconditional circuit-depth lower bound for uniform XXZ Hamiltonians: at fixed evolution time $T$ in the stated regime and sufficiently small fixed error $\epsilon$, any nearest-neighbor quantum circuit simulating these dynamics requires depth $\Omega(\log n/\log\log n)$. This matches our $\mathcal O(\log n)$ upper bound up to a doubly logarithmic factor, establishing near-optimal circuit depth for lattice Hamiltonian simulation in this regime. Together, these results clarify the role of locality and precision in determining the circuit complexity of quantum dynamics and suggest broader principles for spatially constrained quantum computation.
\end{abstract}
\tableofcontents

\section{Introduction}
Quantum computers offer a natural framework for simulating the dynamics of complex quantum systems~\cite{Feynman82,Lloyd}. In lattice spin systems, Hamiltonian dynamics govern the transport of excitations and the propagation of quantum information~\cite{Bose03,OsborneLinden04}, providing a microscopic picture of how local interactions give rise to collective behavior. Simulating these dynamics is therefore essential for understanding nonequilibrium quantum phenomena and represents a central application of quantum computing. A key challenge is to develop quantum simulation algorithms that achieve high accuracy while remaining efficient as the system size and evolution time increase.

A broad range of quantum algorithms has been developed for Hamiltonian simulation, including Trotter--Suzuki product formulas~\cite{Trotter1959,Suzuki1990}, randomized methods such as qDRIFT~\cite{Campbell2019}, and approaches based on quantum signal processing and qubitization~\cite{QSP,Qubitization}.
These methods offer different advantages in their dependence on the evolution time, system size, and target precision.
However, translating these advantages into low-depth circuits for lattice systems requires careful consideration of both error accumulation and implementation costs~\cite{ChildsConcrete18,Trotter}.
Product formulas naturally accommodate local interactions and parallel execution, but for fixed-order formulas, bounds on the number of steps required to control the global simulation error generally grow with the system size and inverse target accuracy~\cite{Trotter}.
The original qDRIFT algorithm instead has a sampling cost governed by the sum of the interaction strengths, which can become substantial for large systems~\cite{Campbell2019}.
Quantum signal processing combined with qubitization achieves optimal query complexity~\cite{QSP,Qubitization}, but the gate count and circuit depth also depend on the cost of implementing each oracle call.
In particular, state preparation and controlled operations in a block encoding introduce implementation costs that must be accounted for when translating query complexity into circuit resources~\cite{Qubitization,ChildsConcrete18}.
At fixed evolution time, local dynamics propagate information predominantly over a bounded distance, but simulating the full system requires controlling the accumulation of small errors across all sites. It is natural to ask: \emph{What is the minimum circuit depth required for Hamiltonian simulation?}

To address this question, we first study the cost of high precision in the local evolution task. Fixing the evolution time allows us to focus on how the region's size and the required accuracy determine circuit depth. We therefore ask: \emph{For a local Hamiltonian on a $D$-dimensional lattice of side length $M$, what is the minimum circuit depth needed to simulate unit-time evolution to error $e^{-M}$?} Here, the main difficulty is to achieve exponential accuracy without losing parallelism. For a fixed-order product formula, the standard error bound reaches $e^{-M}$ only after exponentially many sequential steps~\cite{Trotter}. On other hand, in the context of Taylor-series expansion simulation, the Hamiltonian's absolute coefficient sum can grow with the block volume $\Theta(M^D)$~\cite{Taylor}. Moreover, a coherent expansion requires a superposition of different sequences of local operations, whose locations are specified by quantum control registers. These controls must reach the sites on which the operations act through nearest-neighbor gates; sending them across the block one at a time can cost depth $\mathcal O(M)$ for every operation. The challenge is to overlap this communication while preserving the order of operations on overlapping supports.

We give a quantum algorithm that achieves additive error $e^{-\Omega(M)}$ with circuit depth $\mathcal O(M)$. The idea is to obtain high precision by correcting an inexpensive product-formula approximation. The product formula carries out the main evolution in parallel, while a linear combination of unitaries (LCU) supplies the small correction needed for exponential accuracy. We show that this correction can also be assembled from local operations. To preserve parallelism, we allow the control information for successive operations to travel through the region concurrently, while operations on overlapping sites retain their required order. In this way, both the evolution and its correction preserve spatial parallelism, and the total circuit depth is proportional to the region's diameter.

One of the most significant applications of this local algorithm is to improve the simulation of larger lattice systems. For $D$-dimensional geometrically local Hamiltonians, Haah, Hastings, Kothari, and Low (HHKL) established an algorithm with circuit depth
$\mathcal O(T\log^D(nT/\epsilon))$
and gate count
$\mathcal O(nT\log^D(nT/\epsilon))$,
where $n$ is the number of qubits, $T\ge 1$ is the evolution time, and $\epsilon$ is the simulation error~\cite{HHKL}.
Building on Osborne's one-dimensional construction~\cite{Osborne}, their approach uses Lieb--Robinson bounds~\cite{LR,NS} to decompose each short-time evolution into a constant number of layers of evolutions on overlapping regions of diameter $\mathcal O(\log(nT/\epsilon))$.
The resulting blocks can be simulated independently within each layer. Taking $M=\Theta(\log(nT/\epsilon))$ and using our algorithm for these block evolutions gives depth $\mathcal O(T\log(nT/\epsilon))$ and gate count $\mathcal O(nT\log(nT/\epsilon))$ for time-independent lattice Hamiltonians with normalized local interaction strength. Thus the reduction of the HHKL overhead to a single logarithm follows as a consequence of our local simulation algorithm.

Finally, we return to our opening question about the minimum circuit depth for Hamiltonian simulation. We prove, to our knowledge, the first unconditional lower bound that forces geometric circuit depth to grow with system size for the constant-time dynamics of a fixed, uniform nearest-neighbor Hamiltonian at constant global error. For uniform XXZ dynamics, the required depth is $\Omega(\log n/\log\log n)$ in the stated regime. Earlier gate-count bounds~\cite{HHKL} and depth bounds governed by evolution time or target precision~\cite{ParallelLower,Parallel} remain compatible with constant-depth simulation as the system grows at fixed evolution time and error. Our result rules out this possibility for geometrically local circuits simulating uniform XXZ dynamics: achieving constant global accuracy requires depth growing almost logarithmically with $n$, even when the circuit is tailored to the fully known Hamiltonian. The difficulty is that ballistic propagation alone gives only a constant depth bound at fixed time, while individual long-distance transitions can be too weak to detect at constant error. The proof amplifies weak propagation over many separated regions and controls the effect of interactions relative to the propagation signal itself. Our $\mathcal O(\log n)$ upper bound at fixed time and sufficiently small constant error matches this lower bound up to a doubly logarithmic factor. Our method provides shallow circuits for high-precision many-body evolution, giving quantum subroutines for quantum many-body simulation problems, sheding lights to practical quantum advantages in solving complex quantum systems

\subsection{Problem Statement}
Let $\Lambda\subset\mathbb Z^D$ be a finite rectangular lattice, with a fixed number of data qubits at each site and fixed dimension $D\ge1$. Distances are measured in the lattice metric $\dist(x,y):=\sum_{j=1}^D|x_j-y_j|$. For nonempty sets $X,Y\subseteq\Lambda$, we write
\[
 \diam(X):=\max_{x,x'\in X}\dist(x,x'),\qquad
 \dist(X,Y):=\min_{x\in X,\,y\in Y}\dist(x,y).
\]
The input Hamiltonian is specified by explicit Hermitian local matrices defining
\begin{equation}
 H:=\sum_{Z\subseteq\Lambda}h_Z,\qquad
 \diam(Z)\le R,\qquad
 \max_{x\in\Lambda}\sum_{Z\ni x}\norm{h_Z}\le J.
 \label{eq:model}
\end{equation}
Each term is assigned a nonempty set of lattice sites $Z$, and this convention assigns a location to scalar terms as well. Here, we assume $R=\mathcal  O(1)$, whereas $J>0$ bounds the local interaction strength. Throughout, $\norm A:=\sup_{\norm\psi=1}\norm{A\psi}$ denotes the operator norm; the same notation denotes the Euclidean norm on vectors. We write $\supp(A)$ for the smallest set on which $A$ acts nontrivially.

\begin{problem}[Local operator evolution]\label{prob:local-evolution}
Given a Hamiltonian $H_M$ satisfying Eq.~\eqref{eq:model} with $J=1$ on a $D$-dimensional region of side length $M$, construct a nearest-neighbor circuit $\mathcal U$ of minimum depth such that
\[
 \norm{\mathcal U\iota_0-\iota_0e^{-iH_M}}\le e^{-M},
\]
where $\iota_0\ket\psi:=\ket0_{\mathrm{anc}}\ket\psi$ initializes all ancillary qubits in the all-zero state.
\end{problem}

Problem~\ref{prob:local-evolution} isolates the dependence of circuit depth on spatial extent and precision at fixed evolution time. It also provides a building block for simulating larger systems: Fact~\ref{fact:spatial} decomposes each short-time evolution into layers of evolutions on regions of logarithmic diameter. Simulating these regions to exponentially small error controls the accumulated simulation error, so a circuit of depth $\mathcal O(M)$ for the local task would give logarithmic dependence on system size and inverse error in the following global task.

\begin{problem}[Hamiltonian Simulation]\label{prob:global-evolution}
Given an $n$-qubit lattice Hamiltonian $H$ satisfying Eq.~\eqref{eq:model}, an evolution time $T\in\mathbb R$, and a target error $\epsilon\in(0,1)$, construct a nearest-neighbor circuit $\mathcal U$ of minimum depth such that
\begin{equation}
 \norm{\mathcal U\iota_0-\iota_0e^{-iHT}}\le\epsilon,
 \label{eq:embedding}
\end{equation}
\end{problem}

\subsection{Main Results}\label{sec:spatial-reduction}
We first state our result for local evolution, then its application to global simulation and the depth lower bound for uniform XXZ dynamics. Section~\ref{sec:insight} explains the constructions behind these results.

\begin{theorem}[Informal]\label{thm:local-informal}
For the local evolution task (defined in Problem~\ref{prob:local-evolution}) on a $D$-dimensional region of side length $M$, there is a nearest-neighbor circuit with error at most $e^{-M}$ and depth $\mathcal O(M)$. It uses $\mathcal O(M^{D+1})$ gates and $\mathcal O(M^D)$ ancillary qubits.
\end{theorem}

Theorem~\ref{thm:local-informal} achieves exponential accuracy with depth proportional to the region's diameter, even though the region contains $\Theta(M^D)$ qubits.
The challenge is to combine high precision with geometric parallelism.
For a product formula of fixed order $p$, the standard error bound after $r$ time steps is $\mathcal O(M^D/r^p)$, so certifying error $e^{-M}$ from this bound requires exponentially many steps~\cite{Trotter}.
Direct Taylor-series LCU simulation improves the precision dependence, but its number of time segments is governed by the absolute coefficient sum of $H$, which can scale with the volume $\Theta(M^D)$~\cite{Taylor}.
Even an expansion into local correction operators does not by itself give a shallow circuit: distributing the controls across the region separately for every selected operation can incur a traversal cost each time.
Our construction uses the product formula for the main evolution and expands only its residual error.
The elementary correction operators retain bounded support, and successive controls are distributed concurrently while preserving the order of overlapping operations.
This brings coefficient processing, selected operations, and uncomputation within the same $\mathcal O(M)$ depth bound.
Theorem~\ref{thm:block} gives the precise version of this local simulation result.

\paragraph{Application to global evolution.}
Combining Theorem~\ref{thm:local-informal} with Fact~\ref{fact:spatial} gives the following Hamiltonian simulation bound.

\begin{theorem}[Informal]\label{thm:global}
Let $H$ be a geometrically local Hamiltonian on an $n$-qubit lattice of fixed dimension satisfying Eq.~\eqref{eq:model}.
There exists a quantum circuit $\mathcal U$ that simulates the Hamiltonian dynamics $e^{-iHT}$ to additive error at most $\epsilon$ in operator norm, with circuit depth $\mathcal O(JT\log(nJT/\epsilon))$, gate count $\mathcal O(nJT\log(nJT/\epsilon))$, and $\mathcal O(n)$ ancillary qubits, where $J$ is defined in Eq.~\eqref{eq:model}.
\end{theorem}

For finite-range lattice Hamiltonians, interactions with disjoint supports can be applied in parallel, so each fixed-order product-formula step has constant depth.
Achieving logarithmic depth therefore hinges on controlling global error without requiring too many sequential steps~\cite{ChildsSu,Trotter}.
Randomized sampling does not automatically resolve this difficulty: the original qDRIFT algorithm has a sampling cost governed by the total interaction strength, which generally grows with system size~\cite{Campbell2019}.
High precision alone is also insufficient, since the optimal query complexity of QSP does not account for the geometric depth required to implement its oracles~\cite{QSP,Qubitization}.
For a specified local observable, one can reduce these costs by removing gates outside its causal cone and exploiting its limited support~\cite{Trotter,YuXuZhao2025}.
Theorem~\ref{thm:global} instead approximates the entire Hamiltonian dynamics rather than a specified local observable. The central step is to meet this global accuracy requirement while preserving parallelism: the local simulation of Theorem~\ref{thm:local-informal} achieves exponentially small error with depth only linear in its spatial diameter.
On regions of logarithmic diameter~\cite{HHKL}, this gives the required precision without introducing an additional depth overhead.
Theorem~\ref{thm:global-formal} gives the precise version of this global result.

The resulting global simulation has logarithmic depth at fixed evolution time and error. It is natural to ask: \emph{Does our method achieve optimal circuit depth?}
We address this question using the uniform XXZ model on a cubic lattice $\Lambda=\{1,\ldots,\ell\}^D$ with $n=\ell^D$ qubits:
\begin{equation}
 H_{\rm XXZ}
 =J_{\rm hop}\sum_{\langle x,y\rangle}
 \left(\frac{X_xX_y+Y_xY_y}{2}
 +\Delta Z_xZ_y\right).
 \label{eq:lb-hamiltonian}
\end{equation}
Here $X_x,Y_x,Z_x$ are Pauli operators, the sum runs over nearest-neighbor pairs with open boundaries, and the fixed parameter $\Delta$ controls the strength of the $ZZ$ interaction.
We choose $J_{\rm hop}=J/[2D(1+|\Delta|)]$ so that the local interaction strength is bounded by $J$.

\begin{theorem}[Informal]\label{thm:depth-lower-bound}
For the uniform XXZ Hamiltonian in Eq.~\eqref{eq:lb-hamiltonian}, at a constant evolution time $T$ and with sufficiently small constant error, any geometrically local quantum circuit approximating the full evolution requires depth $\Omega\!\left(J|T|+\log n/\log\log n\right).$
\end{theorem}

The lower bound holds for a fixed, uniform nearest-neighbor model with explicitly known couplings.
HHKL proved a $\widetilde{\Omega}(nT)$ gate-count lower bound for time-dependent lattice Hamiltonians~\cite{HHKL}, but at fixed $T$ this does not force depth to grow with $n$, since $\Theta(n)$ disjoint gates fit into one layer.
Known depth lower bounds include $\Omega(\log\log(1/\epsilon))$ for sparse Hamiltonian simulation~\cite{Parallel} and $\Omega(T)$ in oracle depth for sparse Hamiltonians~\cite{ParallelLower}; neither parameter dependence forces depth to grow with $n$ at fixed time and error.
A lower bound based only on ballistic propagation is also constant at fixed time, so establishing growth with $n$ requires detecting much weaker propagation beyond the distance reached with appreciable probability.
A single distant transition may be smaller than the allowed simulation error, and interactions prevent us from treating separated probes as independent.
Our proof uses many separated excitation sources to make weak propagation detectable at constant global error, while controlling the interaction between the probes.
The result therefore gives an unconditional geometric limitation on parallel simulation of an explicit interacting model.

\paragraph{Comparison of the bounds.}
At fixed evolution time and sufficiently small constant global error, our algorithm achieves $\mathcal O(\log n)$ circuit depth for geometrically local Hamiltonians, while the uniform XXZ model requires depth $\Omega(\log n/\log\log n)$ in the common regime. The system-size dependence therefore matches up to an $\mathcal O(\log\log n)$ factor; the joint dependence on time and precision is discussed in Section~\ref{sec:discussion}.

\subsection{Proof Outline}\label{sec:insight}
Here, we outline the main ideas behind our upper and lower bounds, which reveal two complementary aspects of locality. For the upper bound, locality allows distant parts of the lattice to be simulated in parallel. For the lower bound, we exploit the distinction between the strict causal cone of a shallow quantum circuit and the weak but nonzero propagation of Hamiltonian dynamics beyond that cone.

\paragraph{Local evolution (Theorem~\ref{thm:local-informal}).}
To prove Theorem~\ref{thm:local-informal}, we consider a Hamiltonian $H_{\rm blk}$ of normalized local interaction strength on a $D$-dimensional region $\Lambda_{\rm blk}$ of side length $M$. Let $C(h)$ be a fixed-order product formula approximating the block evolution over a short time step $h$. We express the exact evolution as
\[
e^{-iH_{\rm blk}h}=C(h)E(h),\qquad
E(h):=C(h)^\dagger e^{-iH_{\rm blk}h},
\]
where $E(h)$ is the residual evolution that corrects the product-formula error. Since local interactions with disjoint supports can be implemented in parallel, $C(h)$ admits a constant-depth implementation. However, achieving exponentially small error using the product formula alone would require too many sequential steps. Our strategy is therefore to implement the residual evolution coherently using a linear combination of unitaries (LCU). The main challenge is to achieve exponential accuracy while preserving the parallelism of the product formula, \emph{which requires controlling both the coefficient weights of the correction and the depth of its geometric implementation.}

We first address the coefficient weights of the LCU correction. A direct high-precision expansion of the block Hamiltonian is expensive because the sum of its absolute coefficients can grow with the block volume $v=\Theta(M^D)$, resulting in a larger truncation order. The product formula avoids this difficulty by leaving a much smaller residual evolution $E(h)$ to approximate. Specifically, for an order-$p$ product formula repeated $r$ times, local commutator bounds give an accumulated error of $\mathcal O(v/r^p)$~\cite{Trotter}. Choosing a fixed even order $p>D$ and $r=\Theta(M/\log M)$ makes this error decrease with $M$, although only polynomially. To achieve exponential accuracy through LCU, however, we need a stronger guarantee on the expansion coefficients, rather than just the operator-norm error. We establish this by explicitly expanding the generator of $E(h)$ and exploiting the cancellation of its low-order terms, guaranteed by the product-formula order conditions. These cancellations allow us to bound the absolute coefficient sum even when higher-degree terms are assigned weights that grow polynomially with $M$. As a result, the correction can be truncated at degree $\mathcal O(M/\log M)$ with an exponentially small error $e^{-\Omega(M)}$, without increasing the order of the product formula.

The second challenge is to preserve the locality of the LCU correction. To achieve shallow circuit depth, the correction must be expressed in terms of local operations, whose implementation does not require long-range quantum gates. Although the residual evolution $E(h)$ generally acts on the entire block, its generator can be decomposed into local Pauli operators. This follows from the structure of the product formula $C(h)$: since it consists of a constant number of layers of local gates, the residual generator involves conjugating local Hamiltonian terms through only a constant-depth circuit. As a result, the Pauli operators appearing in its expansion have bounded spatial support, independent of the target precision. To maintain this locality over multiple product-formula steps, we expand the complete evolution $[C(h)E(h)]^r$ as a single LCU, keeping the local corrections interleaved with the product-formula steps. This ordering is essential because moving all corrections to the end would require conjugating them through an increasing number of product-formula steps, causing their supports to grow with $r$. By preserving the interleaved structure, we can improve the simulation accuracy without increasing the spatial support of the elementary Pauli operators used in the correction.

Given the above subroutines, the final step is to implement the LCU correction within the required circuit depth. Although the elementary operators in the correction have bounded spatial support, their coherent implementation remains nontrivial. In particular, an LCU requires preparing a superposition of different sequences of local operations and coherently selecting the corresponding operations across the block. For a nearest-neighbor circuit, distributing the control information across a block of diameter $\mathcal O(M)$ may itself require $\mathcal O(M)$ depth. Repeating this communication separately for each operation would therefore introduce a substantial depth overhead, even though the individual operations are local. To overcome this difficulty, we construct an LCU representation using $\mathcal O(M/\log M)$ ordered control records, each specifying either a product-formula step or a local correction operator. The control labels are prepared independently using single-qubit gates, while the correlations needed to reproduce the correct expansion coefficients are incorporated through reversible arithmetic and controlled phases. More importantly, we develop a pipelined distribution scheme that allows successive control records to travel across the block concurrently, while preserving the required order of local operations. Since each record contains only $\mathcal O(\log M)$ qubits, the total depth for distributing and processing all records scales as $\mathcal O(M)$, rather than $\mathcal O(M)$ for each record. Together with coefficient evaluation and uncomputation, this yields an exponentially accurate local simulation with circuit depth $\mathcal O(M)$, gate count $\mathcal O(M^{D+1})$, and $\mathcal O(M^D)$ ancillary qubits. This proves Theorem~\ref{thm:local-informal}.


\paragraph{Hamiltonian Simulation (Theorem~\ref{thm:global}).}
To prove Theorem~\ref{thm:global}, we use the local algorithm of Theorem~\ref{thm:local-informal} to implement the block evolutions in the following spatial decomposition~\cite{HHKL}.

\vspace{6pt}
\noindent\refstepcounter{fact}\label{fact:spatial}\textbf{Fact~\thefact}~(Spatial decomposition~\cite{HHKL}) \emph{
    For a geometrically local Hamiltonian with normalized interaction strength, its constant-time evolution can be approximated by a constant number of layers of local evolutions. Within each layer, disjoint blocks of diameter $\mathcal O(M)$ evolve in parallel, with a total approximation error $ne^{-\Omega(M)}$. 
}
\vspace{6pt}

\noindent We note that the blocks may overlap between layers, and some evolutions run backward in time to correct the errors introduced at block boundaries.
This decomposition reduces the simulation of a large interacting system to a collection of smaller simulation problems. After normalizing the interaction strength, choosing $M=\Theta(\log(nT/\epsilon))$ keeps the sum of the accumulated decomposition error and the errors of the block simulations below the target accuracy $\epsilon$.
By Theorem~\ref{thm:local-informal}, a block of volume $v=\Theta(M^D)$ can be simulated with depth $\mathcal O(M)$, gate count $\mathcal O(vM)$, and $\mathcal O(v)$ ancillary qubits. Since the blocks within each layer can be simulated simultaneously and have total volume $\mathcal O(n)$, each short-time interval has depth $\mathcal O(M)$ and gate count $\mathcal O(nM)$, using $\mathcal O(n)$ ancillary qubits. Composing these intervals and restoring the local interaction strength $J$ gives Theorem~\ref{thm:global}.

\paragraph{Circuit-depth lower bound.}
Our lower-bound proof exploits a fundamental distinction between the exact causal cone of a shallow quantum circuit and the weak but nonzero propagation of Hamiltonian dynamics beyond that cone. A nearest-neighbor circuit of depth $d$ cannot transmit information beyond distance $O(d)$, whereas a local Hamiltonian can propagate excitations over larger distances, even at constant evolution time. To turn this distinction into a depth lower bound, we compare two initial states: the vacuum and a state containing excitations at well-separated sites. Since the vacuum is an eigenstate of the XXZ Hamiltonian, it remains unchanged up to a global phase. If the excitation sources lie outside the causal cones of a set of distant receivers, a depth-$d$ circuit produces identical measurement statistics at these receivers for both inputs. The exact Hamiltonian evolution, however, can transport excitations into the receiver regions, creating a detectable difference that any accurate simulation must reproduce.

The key challenge is that the probability of an excitation reaching a distant receiver can be extremely small. In the single-excitation sector, the XXZ dynamics reduce to a continuous-time quantum walk. Choose a source--receiver distance $\ell_{\rm pr}$ beyond the circuit's causal radius, and let $p_{\rm pr}$ be the arrival probability for one probe in a reference evolution where the probe regions are decoupled. At fixed evolution time, Lemma~\ref{lem:lb-arrival} gives, for sufficiently large $\ell_{\rm pr}$,
\[
 p_{\rm pr}\ge \exp\!\left[-C\ell_{\rm pr}\log\ell_{\rm pr}\right],
\]
where $C>0$ depends only on the fixed evolution time and local model. We amplify this weak signal using many separated probes. For $k$ independent probes in the reference system, let $P_k$ denote the probability that at least one receiver detects an excitation. Then
\[
 P_k=1-(1-p_{\rm pr})^k\ge1-e^{-k p_{\rm pr}}.
\]
Thus a number of active probes proportional to the inverse arrival probability gives a constant collective signal, provided enough probes fit in the lattice. A $D$-dimensional lattice of $n$ qubits accommodates $\Omega(n/\ell_{\rm pr}^D)$ well-separated probes with buffers proportional to $\ell_{\rm pr}$, allowing increasingly distant propagation to become detectable as the system grows.

A crucial difficulty is that these probes are not independent under the original interacting XXZ Hamiltonian. To overcome this, we introduce a reference evolution in which the probe regions are decoupled, allowing each excitation to evolve independently within its own region. We then compare this reference evolution with the actual interacting dynamics. By surrounding each probe with a sufficiently large buffer, we ensure that the probability of an excitation reaching the artificial boundary is much smaller than its probability of reaching the receiver. A state-dependent bound on the evolution difference then shows that the error introduced by decoupling each probe is only a small fraction of its arrival probability. This relative-error estimate is essential: a conventional operator-norm bound would be too coarse to preserve such weak signals. Consequently, the collective detection probability survives in the interacting XXZ system. At fixed evolution time and sufficiently small constant error, balancing the number of available probes against their arrival probabilities gives the depth lower bound $\Omega(\log n/\log\log n)$.

Finally, we establish a separate lower bound on the dependence on evolution time. In the single-excitation sector, the quantum walk exhibits ballistic spreading, with a constant probability of reaching distances proportional to the evolution time. Using its second and fourth moments, we show that a shallow circuit must have sufficient depth to reproduce this propagation~\cite{Konno05,PaleyZygmund32}. Combining this ballistic constraint with the collective-probe argument gives our lower bound on the joint dependence on system size, evolution time, and simulation accuracy. The full argument is given in Section~\ref{sec:lower}.

\subsection{Comparison with previous work}

HHKL~\cite{HHKL}, building on Osborne's one-dimensional construction~\cite{Osborne}, supplies the spatial decomposition recalled in Fact~\ref{fact:spatial}. Applying our local simulation algorithm to its block evolutions yields the improved global depth bound; the local algorithm itself achieves exponential accuracy in depth proportional to the region's diameter, including the coherent controls and their removal.

Product-formula algorithms exploit local commutator structure to improve gate complexity~\cite{ChildsSu,Trotter}, while Taylor expansions, quantum signal processing, and qubitization improve precision or query complexity~\cite{Taylor,Dyson,QSP,Qubitization}. Multiproduct formulas provide another way to improve precision by combining product formulas~\cite{ChildsWiebe,MPF,Aftab,Mizuta}.

Zeng, Sun, Jiang and Zhao~\cite{Compensation} introduced Trotter--LCU compensation and also analyzed a coherent implementation in their Appendix~H. Our construction combines a bounded-support expansion of the residual generator with an ordered, pipelined selection circuit. This gives a nearest-neighbor implementation in which coefficient processing, control distribution, selected operations, and uncomputation together have depth proportional to the block diameter, with bounded ancillary density.

Murota, Kikuchi, Rinaldi, Sauvage and Todo~\cite{PTER} also treat Trotter error as dynamics to be reversed and exploit local Pauli decompositions of the remainder generator. Their probabilistic Trotter error reversal (PTER) produces unbiased estimates through sampled circuits and classical reweighting. Our construction instead implements a coherent approximation to the full evolution, including the geometric cost of selecting the correction operations.

Zhang, Wang and Ying~\cite{Parallel} obtained improved depth dependence on precision in a model allowing arbitrary two-qubit gates, and proved a depth lower bound of $\Omega(\log\log(1/\epsilon))$. Their lower bound concerns precision dependence; ours concerns growth with system size at fixed time and constant global error under nearest-neighbor connectivity.

Beyond improving simulation algorithms, an equally important question is whether their circuit depth can be reduced further. Existing lower bounds have established fundamental limitations in several settings. HHKL proved a spacetime gate-count lower bound for time-dependent lattice Hamiltonians~\cite{HHKL}, while Ref.~\cite{LowSuTongTran2023} established an $\Omega(n^2)$ gate lower bound for generic commuting two-local Hamiltonians. More recently, Zlokapa, Allen and Harrow obtained coefficient-sensitive lower bounds on gate and query complexity using local Hamiltonian constructions~\cite{Zlokapa26}. However, lower bounds on the total number of gates do not necessarily imply lower bounds on circuit depth, since many gates may be executed in parallel. Chia, Chung, Hsieh, Lin, Lin and Shen~\cite{ParallelLower} studied limitations on parallel Hamiltonian simulation, establishing unconditional oracle-depth lower bounds and further circuit-depth results under additional assumptions. These results address important complexity limitations, but do not directly establish a system-size-dependent geometric depth lower bound for the constant-time evolution of a fixed, uniform, nearest-neighbor Hamiltonian.

The use of rare long-distance propagation has a precedent in the precision lower bounds of Berry, Childs, Cleve, Kothari and Somma~\cite[Theorems~2.2 and~7.1]{BCCKS}. Their argument uses weak quantum-walk transitions to obtain query lower bounds. Our proof packs separated probes and bounds the error of decoupling them relative to the arrival probability, so that their combined signal remains detectable in the uniform interacting XXZ model at constant global error.

\subsection*{Acknowledgments}
The authors thank Xiaoming Zhang for helpful discussions. This work is supported by the National Natural Science Foundation of China (NSFC) (Grant Nos.~62501060, 62461160263, 62371050, 62332009, 12347104, and 12361161602), the NSAF (Grant No.~U2330201), the Quantum Science and Technology--National Science and Technology Major Project (Grant Nos.~2021ZD0302901 and 2023ZD0300200), the NSFC/RGC Joint Research Scheme (Grant No.~12461160276), the Beijing Natural Science Foundation (Grant No.~Z250004), the Beijing Science and Technology Planning Project (Grant No.~Z25110100810000), the Natural Science Foundation of Jiangsu Province (Grant No.~BK20243060), the Fundamental and Interdisciplinary Disciplines Breakthrough Plan of the Ministry of Education of China (Grant No.~JYB2025XDXM118), the ``111 Center'' (Grant No.~B26023), and the Fundamental Research Funds for the Central Universities (Grant No.~2026300376).

\subsection*{Organization}
Sections~\ref{sec:architecture}--\ref{sec:error} prove the local simulation result in Theorem~\ref{thm:local-informal} and its global application in Theorem~\ref{thm:global-formal}. Section~\ref{sec:architecture} states the precise local result, Theorem~\ref{thm:block}, and recalls Fact~\ref{fact:spatial}. Section~\ref{sec:expansion} constructs an LCU correction to a fixed-order product formula, and Section~\ref{sec:local} implements it with nearest-neighbor gates. Section~\ref{sec:error} combines the error and resource bounds to prove the local result, then composes the local circuits to obtain the global bound. The depth lower bound in Theorem~\ref{thm:depth-lower-bound} is proved separately in Section~\ref{sec:lower}.

\section{Spatial decomposition and block simulation}\label{sec:architecture}

\subsection{Spatial decomposition}\label{sec:block-reduction}
For a region $Y\subseteq\Lambda$, retain the interactions whose declared supports lie entirely in $Y$:
\begin{equation}
 H_Y:=\sum_{Z\subseteq Y}h_Z,\qquad U_Y(t):=e^{-iH_Yt}.
 \label{eq:restrictiondefinition}
\end{equation}

We absorb the local strength into time and divide the evolution into intervals of normalized length at most one. Set
\begin{equation}
 \mathcal T:=\max\{1,\lceil J|T|\rceil\},\qquad
 M:=c_0\log(2+n\mathcal T/\epsilon),
 \label{eq:globalM}
\end{equation}
where $c_0$ is a sufficiently large model-dependent constant. For $T\ne0$, we write $\operatorname{sgn}(T):=T/|T|$ and define
\begin{equation}
 \bar H:=\operatorname{sgn}(T)H/J,\qquad
 \bar H_Y:=\operatorname{sgn}(T)H_Y/J.
 \label{eq:normalizedrestriction}
\end{equation}
\begin{equation}
 \bar U_Y(t):=e^{-i\bar H_Yt},\qquad
 \tau:=\frac{J|T|}{\mathcal T}\le1.
 \label{eq:normalizedevolution}
\end{equation}
Then $\bar H$ has local strength at most one and $e^{-iHT}=\bar U_\Lambda(\tau)^{\mathcal T}$. 

Fact~\ref{fact:spatial} replaces each normalized interval by a constant number of layers of restricted evolutions. We state its precise form below. Let $\mathcal B_a$ be the disjoint regions in layer $a$, and let $\sigma_{a,Y}\in\{-1,1\}$ specify forward or inverse evolution. Regions in different layers may overlap; $\ell_{\rm buf}$ controls the overlap width.

\begin{lemma}[Formal form of Fact~\ref{fact:spatial}; {\cite[Lemma~6 and proof of Theorem~1]{HHKL}}]\label{lem:spatial}
For the normalized local model, there are constants $C,\mu>0$ and an integer $\ell_0$ such that the following holds. Suppose $\ell_{\rm buf}\ge\ell_0$ and every side of $\Lambda$ is at least $C\ell_{\rm buf}$. For $|t|\le1$, there are $L_{\rm H}=O(1)$ layers as above for which the product
\begin{equation}
 \mathcal V(t):=\mathcal V_{L_{\rm H}}(t)\cdots\mathcal V_1(t),\qquad
 \mathcal V_a(t):=\prod_{Y\in\mathcal B_a}\bar U_Y(\sigma_{a,Y}t),
 \label{eq:hhkl-decomposition}
\end{equation}
obeys
\begin{equation}
 \norm{\bar U_\Lambda(t)-\mathcal V(t)}
 \le Cn e^{-\mu\ell_{\rm buf}}.
 \label{eq:hhkl-error}
\end{equation}
Every region $Y$ has side lengths in $[\ell_{\rm buf},C\ell_{\rm buf}]$, and
$\sum_{a=1}^{L_{\rm H}}\sum_{Y\in\mathcal B_a}|Y|=O(n)$.
The regions, signs, and layer order can be computed classically from the lattice and $\ell_{\rm buf}$.
\end{lemma}

Fact~\ref{fact:spatial} allows the disjoint block evolutions within each layer to run in parallel. The circuit in Section~\ref{sec:local} needs auxiliary registers around each block, so neighboring circuits may overlap even when their data blocks are disjoint. Each circuit occupies only a constant-factor enlargement of its block and therefore overlaps only constantly many neighbors. We divide each layer into a constant number of groups with disjoint circuit regions and process the groups in sequence. This adds only a constant factor to the depth.

Taking $\ell_{\rm buf}=\lceil M/\mu\rceil$ in Eq.~\eqref{eq:hhkl-error} makes the decomposition error at most $Cn e^{-M}$. The resulting blocks have side lengths $\Theta(M)$, so it remains to simulate each block evolution with error $e^{-M}$ in depth $\mathcal O(M)$.
The resource bounds use exact single-qubit gates and nearest-neighbor CNOTs; finite-gate-set synthesis is described in Appendix~\ref{app:gates}.

\begin{theorem}[Short-time block simulation]\label{thm:block}
Fix constants $0<c_-\le c_+$. Suppose $H_{\rm blk}$ satisfies Eq.~\eqref{eq:model} with $J\le1$ on a rectangular $D$-dimensional block $\Lambda_{\rm blk}$ whose side lengths lie in $[c_-M,c_+M]$. The number of sites is $v:=|\Lambda_{\rm blk}|$, so $v=\Theta(M^D)$. There is a constant $M_0$, depending only on the local model and $c_-,c_+$, such that, for $M\ge M_0$ and $0<\tau\le1$, a classically specified nearest-neighbor circuit $\mathcal U$ satisfies
\begin{equation}
 \norm{\mathcal U\iota_0-\iota_0e^{-iH_{\rm blk}\tau}}\le e^{-M}
 \label{eq:blockerror}
\end{equation}
with resources
\begin{equation}
 \begin{aligned}
 \text{circuit depth}\quad\dcnt&=O(M),\\
 \text{gate count}\quad\gcnt&=O(vM),\\
 \text{ancillary qubits}\quad\acnt&=O(v).
 \end{aligned}
 \label{eq:blockresources}
\end{equation}

\end{theorem}

\subsection{LCU circuit overview}\label{sec:lcu-overview}
We start with a product formula $C(h)$ of fixed even order $p>D$, using $r$ steps of length $h=\tau/r$. Its accumulated error is $\mathcal O(v/r^p)$, as shown in Lemma~\ref{lem:pferror}. For $r=\Theta(M/\log M)$, this bound is only polynomially small in $M$ and cannot guarantee the required accuracy $e^{-M}$. We therefore implement the residual correction $E(h):=C(h)^\dagger e^{-iH_{\rm blk}h}$ together with each formula step:
\[
 e^{-iH_{\rm blk}\tau}=[C(h)E(h)]^r.
\]

The LCU arises from the time-ordered expansion of $E(h)$, whose generator is a sum of local Pauli operators. With the product-formula steps kept exact, each term is an ordered unitary product $W_z$ of Pauli insertions and formula steps. Coherently combining these products approximates the full block evolution. Section~\ref{sec:expansion} bounds the expansion coefficients and constructs a finite LCU with error exponentially small in $M$, using $N=\Theta(M/\log M)$ control records. The following circuit encodes this sum; Section~\ref{sec:local} gives its nearest-neighbor implementation.

Figure~\ref{fig:lcu-overview}(a) shows how the circuit combines these terms. The bottom line, $\reg D$, is the block of spins whose state $\ket\psi$ we wish to evolve. The top line, $\reg S$, holds the labels $z$ specifying the ordered products $W_z$. PREPARE creates a superposition of these labels, and SELECT uses them to apply the corresponding products to $\reg D$.

The middle line, $\reg W$, provides workspace for SELECT: it stores intermediate values, such as the number of formula steps already performed, and the copies of controls used near the selected interactions. Thus $\reg S$ selects the LCU terms, $\reg D$ undergoes their evolution, and $\reg W$ supports the control calculation. Section~\ref{sec:distribution} specifies the labels and their preparation probabilities $\pi(z)$.

\begin{figure}[!htbp]
\centering
\includegraphics[width=.90\linewidth]{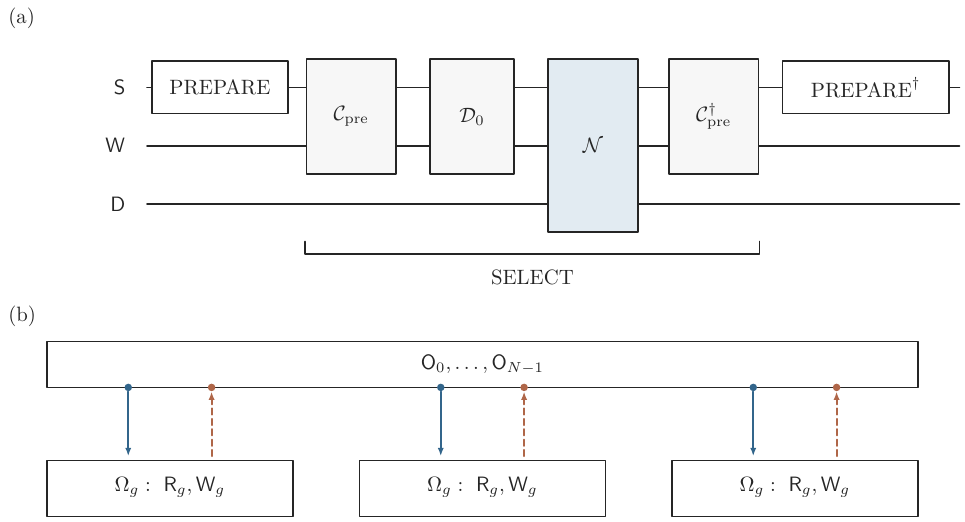}
\caption{LCU circuit and control distribution. (a) The bottom line $\reg D$ contains the simulated spins. The top line $\reg S$ contains the labels selecting the LCU terms and the phase qubits encoding their coefficients. The middle line $\reg W$ supplies workspace for control calculations and distribution. Gates act left to right; the bracketed circuit is SELECT, whose components are described in Section~\ref{sec:local}. (b) Control records are distributed from the source to the shared registers $\reg R_g$ in the data subblocks $\Omega_g$. Solid and dashed arrows show the forward and return paths.}
\label{fig:lcu-overview}
\end{figure}

The circuit must also encode the coefficient multiplying each $W_z$. We do this with additional qubits whose phases reproduce the scalar coefficient factors~\cite[Appendix~E.1]{HHKL}. A complex number of modulus at most one is the average of two phases.
For each scalar factor, we prepare one qubit in $\ket+:=(\ket0+\ket1)/\sqrt2$.
For $z_{\rm c}\in\mathbb C$ with $|z_{\rm c}|\le1$, we choose $\phi:=\arg z_{\rm c}\in(-\pi,\pi]$ (with $\phi:=0$ at zero) and $\theta:=\arccos|z_{\rm c}|$.
These two angles define a diagonal unitary with the required matrix element:
\begin{equation}
 \Dph(z_{\rm c}):=
 \begin{pmatrix}e^{i(\phi+\theta)}&0\\0&e^{i(\phi-\theta)}\end{pmatrix}
 ,\qquad
 \bra+\Dph(z_{\rm c})\ket+=z_{\rm c}.
 \label{eq:phaseaverage}
\end{equation}
Independent phase qubits multiply these scalar factors when their joint $\ket+$ matrix element is taken. Section~\ref{sec:tails} proves that all required factors have modulus at most one.

These phase qubits join the label qubits on the top line $\reg S$. Let $s$ be their joint basis string and $n_s$ their number. Preparing the labels with probabilities $\pi(z)$ and the phase qubits in $\ket+$ gives
\begin{equation}
 \PREPARE\ket0_{\reg S}
 =\sum_{z,s}\sqrt{\frac{\pi(z)}{2^{n_s}}}\ket{z,s}_{\reg S}.
 \label{eq:overviewprepare}
\end{equation}
For the distribution and phase qubits specified in Sections~\ref{sec:distribution} and~\ref{sec:prepare}, the corresponding $\SELECT$ circuit must apply an ordered data unitary $W_z$ and a real phase $\Phi(z,s)$ on each branch:
\begin{equation}
 \SELECT\ket0_{\reg W}\ket{z,s}_{\reg S}\ket\psi_{\reg D}
 =e^{i\Phi(z,s)}\ket0_{\reg W}\ket{z,s}_{\reg S}W_z\ket\psi_{\reg D}.
 \label{eq:selectinterface}
\end{equation}
Here $W_z$ applies the operations in slot order, stopping after the prescribed number of product-formula steps, and $\Phi(z,s)$ is the sum of the selected phases.

Combining SELECT with the preparation and its inverse gives the LCU circuit
\begin{equation}
 B:=\PREPARE^\dagger\SELECT\PREPARE.
 \label{eq:Bdefinition}
\end{equation}

Reversing preparation sums the selected branches coherently: the zero-signal and zero-work matrix element in Eqs.~\eqref{eq:overviewprepare} and~\eqref{eq:selectinterface} is
\begin{equation}
 \iota_0^\dagger B\iota_0
 =\sum_z\pi(z)\left(2^{-n_s}\sum_s e^{i\Phi(z,s)}\right)W_z.
 \label{eq:overview-phase-sum}
\end{equation}
By Eq.~\eqref{eq:phaseaverage}, the parenthesized average is the product of the chosen scalar factors.
Sections~\ref{sec:expansion} and~\ref{sec:local} choose them so that the sum is $\widehat A_Q/2$, for a finite approximation $\widehat A_Q$ to $U$ with degree threshold $Q$.
Thus the complete output on a data state $\ket\psi$ is
\begin{equation}
 B\ket0_{\reg W\reg S}\ket\psi_{\reg D}
 =\ket0_{\reg W\reg S}\frac{\widehat A_Q}{2}\ket\psi_{\reg D}
   +\ket{\perp_\psi},
 \qquad (\bra0_{\reg W\reg S}\otimes I_{\reg D})\ket{\perp_\psi}=0.
 \label{eq:overviewencodedstate}
\end{equation}
The remaining signal outcomes are contained in $\ket{\perp_\psi}$. Thus $B$ is a block encoding of $\widehat A_Q$ with normalization two~\cite{ChildsWiebe,Taylor,QSVT}. One step of oblivious amplitude amplification removes this normalization up to an error controlled in Section~\ref{sec:amplification}~\cite[Lemma~4.6]{BCCKS}.

This describes how the circuit combines a given set of products $W_z$ with chosen weights. To simulate $H_{\rm blk}$, we still need to determine these products, their label probabilities $\pi(z)$, and their coefficient phases $\Phi(z,s)$. The next section constructs them from the Hamiltonian and bounds the approximation error.

\section{Construction and error analysis of the LCU}\label{sec:expansion}
We now specify the LCU encoded by Eq.~\eqref{eq:overview-phase-sum}. Expanding the residual generator in local Pauli operators determines the insertions in each $W_z$, and integrating their time-ordered products determines the coefficients. We then express these coefficients using independent label states and phase factors. The resulting finite operator $\widehat A_Q$ agrees with the exact evolution through degree $Q$; bounding the remaining terms gives the approximation error in Lemma~\ref{lem:tails}.

Throughout this section, $H_{\rm blk}$ acts on $v$ sites with local strength at most one, and $0<\tau\le1$.

\subsection{Local expansion of the product-formula error}\label{sec:residual}
We first expand the block Hamiltonian in Pauli operators and group terms with disjoint supports:
\begin{equation}
 H_{\rm blk}=\sum_a F_a=\sum_{j=1}^c H_{\rm blk,j},
 \qquad F_a:=h_aP_a,\qquad
 H_{\rm blk,j}:=\sum_{a:\,\operatorname{col}(a)=j}F_a.
 \label{eq:color}
\end{equation}
Here $h_a\in\mathbb R$, $P_a$ is a Hermitian Pauli string (a tensor product of $I,X,Y,Z$), and $\operatorname{col}(a)$ is the color of $F_a$.

Since each support intersects only a bounded number of others, greedy coloring uses a constant number $c$ of classes, with disjoint supports within a class.
Moreover, the Pauli expansion changes the local-strength bound by only a fixed factor, since each input matrix has bounded size.
Assigning $F_a$ to the lexicographically least site $x_a$ of its declared support determines which subblock controls its rotation, while the rotation still acts on the full support.

The volume dependence of the correction will be offset by a fixed product-formula order larger than the spatial dimension.
We therefore choose
\begin{equation}
 p:=2\bigl(\lfloor D/2\rfloor+1\bigr)>D.
 \label{eq:fixedorder}
\end{equation}
This is the smallest even order greater than $D$; it does not grow with the target precision.
Starting from the splitting of an exponential into exponentials of its summands~\cite{Trotter1959}, we use the symmetric second-order formula~\cite{Strang} as the base of Suzuki's recursion~\cite{Suzuki}.
In physical stage order, the resulting formula is
\begin{equation}
 C(s):=V_L(s)\cdots V_1(s),\qquad
 V_k(s):=e^{-i\alpha_kH_{\rm blk,\kappa_k}s}.
 \label{eq:PF}
\end{equation}
The number of stages is $L$, and their coefficients $\alpha_k\in\mathbb R$ and color indices $\kappa_k\in\{1,\ldots,c\}$ are fixed by the formula.
For completeness, the recursion that specifies them is
\begin{align}
 S_2(s)&:=e^{-iH_{\rm blk,1}s/2}\cdots e^{-iH_{\rm blk,c}s/2}
          e^{-iH_{\rm blk,c}s/2}\cdots e^{-iH_{\rm blk,1}s/2},\notag\\
 S_{2k}(s)&:=S_{2k-2}(a_ks)^2S_{2k-2}((1-4a_k)s)S_{2k-2}(a_ks)^2,
 \qquad a_k:=(4-4^{1/(2k-1)})^{-1},
 \label{eq:suzuki}
\end{align}
with $2\le k\le p/2$ and $C:=S_p$.
The resulting stage count is at most $2c\,5^{p/2-1}$ before adjacent stages are combined.
Since $p$ and the color count are fixed, both $L$ and all stage coefficients are independent of block volume and precision.

The relevant error estimate is the commutator bound of Ref.~\cite[Theorem~6]{Trotter}, whose local form underlies nearly optimal product-formula gate counts~\cite{ChildsSu}.
Applied to our fixed color decomposition, it gives the following extensive dependence on block volume.
Only the finite-range counting argument is specific to the present setting.
\begin{lemma}[Local product-formula error; {\cite[Theorem~6]{Trotter}}]\label{lem:pferror}
For the local summands in Eq.~\eqref{eq:color}, define
\begin{equation}
 \alpha_{\rm com}:=\sum_{a_1,\ldots,a_{p+1}}
 \norm{[F_{a_{p+1}},\ldots,[F_{a_2},F_{a_1}]\ldots]}.
 \label{eq:commutatorsum}
\end{equation}
There are constants $C_p,C'_p$, depending only on the fixed local model and $p$, such that the formula in Eq.~\eqref{eq:PF} satisfies
\begin{equation}
 \norm{C(s)-e^{-iH_{\rm blk}s}}
 \le C_p\alpha_{\rm com}|s|^{p+1}
 \le C'_p v|s|^{p+1}
 \qquad(s\in\mathbb R).
 \label{eq:pfcommutator}
\end{equation}
For an unnormalized input of local strength at most $J$, the last bound is $C'_p v(J|s|)^{p+1}$. In the normalized model, the $r$-step approximation with $h:=\tau/r$ satisfies
\begin{equation}
 \norm{C(h)^r-e^{-iH_{\rm blk}\tau}}
 \le C'_p v\tau^{p+1}/r^p.
 \label{eq:pfaccumulated}
\end{equation}
\end{lemma}
\begin{proof}
Applying the anti-Hermitian case of the cited theorem to the fixed color decomposition $-iH_{\rm blk}=\sum_{j=1}^c(-iH_{\rm blk,j})$ gives a prefactor independent of $v$.
Multilinearity expands each commutator of color sums into commutators of local terms; the triangle inequality then bounds their sum by $\alpha_{\rm com}$ in Eq.~\eqref{eq:commutatorsum}.
Since the local terms within each color class commute, each color exponential is exactly the product of the local exponentials used in the circuit.
The anti-Hermitian hypothesis also means that all conjugations in the error integrals are unitary, so no exponential factor in $\norm{H_{\rm blk}}$ appears.
Negating the summands gives the same estimate for $s<0$.

It remains to show that $\alpha_{\rm com}$ is extensive.
A nonzero nested commutator starts with one local term, and every subsequent term must meet the union of the preceding supports.
There are $\mathcal O(v)$ choices for the first term but only constantly many choices at each later step, because the order and range are fixed.
Together with bounded term norms and $\norm{[A,B]}\le2\norm A\norm B$, this gives $\alpha_{\rm com}=\mathcal O(v)$, or $\mathcal O(vJ^{p+1})$ before normalization.
Telescoping the $r$ unitary microsteps now yields Eq.~\eqref{eq:pfaccumulated}.
\end{proof}

The norm estimate alone does not specify the coefficients needed for an LCU. We therefore use $C(s)$ as an interaction-picture change of frame~\cite{LowWiebe} and construct a finite classical table $\mathscr R$ with $A$ entries.
For each entry, $x_\lambda$ locates its controls, $P_\lambda$ is a Pauli string near that site, $c_\lambda$ is its coefficient, and $\omega_\lambda$ is its real frequency. The table specifies gates classically.

\begin{lemma}[Local residual table]\label{lem:residual}
For $E(s):=C(s)^\dagger e^{-iH_{\rm blk}s}$, there is a classically computable table $\mathscr R$ of the form described above, with $A=\mathcal O(v)$ entries, whose associated generator satisfies
\begin{equation}
 K(s):=\sum_{\lambda=0}^{A-1}c_\lambda e^{i\omega_\lambda s}P_\lambda,
 \qquad E'(s)=-iK(s)E(s),\qquad E(0)=I.
 \label{eq:Fourier}
\end{equation}
Each site is assigned a bounded number of entries, and $\supp(P_\lambda)$ lies within a fixed radius of its assigned site $x_\lambda$. There are model constants $c_*,\omega_*\ge1$ such that $|c_\lambda|\le c_*$ and $|\omega_\lambda|\le\omega_*$. After its vanishing Taylor coefficients are removed, the generator has the expansion
\begin{equation}
 K(s)=\sum_{\lambda=0}^{A-1}\sum_{q=p}^{\infty}
 c_\lambda\frac{(i\omega_\lambda)^q}{q!}s^qP_\lambda.
 \label{eq:cancelled}
\end{equation}
The table can be padded to a power-of-two address range using entries with $(c_\lambda,\omega_\lambda,P_\lambda)=(0,0,I)$.
\end{lemma}
\begin{proof}
Differentiating $E$ expresses its generator through $C$ and the prefixes $C_{<k}:=V_{k-1}\cdots V_1$, with $C_{<1}:=I$.
Using $C^\dagger C=I$ gives
\begin{equation}
 K(s)=C(s)^\dagger H_{\rm blk}C(s)-iC(s)^\dagger C'(s)
 =C(s)^\dagger H_{\rm blk}C(s)-\sum_{k=1}^L\alpha_k
 C_{<k}(s)^\dagger H_{\rm blk,\kappa_k}C_{<k}(s).
 \label{eq:residualformula}
\end{equation}
Each local term is conjugated through at most $L$ layers. Only rotations overlapping its support can affect it, so its support remains within a constant-radius neighborhood of the assigned site.

The scalar dependence follows from conjugating one Hermitian Pauli string $Q$ by another, $P$.
The conjugation leaves $Q$ unchanged when the strings commute; when they anticommute, it gives
\begin{equation}
 e^{i\vartheta P}Qe^{-i\vartheta P}
 =\cos(2\vartheta)Q+i\sin(2\vartheta)PQ.
 \label{eq:PauliConj}
\end{equation}
Here $iPQ$ is a Hermitian Pauli string up to a real sign, which can be absorbed into the coefficient.
Applying this identity through the bounded set of overlapping rotations gives constantly many sines and cosines per input term. Expanding them into exponentials yields $\mathcal O(v)$ Fourier entries with bounded frequencies and coefficients. The table is independent of the expansion degree, so increasing precision changes only the scalar coefficients and preserves each Pauli support.

The remaining issue is the order at which this scalar expansion begins.
For a fixed block, Lemma~\ref{lem:pferror} gives $E(s)=I+\mathcal O(s^{p+1})$ near zero.
Analyticity then yields $E'(s)=\mathcal O(s^p)$, the usual transfer of product-formula order conditions to the error generator~\cite[Theorem~4]{Trotter}.
Since $K(s)=iE'(s)E(s)^\dagger$, it follows that
\begin{equation}
 \sum_{\lambda=0}^{A-1}c_\lambda(i\omega_\lambda)^qP_\lambda
 =K^{(q)}(0)=0,\qquad 0\le q<p.
 \label{eq:residual-order-conditions}
\end{equation}
These equalities hold for the sum of the table entries; they need not hold entry by entry.
Subtracting those zero sums from the Taylor expansion of Eq.~\eqref{eq:Fourier} gives Eq.~\eqref{eq:cancelled}.
Only after this cancellation do we take absolute values, thereby retaining the factor $|s|^p$ in the coefficient estimate.
\end{proof}
The explicit coefficients in Eq.~\eqref{eq:cancelled} now give the uniform residual estimate
\begin{equation}
 \norm{K(s)}
 \le Ac_*\sum_{q=p}^{\infty}\frac{(\omega_*|s|)^q}{q!}
 \le\frac{Ac_*}{p!}(\omega_*|s|)^p e^{\omega_*|s|}
 =\mathcal O(v|s|^p)\qquad(|s|\le1).
 \label{eq:residualnorm}
\end{equation}
The second inequality uses $(p+j)!\ge p!j!$.
Thus the absolute coefficient bound, not just the operator-norm Trotter estimate, has the extensive volume dependence used in the parameter balance of Eq.~\eqref{eq:balance}.
At the same time, every elementary insertion retains the support of its table entry, regardless of how many scalar powers are used.

\begin{remark}[Keeping the corrections interleaved]\label{rem:correction-order}
The corrections remain between microsteps to preserve locality.
Indeed, writing $C=C(h)$ and $E=E(h)$ and collecting the formula factors on the left gives
\begin{equation}
 (CE)^r=C^r(C^{-(r-1)}EC^{r-1})\cdots(C^{-1}EC)E.
 \label{eq:collectedcorrection}
\end{equation}
The right-hand side conjugates some corrections through as many as $r-1$ formula steps, so its support bound would grow with the number of microsteps.
Keeping the original interleaved order avoids those conjugations and preserves the fixed support radius in Lemma~\ref{lem:residual}.
\end{remark}

\subsection{Time-ordered expansion of the correction}\label{sec:insertions}
We divide the interval into $r\ge1$ microsteps of length $h:=\tau/r$.
Within each microstep, $e^{-iH_{\rm blk}h}=C(h)E(h)$ applies the correction $E(h)$ followed by the product-formula step $C(h)$.
The differential equation in Eq.~\eqref{eq:Fourier} gives the Dyson expansion below.
Whereas truncated Dyson-series simulation represents time ordering by coherent time labels~\cite{Dyson}, the scalar coefficients of our residual table allow the ordered integrals to be evaluated analytically:
\begin{equation}
 E(h)=I+\sum_{k=1}^{\infty}(-i)^k
 \int_{0<s_1<\cdots<s_k<h}
 K(s_k)\cdots K(s_1)\,ds_1\cdots ds_k.
 \label{eq:Dysonintegral}
\end{equation}
Each occurrence of $K$ is an insertion.
Substituting Eq.~\eqref{eq:cancelled} assigns insertion $i$ an address $\lambda_i$ and degree $q_i\ge p$.
The scalar time integral is

\begin{equation}
 \int_{0<s_1<\cdots<s_k<h}\prod_{i=1}^k s_i^{q_i}\,ds_1\cdots ds_k
 =\frac{h^{\sum_i(q_i+1)}}{\prod_{i=1}^k b_i},
 \qquad b_i:=\sum_{a=1}^i(q_a+1).
 \label{eq:simplex}
\end{equation}
Combining this integral with the coefficient of each inserted term yields
\begin{equation}
 E(h)=\sum_{k=0}^{\infty}\ \sum_{\substack{\lambda_1,\ldots,\lambda_k\in\{0,\ldots,A-1\}\\q_1,\ldots,q_k\ge p}}
 \left(\prod_{i=1}^k
 \frac{-ic_{\lambda_i}(i\omega_{\lambda_i})^{q_i}h^{q_i+1}}
 {q_i!\,b_i}\right)
 P_{\lambda_k}\cdots P_{\lambda_1}.
 \label{eq:Dyson}
\end{equation}

All termwise integrations and products are justified by absolute convergence: the Fourier table is finite, and each scalar exponential series converges absolutely on the bounded time interval.

For the complete evolution, a path $\Gamma$ specifies the insertions in each microstep. Microstep $j$ has $m_j\ge0$ insertions, indexed in physical time order by $i$. Insertion $i$ has address $0\le\lambda_{j,i}<A$ and degree $q_{j,i}\ge p$. The insertion count may vary between microsteps.
The denominators reset after each formula step, whereas the total insertion degree adds across all microsteps.
Accordingly, we define
\begin{align}
 b_{j,i}&:=\sum_{a=1}^i(q_{j,a}+1),\label{eq:pathdenominator}\\
 |\Gamma|&:=\sum_{j=1}^r\sum_{i=1}^{m_j}(q_{j,i}+1).\label{eq:degree}
\end{align}
The path coefficient $d_\Gamma$ and ordered data operator $W_\Gamma$ are
\begin{align}
 d_\Gamma&:=\prod_{j=1}^r\prod_{i=1}^{m_j}
 \frac{-ic_{\lambda_{j,i}}(i\omega_{\lambda_{j,i}})^{q_{j,i}}h^{q_{j,i}+1}}
 {q_{j,i}!\,b_{j,i}},\label{eq:pathcoeff}\\
 W_\Gamma&:=[C(h)P_{\lambda_{r,m_r}}\cdots P_{\lambda_{r,1}}]\cdots
 [C(h)P_{\lambda_{1,m_1}}\cdots P_{\lambda_{1,1}}].\label{eq:pathword}
\end{align}

Thus $W_\Gamma$ executes the insertions of microstep $1$, its formula step, then the insertions of microstep $2$, and so on.
The formula steps remain exact product-formula operations: $|\Gamma|$ counts powers from the insertions only, not the Taylor degree of the whole operator.
Substituting the expansion of each $E(h)$ into $[C(h)E(h)]^r$ consequently gives
\begin{equation}
 U:=e^{-iH_{\rm blk}\tau}=\sum_\Gamma d_\Gamma W_\Gamma.
 \label{eq:fullpath}
\end{equation}
Although paths in this exact sum have unbounded length, only finitely many insertions are needed through any fixed degree.
This observation suggests encoding the low-degree part with a fixed number of independently prepared slots, leaving all other assignments to the error analysis.

\subsection{LCU coefficients from independent label states}\label{sec:distribution}
The exact paths contain different numbers of insertions. A circuit needs a fixed number of control registers, so we represent each path by $N$ slots, each choosing an insertion or a product-formula step. Slots after the $r$th formula step are inactive.

For an integer degree threshold $Q\ge p+1$, the required number of slots and the offset range are
\begin{equation}
 \begin{aligned}
 K_0&:=\left\lfloor\frac{Q}{p+1}\right\rfloor,& N&:=r+K_0,\\
 m&:=\lceil\log_2(Q-p)\rceil,& B_0&:=2^m.
 \end{aligned}
 \label{eq:finiteparameters}
\end{equation}
Every insertion contributes at least $p+1$ to the degree, so a path of degree at most $Q$ has at most $K_0$ insertions.
Together with the $r$ formula steps, it therefore fits in $N$ slots.
The $m$-bit offset register has $B_0$ possible values; in the boundary case $Q=p+1$, it is empty and represents $J=0$ because $m=0$ and $B_0=1$.

For each $u\in\{0,\ldots,N-1\}$, the labels range over
\[
 X_u\in\{0,1\},\qquad
 \lambda_u\in\{0,\ldots,A-1\},\qquad
 J_u\in\{0,\ldots,B_0-1\}.
\]
The address selects $(c_{\lambda_u},\omega_{\lambda_u},P_{\lambda_u})$ from the residual table.
Padding that table to a power of two preserves $A=\mathcal O(v)$ and allows uniform address preparation by Hadamard gates.
Since the expansion starts at $q=p$, it suffices to store the offset $J=q-p$.
Using $\ind{\mathcal C}$ for the indicator of a condition $\mathcal C$, we determine the active slots from the prefix values
\begin{equation}
 j_u:=\sum_{a<u}(1-X_a),\qquad e_u:=\ind{j_u<r},\qquad q_u:=p+J_u.
 \label{eq:prefix}
\end{equation}
Thus $j_u$ counts formula steps before slot $u$, and $e_u=1$ precisely when fewer than $r$ have occurred.
We call such a slot \emph{active}.

For example, take two microsteps with two insertions before the first formula step and one before the second. Their active slot types are $1,1,0,1,0$. With zero offsets, the cumulative insertion denominators are $p+1$ and $2(p+1)$ in the first microstep, reset to zero at its formula step, then become $p+1$ in the second. After the second formula step, all remaining slots are inactive.

For a complete assignment $z=(X_u,\lambda_u,J_u)_{u=0}^{N-1}$, these rules give
\begin{equation}
 G_u(z):=
 \begin{cases}
 P_{\lambda_u},&e_uX_u=1,\\
 C(h),&e_u=1,\ X_u=0,\\
 I,&e_u=0,
 \end{cases}
 \qquad W_z:=G_{N-1}(z)\cdots G_0(z).
 \label{eq:recordword}
\end{equation}
The total insertion degree of this assignment is
\begin{equation}
 d(z):=\sum_{u:e_uX_u=1}(p+J_u+1).
 \label{eq:labeldegree}
\end{equation}

We prepare the three labels independently, with distributions
\begin{equation}
 \Pr(X_u=1)=a,\quad \Pr(X_u=0)=1-a,\qquad
 \Pr(J_u=J)=\frac{2^{-J}}{Z_0},\qquad
 \Pr(\lambda_u=\lambda)=\frac1A.
 \label{eq:labelmarginals}
\end{equation}
We choose the probabilities so that $r$ formula steps contribute a common factor $1/2$, while the degree and address states can be prepared bit by bit. \begin{equation}
 a:=1-2^{-1/r},\qquad
 Z_0:=\sum_{J=0}^{B_0-1}2^{-J}.
 \label{eq:productdist}
\end{equation}
The resulting product distribution is
\begin{equation}
 \pi(z):=\prod_{u=0}^{N-1}
 a^{X_u}(1-a)^{1-X_u}\frac1A\frac{2^{-J_u}}{Z_0}.
 \label{eq:pi}
\end{equation}
The binary expansion $J=\sum_{\beta=0}^{m-1}2^\beta J_\beta$, with $J_\beta\in\{0,1\}$, factorizes both the weights and their normalization:
\[
 2^{-J}=\prod_{\beta=0}^{m-1}(2^{-2^\beta})^{J_\beta},
 \qquad Z_0=\prod_{\beta=0}^{m-1}(1+2^{-2^\beta}),
\]
so the square-root probability state for bit $J_\beta$ is
\begin{equation}
 \frac{\ket0+2^{-2^\beta/2}\ket1}{\sqrt{1+2^{-2^\beta}}}.
 \label{eq:Jprep}
\end{equation}
Together with the independent choices of $X$ and $\lambda$, these bit states prepare the square-root state of $\pi$ using only single-qubit gates.

For each insertion, SELECT must multiply its preparation probability by a factor that recovers the Dyson coefficient. Only the cumulative denominator depends on earlier insertions.
We store this denominator modulo $L_b:=2^{\lceil\log_2(Q+1)\rceil}>Q$, using the finite-register recurrence
\begin{equation}
 \widetilde b_{-1}:=0,\qquad
 \widetilde b_u:=
 \begin{cases}
 (\widetilde b_{u-1}+J_u+p+1)\bmod L_b,&X_u=1,\\
 0,&X_u=0.
 \end{cases}
 \label{eq:moddenom}
\end{equation}
For a label assignment with $d(z)\le Q$, no active cumulative denominator reaches the modulus, so $\widetilde b_u$ equals the exact denominator in the current microstep.
Section~\ref{sec:local} computes these values together with $j_u$ and $e_u$ from the original labels.
The degree itself is represented by the offset $J_u$, since $q_u=p+J_u$.

At the source, we apply the reciprocal denominator and factorial ratio:
\begin{align}
 \gamma(b)&:=
 \begin{cases}1/b,&1\le b\le Q,\\1,&\text{otherwise},\end{cases}
 \label{eq:extension}\\
 f(J)&:=\frac{p!}{(p+J)!}.
 \label{eq:factorialfactor}
\end{align}
These factors are bounded on every register value. Below degree $Q$, the denominator has not wrapped and $\gamma$ gives the exact reciprocal. Its extension to other values, including zero, keeps the phase gates defined above the cutoff.
With $t_\omega:=\omega_*h$, the remaining factors are
\begin{align}
 \kappa&:=\frac{AZ_0c_*h\,t_\omega^p}{a\,p!},\qquad
 \chi_\lambda:=\kappa\frac{-ic_\lambda}{c_*}
 \left(\frac{i\omega_\lambda}{\omega_*}\right)^p,
 \nonumber\\
 z_{\lambda,\beta}&:=(2i\omega_\lambda h)^{2^\beta}
 \qquad(0\le\beta<m).
 \label{eq:factors}
\end{align}
In $\chi_\lambda$, the common factor $\kappa$ compensates for the preparation probability, while the other factors supply $c_\lambda$ and the first $p$ frequency powers.
The remaining $J$ powers can be supplied bit by bit, since
\[
 \prod_{\beta=0}^{m-1}z_{\lambda,\beta}^{J_\beta}
 =(2i\omega_\lambda h)^J.
\]
When $J_\beta=0$, that bit contributes the identity, including when $\omega_\lambda=0$.
Combining the source and table-dependent factors gives the complete coefficient for slot $u$:
\begin{equation}
 \zeta_u(z):=
 \begin{cases}
 f(J_u)\gamma(\widetilde b_u)\chi_{\lambda_u}
 \displaystyle\prod_{\beta=0}^{m-1}z_{\lambda_u,\beta}^{J_{u,\beta}},&e_uX_u=1,\\
 1,&e_uX_u=0.
 \end{cases}
 \label{eq:zeta}
\end{equation}
The finite operator to be implemented is
\begin{equation}
 \widehat A_Q:=2\sum_z\pi(z)\left(\prod_{u=0}^{N-1}\zeta_u(z)\right)W_z.
 \label{eq:implemented}
\end{equation}
The sum also includes assignments that end before all $r$ formula steps are completed. Lemma~\ref{lem:matching} shows that these have degree greater than $Q$, so they can be bounded with the other high-degree terms.

Table~\ref{tab:slotconstruction} records which values are prepared, which are computed, and where their coefficient factors are applied.
Section~\ref{sec:local} realizes each factor by the phase-average identity: source factors remain at their slots, whereas the $\lambda_u$-dependent factors are applied beside the selected interaction.
The distribution circuit of Section~\ref{sec:geometry} supplies the necessary controls in slot order and removes their copies afterward.
\begin{table}[t]
\centering
\small
\renewcommand{\arraystretch}{1.18}
\begin{tabular}{@{}>{\raggedright\arraybackslash}p{.18\textwidth}>{\raggedright\arraybackslash}p{.35\textwidth}>{\raggedright\arraybackslash}p{.40\textwidth}@{}}
\toprule
Quantity in slot $u$ & Meaning or distribution & Preparation or use \\
\midrule
$X_u$ & Insertion/formula step label; probabilities $a,1-a$ & One-qubit preparation in Eq.~\eqref{eq:prepared}. \\
$\lambda_u$ & Address of $(c_\lambda,\omega_\lambda,P_\lambda)$; uniform on $A$ entries & Hadamard gates on the address register; local decoding in Lemma~\ref{lem:local}. \\
$J_u$ & Offset $q_u-p$; probability $2^{-J_u}/Z_0$ & Independent bit states in Eq.~\eqref{eq:Jprep}. \\
$j_u,e_u,\widetilde b_u$ & Formula-step count, execution bit, cumulative denominator modulo $L_b$ & Reversible prefix computation in Section~\ref{sec:source}; arithmetic in Appendix~\ref{app:gates}. \\
$f(J_u),\gamma(\widetilde b_u)$ & Factorial and reciprocal-denominator factors & Phase gates on source phase qubits $\reg S_u^F,\reg S_u^b$; Eq.~\eqref{eq:phaseaverage}. \\
$\chi_{\lambda_u},\ (z_{\lambda_u,\beta})_\beta$ & Table coefficient and binary frequency factors & Phase gates on $\reg S_u^\chi,\reg S_{u,\beta}^\omega$ near the selected entry; Lemma~\ref{lem:local}. \\
$G_u(z)$ & Pauli insertion, formula step $C(h)$, or identity & Controls from Section~\ref{sec:source} are distributed by Section~\ref{sec:geometry} and processed by Lemma~\ref{lem:local}. \\
\bottomrule
\end{tabular}
\caption{Slot labels, derived values, and the operations they control. Section~\ref{sec:prepare} specifies their quantum registers and the independent phase qubits; phase averages realize the numerical coefficient factors without changing the label distribution.}
\label{tab:slotconstruction}
\end{table}

\Needspace{15\baselineskip}
\begin{lemma}[Low-degree matching]\label{lem:matching}
The terms in Eq.~\eqref{eq:implemented} with insertion degree $d(z)\le Q$ sum to the truncated path expansion
\begin{equation}
 A_Q:=\sum_{|\Gamma|\le Q}d_\Gamma W_\Gamma.
 \label{eq:hardcutoff}
\end{equation}
\end{lemma}
\begin{proof}
First, low degree forces a label assignment to complete every formula step.
Indeed, fewer than $r$ zeros in an $N$-bit operation string leave at least $N-r+1=K_0+1$ ones, all active, giving degree at least $(K_0+1)(p+1)>Q$.
Conversely, any path with $|\Gamma|\le Q$ has at most $K_0$ insertions and fits in $N=r+K_0$ slots.
For each insertion, $p\le q\le Q-1$ implies $J=q-p\le Q-p-1<B_0$, so the offset is representable.
Moreover, every $b_{j,i}\le Q<L_b$, and the modular recurrence therefore gives the exact path denominators through the last active formula step.

It remains to check the coefficient of the represented path.
For an active insertion with labels $(\lambda,J)$, binary digits $J_\beta$, and denominator $b$, let $\Pr(1,\lambda,J)$ denote its marginal preparation probability.
Multiplying this probability by its selected factor gives
\begin{align}
 &\Pr(1,\lambda,J)\,
 f(J)\gamma(b)\chi_\lambda
 \prod_{\beta=0}^{m-1}z_{\lambda,\beta}^{J_\beta}
 \nonumber\\
 &\quad=\frac{a2^{-J}}{AZ_0}
 \frac{p!}{(p+J)!}\frac1b
 \frac{AZ_0c_*h(\omega_*h)^p}{a\,p!}
 \frac{-ic_\lambda}{c_*}
 \left(\frac{i\omega_\lambda}{\omega_*}\right)^p
 (2i\omega_\lambda h)^J
 \nonumber\\
 &\quad=\frac{-ic_\lambda(i\omega_\lambda)^{p+J}h^{p+J+1}}
 {(p+J)!\,b}.
 \label{eq:matchingidentity}
\end{align}
The cancellation leaves exactly the insertion coefficient in Eq.~\eqref{eq:pathcoeff}, with $q=p+J$.
The product over all active insertions is therefore $d_\Gamma$.
The $r$ formula-step probabilities contribute $(1-a)^r=1/2$, which cancels the leading $2$ in Eq.~\eqref{eq:implemented}; their unused fields, and all fields after the last active formula step, sum to probability one.
Finally, the active operations occur in the order defining $W_\Gamma$.
After summing unused labels, each low-degree path thus contributes precisely $d_\Gamma W_\Gamma$.
\end{proof}

\subsection{Truncation error and coefficient bounds}\label{sec:tails}
Lemma~\ref{lem:matching} identifies a common low-degree sum $A_Q$ in the exact evolution and the finite LCU. Their difference therefore consists of two tails:
\[
 \norm{\widehat A_Q-U}
 \le \norm{U-A_Q}+\norm{\widehat A_Q-A_Q}.
\]
Unitarity bounds each tail by its absolute coefficient sum~\cite{Taylor,Dyson}. We control both sums by weighting a degree-$d$ term by $\xi^d$, with $\xi>1$.

\begin{lemma}[Approximation error and coefficient bounds]\label{lem:tails}
For $\xi>1$ with $t_\omega\xi<1$, define
\begin{equation}
 \calA(\xi):=\frac{c_*Ah\,t_\omega^p\xi^{p+1}}
 {p!\,(1-t_\omega\xi)}.
 \label{eq:Acal}
\end{equation}
If $r\calA(\xi)\le1/8$, the approximation error satisfies
\begin{equation}
 \norm{\widehat A_Q-U}\le(e^{1/8}+2)\xi^{-Q}.
 \label{eq:twotails}
\end{equation}
If also $t_\omega\xi\le1/2$, the coefficient factors satisfy
\begin{equation}
 |f(J)|\le1,\qquad |\gamma(b)|\le1,\qquad
 \xi^{p+1}|\chi_\lambda|<\tfrac12,\qquad
 \xi^{2^\beta}|z_{\lambda,\beta}|\le1.
 \label{eq:weightedcontractions}
\end{equation}
In particular, all unweighted factors used by the circuit have modulus at most one.
\end{lemma}
\begin{proof}
For nonnegative coefficients $w_\alpha$ of degrees $d_\alpha$, we use
\begin{equation}
 \sum_{d_\alpha>Q}w_\alpha
 \le\xi^{-Q}\sum_\alpha w_\alpha\xi^{d_\alpha}.
 \label{eq:weightedtailprinciple}
\end{equation}
Thus a bound on the full weighted sum immediately bounds the part above degree $Q$.
We will obtain bounds $e^{1/8}$ for the exact expansion and $2$ for the finite representation.

\paragraph{The exact expansion.}
In the exact expansion, omitting the cumulative denominator allows us to bound the scalar sum for one insertion by a geometric series:
\begin{align}
 \sum_{\lambda=0}^{A-1}\sum_{q=p}^{\infty}
 \frac{|c_\lambda|\,|\omega_\lambda|^qh^{q+1}\xi^{q+1}}{q!}
 &\le c_*Ah\sum_{q=p}^{\infty}\frac{t_\omega^q\xi^{q+1}}{q!}\notag\\
 &\le\frac{c_*Ah\,t_\omega^p\xi^{p+1}}{p!}
       \sum_{j=0}^{\infty}(t_\omega\xi)^j
 =\calA(\xi).
 \label{eq:insertionmass}
\end{align}
The last inequality uses $(p+j)!\ge p!$.
Restoring the cumulative denominators improves the bound, because $b_i\ge i$ for $k$ ordered insertions supplies a factor at most $1/k!$.
Consequently, their weighted sum is at most $\calA(\xi)^k/k!$, and summing over the insertion counts in all $r$ microsteps yields
\begin{equation}
 \sum_\Gamma|d_\Gamma|\xi^{|\Gamma|}
 \le\left(\sum_{k=0}^{\infty}\frac{\calA(\xi)^k}{k!}\right)^r
 =e^{r\calA(\xi)}\le e^{1/8}.
 \label{eq:truemass}
\end{equation}
Equation~\eqref{eq:weightedtailprinciple} and unitarity of $W_\Gamma$ therefore bound the exact tail by $e^{1/8}\xi^{-Q}$.

\paragraph{The finite representation.}
For the finite sum, we instead condition on the labels before an active slot.
Regardless of that prefix, its insertion choices have weighted absolute sum at most
\begin{align}
 &\sum_{\lambda=0}^{A-1}\sum_{J=0}^{B_0-1}
 \frac{a2^{-J}}{AZ_0}
 \left|f(J)\gamma(\widetilde b)\chi_\lambda
             \prod_\beta z_{\lambda,\beta}^{J_\beta}\right|
 \xi^{p+J+1}\notag\\
 &\hspace{8mm}\le
 \frac{a\kappa\xi^{p+1}}{Z_0}
 \sum_{J=0}^{B_0-1}f(J)(t_\omega\xi)^J
 =: \eta(\xi)\le\calA(\xi).
 \label{eq:eta}
\end{align}
The first inequality uses $|\gamma|\le1$, $|\chi_\lambda|\le\kappa$, and
$\prod_\beta|z_{\lambda,\beta}|^{J_\beta}\le(2t_\omega)^J$.
Substituting $\kappa$, using $f(J)\le1$, and extending the geometric sum gives $\eta(\xi)\le\calA(\xi)$.
The resulting bound is independent of the preceding labels and their modular denominator.
It can therefore be iterated uniformly over the remaining slots.

A scalar recurrence gives a bound $M_j(n)$ that is uniform over all prefixes with $j$ formula steps still required and $n$ slots remaining.
Once $j=0$, all later fields are inactive and sum to probability one.
When $n=0$, the contribution is also one even if $j>0$, because incomplete assignments are retained in the implemented sum.
These two boundary conditions, together with the uniform insertion bound, give
\begin{equation}
 \begin{aligned}
 M_0(n)&:=1,\qquad M_j(0):=1,\\
 M_j(n)&:=(1-a)M_{j-1}(n-1)+\eta(\xi)M_j(n-1)
 \quad(j,n\ge1).
 \end{aligned}
 \label{eq:massrec}
\end{equation}
Bernoulli's inequality gives $1/2=(1-a)^r\ge1-ra$, hence $ra\ge1/2$.
Since $\eta(\xi)\le1/(8r)\le a/4$, the two recurrence coefficients sum to at most $1-a+a/4\le1$.
Starting from boundary values equal to one, induction therefore gives $M_j(n)\le1$.
Restoring the leading factor of two in $\widehat A_Q$, we obtain
\begin{equation}
 2\sum_z\pi(z)\left|\prod_u\zeta_u(z)\right|\xi^{d(z)}
 \le2M_r(N)\le2.
 \label{eq:implementedmass}
\end{equation}
Every $W_z$ is a product of unitary gates, including for assignments that end before the $r$th formula step.
Consequently, the weighted coefficient bound also covers these incomplete assignments and gives $\norm{\widehat A_Q-A_Q}\le2\xi^{-Q}$.
Adding the exact-tail estimate proves Eq.~\eqref{eq:twotails}.

\paragraph{The circuit factors.}
The same parameter condition also ensures that the selected factors can be implemented by phase averages.
Indeed, $Z_0<2$ and $ra\ge1/2$ imply
\begin{equation}
 \xi^{p+1}\kappa
 =\frac{\calA(\xi)Z_0(1-t_\omega\xi)}a
 \le\frac{Z_0(1-t_\omega\xi)}{8ra}
 \le\frac{Z_0}4<\frac12.
 \label{eq:kappacontraction}
\end{equation}
This proves $\xi^{p+1}|\chi_\lambda|<1/2$. The additional hypothesis $t_\omega\xi\le1/2$ gives
$\xi^{2^\beta}|z_{\lambda,\beta}|=(2|\omega_\lambda|h\xi)^{2^\beta}\le1$,
while the bounds on $f$ and $\gamma$ follow directly from their definitions.
Since $\xi>1$, all unweighted factors are contractions as well.
\end{proof}

\FloatBarrier
\section{Nearest-neighbor implementation of the LCU}\label{sec:local}
This section implements the LCU constructed in Section~\ref{sec:expansion}, completing the circuit construction for Theorem~\ref{thm:block}. The main challenge from the proof outline is to distribute the controls across a block in depth $\mathcal O(M)$ while keeping the insertions and product-formula steps in order. We first establish the distribution bound, then give the preparation and selected operations.

In Figure~\ref{fig:lcu-overview}(a), $\calT$ computes the prefix counts and denominators in $\reg W$, $\mathcal D_0$ applies the coefficient phases that depend on these values, and $\mathcal N$ distributes the controls and applies the local phases and data operations. The SELECT circuit is
\begin{equation}
 \SELECT:=\calT^\dagger\mathcal N\mathcal D_0\calT.
 \label{eq:selectdecomposition}
\end{equation}
Table~\ref{tab:lcu-interface} summarizes these components.

\begin{table}[t]
\centering\small
\begin{tabular}{p{.14\linewidth}p{.52\linewidth}p{.24\linewidth}}
\toprule
Circuit & Operation & Construction\\
\midrule
$\PREPARE$ & Prepares the original labels and independent $\ket+$ phase qubits in the source region. & Section~\ref{sec:prepare}\\
$\calT$ & Computes execution bits and cumulative denominators, retaining the values needed to reverse the calculation. & Section~\ref{sec:source}\\
$\mathcal D_0$ & Applies factorial and denominator coefficient phases at the original slots. & Section~\ref{sec:source}\\
$\mathcal N$ & Distributes the local control fields, applies address-dependent phases and ordered data gates, and returns and clears the copies. & Sections~\ref{sec:geometry} and~\ref{sec:local-use}\\
$\calT^\dagger$ & Clears the retained prefix outputs after all controlled uses. & Section~\ref{sec:source}\\
$\PREPARE^\dagger$ & Interferes the signal branches to encode their weighted sum. & Eq.~\eqref{eq:blockencoding}\\
\bottomrule
\end{tabular}
\caption{The block circuit and its temporary work. The components of $\SELECT$ preserve the original signal basis values, which allows their temporary work to be cleared after the selected data operations and phases have been applied.}
\label{tab:lcu-interface}
\end{table}

The coefficient phases require
\begin{equation}
 \kappa\le1,\qquad 2t_\omega\le1.
 \label{eq:localcoefficientconditions}
\end{equation}
The parameter choice in Section~\ref{sec:error} will satisfy these conditions, as well as the stronger bounds needed for the approximation error.

\subsection{Distributing the control registers}\label{sec:layout-overview}
We divide the data into subblocks $\Omega_g$ of $\Theta(w)$ sites, which share a width-$w$ control register $\reg R_g$ and local work $\reg W_g$. The address is decoded within its subblock, placing a selector beside the chosen interaction. The original records stay on an auxiliary source path of length $\mathcal O(Nw+\ell)$ and are retained for uncomputation.

For $D\ge2$, a tree carries successive records concurrently to the subblocks, with separate return paths for uncomputation. The depth is one traversal time plus the processing time of the records. On a line, records reach different subblocks at different times; we space their operations to preserve the order on overlapping supports.

Only interaction-dependent coefficient factors need these distributed controls. Degree- and prefix-dependent factors are applied at the source.

\subsection{Preserving the order of selected operations}\label{sec:geometry}

Local alignment uses a prescribed permutation on a path~\cite{CSU,Routing}. Its SWAP schedule is computed classically and is independent of the quantum address.

\begin{lemma}[Fixed path permutations; {\cite[Section~3.1]{CSU}}]\label{lem:path-permutation}
A prescribed permutation of $k$ qubits on a nearest-neighbor path has an exact implementation of depth $\mathcal O(k)$ using $\mathcal O(k^2)$ adjacent SWAP gates and no auxiliary qubits.
\end{lemma}
\begin{proof}
Assigning each input its destination gives a classical list of $k$ labels.
Odd--even transposition sorting orders this list in at most $k$ layers of disjoint adjacent comparisons, which is routing by matchings on a path~\cite{Routing}.
Because the list is known, every comparison outcome can be computed in advance.
Replacing the exchanging comparisons by SWAPs therefore gives a fixed circuit implementing the desired permutation on arbitrary quantum states.
There are at most $k/2$ swaps per layer, and three adjacent CNOTs implement each swap, giving the stated bounds.
\end{proof}

A subblock must accommodate both a word and the work used to align and decode it.
For $D\ge2$, we therefore choose the smallest power of two $s$ satisfying $s^D\ge4w+8$ and set its side to $b:=2s$.
This gives $\Theta(w)$ sites.
A nearest-neighbor snake orders them by traversing successive first-coordinate lines in alternating directions and applying the same rule recursively in the other coordinates.

The local action is divided into a fixed number of stages $k=0,\ldots,L$, because operations belonging to different stages need not commute.
For a basis record $y$, the data unitary in subblock $g$ is denoted by $D_{g,k}(y)$ and its scalar phase by $\phi_{g,k}(y)$.
The stage circuit acts as
\begin{equation}
 \mathcal L_{g,k}\ket y_{\reg R_g}\ket0_{\reg W_g}\ket\psi_{\reg D}
 =e^{i\phi_{g,k}(y)}
 \ket y_{\reg R_g}\ket0_{\reg W_g}D_{g,k}(y)\ket\psi_{\reg D}.
 \label{eq:distribution-local-interface}
\end{equation}
Its cost is depth $\mathcal O(w)$, $\mathcal O(w^2)$ gates, and $\mathcal O(w)$ work. Assigned data gates lie within a fixed distance of $\Omega_g$, so a constant coloring resolves neighboring routing conflicts. For each record and stage, we require the operators $D_{g,k}(y)$ to commute across groups. Completing one stage everywhere before the next then preserves the order of noncommuting operations.

Under this condition, one record applies the data operator and phase
\begin{equation}
 \mathscr D(y):=\left(\prod_gD_{g,L}(y)\right)\cdots
                    \left(\prod_gD_{g,0}(y)\right),\qquad
 \phi(y):=\sum_{g,k}\phi_{g,k}(y).
 \label{eq:distribution-ordered-stages}
\end{equation}
In the simulator, stage zero will be the selected insertion, while stages $1,\ldots,L$ will implement the product formula.
Section~\ref{sec:local-use} verifies the required commutation and cleanup properties for these operations.
Once they hold, the following lemma supplies the global ordered circuit.

\begin{lemma}[Ordered geometric distribution]\label{lem:geometry}
Fix $D\ge2$ and the local circuit assumptions in Eq.~\eqref{eq:distribution-local-interface}. The all-zero record acts trivially, so $D_{g,k}(0)=I$ and $\phi_{g,k}(0)=0$. Suppose a cube has side $\ell$, $N\ge1$ records have width at most the power of two $w$, and
\[
 w\ge\log_2\ell,\qquad Nw+w\log\ell=\mathcal O(\ell),\qquad b\le\ell.
\]
There is a nearest-neighbor circuit in a constant-factor enlargement of the cube that applies, for every sequence of basis records $y_0,\ldots,y_{N-1}$,
\begin{equation}
 e^{i\sum_{u=0}^{N-1}\phi(y_u)}
 \mathscr D(y_{N-1})\cdots\mathscr D(y_0)
 \label{eq:distributed-product}
\end{equation}
to the data, with resources
\begin{equation}
 \text{depth }\mathcal O(\ell),\qquad
 \text{gates }\mathcal O(\ell^{D+1}),\qquad
 \text{ancillas }\mathcal O(\ell^D),
 \label{eq:compilerresources}
\end{equation}

\end{lemma}

Each period performs a local use and a word shift in depth $\mathcal O(w)$. A level-$k$ connection takes $\mathcal O(b2^k/w+1)$ periods. Equal delays within a level synchronize arrivals at all leaves; after the initial traversal, one new record is processed per period. Since $b2^K=\Theta(\ell)$, the depth is
\begin{equation}
 \dcnt=O\!\left(w\sum_{k=1}^{K}(b2^k/w+1)+Nw\right)
       =\mathcal O(\ell+w\log\ell+Nw).
 \label{eq:overviewdistribution}
\end{equation}
The construction is given in Appendix~\ref{app:geometry}, with the one-dimensional schedule in Appendix~\ref{app:1d}.

\subsection{Preparing the labels and phase qubits}\label{sec:prepare}
The factorization in Section~\ref{sec:distribution} allows all labels to be prepared in one layer of single-qubit gates. General conditional-probability or uniformly controlled preparation~\cite{GroverRudolph,Mottonen} is unnecessary here.

For each slot $u$, the label registers are $\reg X_u$ for the operation type, $\reg\Lambda_u$ for the table address, and $\reg J_u$ for the degree offset. They use one, $\log_2 A$, and $m$ qubits, respectively.

The coefficient $\zeta_u$ is implemented with separate phase qubits for its factors. We use $\reg S_u^F$ for the factorial ratio and $\reg S_u^b$ for the reciprocal denominator. The table coefficient uses $\reg S_u^\chi$, while each offset bit $J_{u,\beta}$ has a phase qubit $\reg S_{u,\beta}^\omega$ for its frequency factor. These $m+3$ phase qubits are prepared in $\ket+$ and used only when $e_uX_u=1$. All label and phase qubits together form $\reg S$.

A rotation $R_y(2\arcsin\sqrt a)$ prepares $X_u$, where $R_y(\theta):=e^{-i\theta Y/2}$; Hadamards prepare the address and phase qubits, and the single-qubit rotations of Eq.~\eqref{eq:Jprep} prepare the offset.
Their tensor product consequently satisfies
\begin{equation}
 \PREPARE\ket0_{\reg S}
 =\sum_{z,s}\sqrt{\frac{\pi(z)}{2^{n_s}}}\ket{z,s}_{\reg S},
 \qquad n_s:=N(m+3),
 \label{eq:prepared}
\end{equation}
where $s$ is the string of phase-qubit bits. Thus the preparation of all original labels and phase qubits has
\begin{equation}
 \begin{aligned}
 \text{depth }\dcnt_{\rm prep}&=\mathcal O(1),\\
 \text{gate count }\gcnt_{\rm prep}&=\mathcal O(Nw),\\
 \text{signal ancillas }\acnt_{\rm prep}&=\mathcal O(Nw).
 \end{aligned}
 \label{eq:preparationresources}
\end{equation}

\subsection{Computing the coefficients at the source}\label{sec:source}
The source computes which slots are active and the cumulative denominators. These values control the factorial and reciprocal factors; interaction-dependent phases are applied near the selected term. This keeps the transported register logarithmic in width.

Consecutive slots occupy $\mathcal O(w)$ positions each on the auxiliary path.
Besides the signal fields, slot $u$ has a denominator register $\reg B_u$, an execution bit $\reg E_u$, and $\mathcal O(w)$ arithmetic work qubits; count registers $\reg Z_u$ are allocated for $0\le u\le N$.

The denominator requires $\log_2L_b$ bits, while $\lceil\log_2(N+1)\rceil$ count bits represent every value from zero to $N$.

At slot $u$, $\reg Z_u$ contains the number $j_u$ of earlier formula steps. A comparison writes $e_u=\ind{j_u<r}$ into $\reg E_u$. Controlled on $X_u$, the circuit copies the preceding denominator into $\reg B_u$ and adds $J_u+p+1$ modulo $L_b$; at $u=0$ there is no copy. It then increments the count by $1-X_u$ and swaps it into the zero register $\reg Z_{u+1}$. These operations realize Eq.~\eqref{eq:moddenom} while retaining the inputs needed for reversal.

Arithmetic and alignment cost depth $\mathcal O(w)$ and $\mathcal O(w^2)$ gates per slot, giving
\begin{equation}
 \begin{aligned}
 \text{depth }\dcnt_{\rm pre}&=\mathcal O(Nw),\\
 \text{gate count }\gcnt_{\rm pre}&=\mathcal O(Nw^2),\\
 \text{work ancillas }\acnt_{\rm pre}&=\mathcal O(Nw).
 \end{aligned}
 \label{eq:prefixresources}
\end{equation}
In reverse order, the swap and decrement recover $j_u$ before the comparison clears $\reg E_u$; the retained denominators similarly clear each update.

\paragraph{Scalar coefficients.}
Separate phase qubits encode $f(J_u)$ and $\gamma(\widetilde b_u)$ through Eq.~\eqref{eq:phaseaverage}.

For each classical value $j_0\in\{0,\ldots,B_0-1\}$, a reversible comparison computes $e_uX_u\ind{J_u=j_0}$, which controls $\Dph(f(j_0))$ on $\reg S_u^F$; the comparison is then reversed.
Enumerating $b_0\in\{1,\ldots,Q\}$ similarly applies $\Dph(\gamma(b_0))$ on $\reg S_u^b$ under $e_uX_u\ind{\widetilde b_u=b_0}$.
Other denominator encodings have $\gamma=1$ and need no gate.

Each comparison and its inverse cost $\mathcal O(w)$ depth and gates.
Since $B_0=\mathcal O(Q)$ and different slots run in parallel, the two source tables together have
\begin{equation}
 \begin{aligned}
 \text{depth }\dcnt_{0}&=\mathcal O(Qw),\\
 \text{gate count }\gcnt_{0}&=\mathcal O(NQw),\\
 \text{work ancillas }\acnt_{0}&=\mathcal O(Nw).
 \end{aligned}
 \label{eq:sourcecoefficientresources}
\end{equation}
We call this diagonal circuit $\mathcal D_0$. Comparisons are reversed after use, so the workspace is reusable. The remaining factors $\chi_{\lambda_u}$ and $z_{\lambda_u,\beta}$ are applied near their table entry by the local circuit below.

Only the fields used by the local gates leave the source. For slot $u$, the transported register $\reg O_u$ contains the execution bit and operation bit, the address and offset, and the phase qubits for the table coefficient and frequency factors. It is padded to width $w$. The factorial and denominator phase qubits remain at the source with the prefix values until the computation is reversed.

\subsection{Implementing the selected local operations}\label{sec:local-use}
Each subblock has $v_0=\Theta(w)$ sites, with $v_0=b^D$ for $D\ge2$ and $v_0=w$ on a line. The padding in Appendices~\ref{app:geometry} and~\ref{app:1d} makes the number of subblocks a power of two.
A fixed binary address format is obtained by assigning every site the same power-of-two number $q_{\rm rec}$ of residual entries and padding missing entries with zero coefficients.
Each entry gives classical data $(c_\lambda,\omega_\lambda,P_\lambda)$ for its local gates.
In the snake ordering of Section~\ref{sec:geometry}, the number of entries per subblock is $T_0:=q_{\rm rec}v_0$, which is also a power of two.
The address therefore separates into subblock, site, and entry fields:
\begin{equation}
 \lambda=gT_0+i,
 \qquad T_0=q_{\rm rec}v_0,
 \qquad i:=q_{\rm rec}j+\mu.
 \label{eq:address}
\end{equation}
Here $g$ identifies the subblock, $j\in\{0,\ldots,v_0-1\}$ identifies a site along its snake, and $\mu\in\{0,\ldots,q_{\rm rec}-1\}$ identifies the entry at that site.
Equivalently, $i=(j,\mu)$ is the local address and $\lambda=(g,i)$ separates global and local decoding.
A padded entry retains a definite assigned site $(g,j)$, so the same decoder works on the entire binary address range.

At site $j$ in subblock $g$, $\reg D_{g,j}$ holds the data. The arriving control occupies the shared register $\reg R_g$. A selector $\reg a_{g,i}$ lies near each table entry, with the selectors and further work forming $\reg W_g$. Their distribution along the subblock uses constant qubit density. Figure~\ref{fig:layout} shows the registers; Appendix~\ref{app:geometry} gives their physical placement.

\begin{figure}[tbp]
\centering
\includegraphics[width=.90\linewidth]{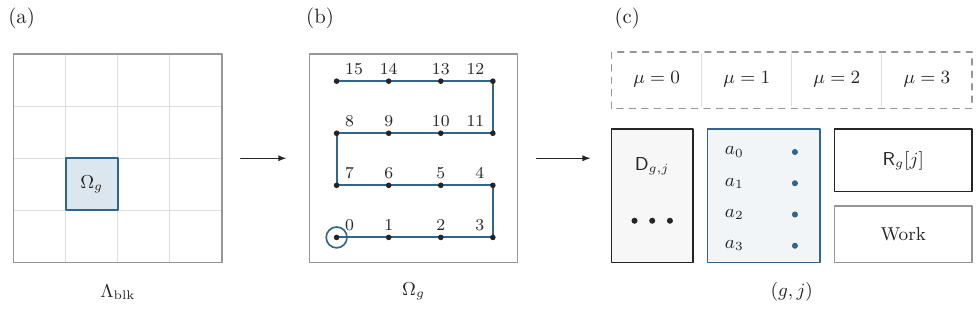}
\caption{Local address decoding. Each subblock contains $v_0=\Theta(w)$ data sites, ordered along a nearest-neighbor snake. The example has $q_{\rm rec}=4$ classical table entries per site. The shared control register $\mathsf R_g$ is distributed along the snake, with selectors beside their assigned entries and work in $\mathsf W_g$. The address separates as $\lambda=gT_0+i$, where $i=q_{\rm rec}j+\mu$.}
\label{fig:layout}
\end{figure}

For a local call, let $e$ denote the execution bit and $X$ the operation bit. Write the selected address as $(g_*,i_*)$, identifying its subblock and local entry. The remaining fields are the degree offset $J$ and its coefficient phase qubits, defined in Section~\ref{sec:source}.
The local action separates into the insertion and the individual product-formula stages:
\begin{equation}
 \begin{aligned}
 D_{g,0}(y)&:=
 \begin{cases}P_{gT_0+i_*},&eX=1\text{ and }g_*=g,\\ I,&\text{otherwise},\end{cases}\\
 D_{g,k}(y)&:=
 \begin{cases}
 \displaystyle\prod_{\substack{a:\,x_a\in\Omega_g\\\operatorname{col}(a)=\kappa_k}}
 e^{-i\alpha_kh h_aP_a},&e(1-X)=1,\\
 I,&\text{otherwise},
 \end{cases}\qquad 1\le k\le L.
 \end{aligned}
 \label{eq:local-stage-operators}
\end{equation}
The assignment of $x_a$ was fixed after Eq.~\eqref{eq:color}.
Within $D_{g,k}$, the terms have disjoint declared supports and hence commute, although a term may act across the boundary of $\Omega_g$.
These operators are therefore understood on the full data space.
The data operator of a complete local invocation is
\[
 W_g(y):=D_{g,L}(y)\cdots D_{g,0}(y).
\]
The phase $\phi_g(y)$ is the sum of the branch phases of the $\Dph$ gates implementing $\chi_\lambda$ and the factors $z_{\lambda,\beta}$ with $J_\beta=1$ in Eq.~\eqref{eq:factors}, at the selected entry in $\Omega_g$.
It vanishes unless the insertion is active and its address selects this subblock, namely $eX=1$ and $g_*=g$.

With the same record at every subblock, the staged product $\mathscr D(y)$ is the selected insertion, $C(h)$, or $I$, according to its branch.
All coefficient phases occur in the insertion stage, so Eq.~\eqref{eq:distribution-local-interface} uses $\phi_{g,0}(y):=\phi_g(y)$ and $\phi_{g,k}(y):=0$ for $k\ge1$.
For an insertion, only the addressed group can act nontrivially on a fixed basis branch; for a formula stage, all groups act on terms of the same disjoint-support color.
Thus operators within a stage commute across groups, which is the hypothesis required by Lemma~\ref{lem:geometry}.
The staged product still retains the order between different, potentially noncommuting stages.

\FloatBarrier
\Needspace{17\baselineskip}
\begin{lemma}[Selection within a subblock]\label{lem:local}
Suppose $\reg R_g$ contains the fields defined in Section~\ref{sec:source}, $T_0=\Theta(w)$, and $\reg W_g$ is zero. The address-dependent phases and data operations assigned to $\Omega_g$ have a nearest-neighbor implementation with
\begin{equation}
 \begin{aligned}
 \text{depth }\dcnt_{\rm loc}&=\mathcal O(w),\\
 \text{gate count }\gcnt_{\rm loc}&=\mathcal O(w^2),\\
 \text{ancillas }\acnt_{\rm loc}&=\mathcal O(w).
 \end{aligned}
 \label{eq:localresources}
\end{equation}
The bounds include the insertion and all $L$ product-formula stages, in the order of Eq.~\eqref{eq:distribution-ordered-stages}. Coefficient phases occur only in the insertion stage.
\end{lemma}
\begin{proof}
For an insertion, we decode the address, apply the phase and Pauli gates, then reverse the decoder~\cite{Bennett}. For a product-formula stage, a distributed branch bit controls the assigned rotations. We count these operations separately.

\paragraph{Address decoding.}
The binary address is useful for transport, but the data operation needs a control beside its assigned site.
To convert between these representations, we first compute
\begin{equation}
 h_g:=eX\ind{g_*=g}
 \label{eq:blockhit}
\end{equation}
into a zero condition bit and copy it to $\reg a_{g,0}$.
The bits in the binary expansion $i_*=\sum_{b=0}^{\log_2T_0-1}2^bi_b$, processed in increasing order of $b$, determine the successive positions of this selector.
At level $b$, the decoder applies
\begin{equation}
 \prod_{j=0}^{2^b-1}\CSWAP(i_b;\reg a_{g,j},\reg a_{g,j+2^b}).
 \label{eq:decoder}
\end{equation}
Initially, the only possibly nonzero selector is $a_{g,0}=h_g$.
After processing bits through $b$, this value is at $i_*\bmod2^{b+1}$.
Indeed, the next controlled swap either leaves it there or adds $2^{b+1}$, according to $i_{b+1}$, while all other selectors remain zero.
Induction therefore gives
\begin{equation}
 a_{g,i}=eX\ind{g_*=g}\ind{i_*=i}.
 \label{eq:onehot}
\end{equation}

The controlled swaps in Eq.~\eqref{eq:decoder} need not initially join adjacent selectors.
A fixed interleaving of the two halves makes their targets adjacent, so Lemma~\ref{lem:path-permutation} supplies the required routing without consulting the quantum address.

We first align the $\mathcal O(\log w)$ low address bits near the start of the selector path. At level $b$, we interleave the first $2^{b+1}$ selectors in the order
\[
 a_{g,0},a_{g,2^b},a_{g,1},a_{g,2^b+1},\ldots,
 a_{g,2^b-1},a_{g,2^{b+1}-1}.
\]
After interleaving, every target pair in Eq.~\eqref{eq:decoder} is adjacent.
A CNOT chain distributes $i_b$ to neighboring zero work bits~\cite{MooreNilsson}, placing one control beside each pair so that the controlled swaps run in parallel.
Because these swaps preserve the copied controls, their CNOT chain can then be reversed before undoing the interleaving.

Routing at level $b$ costs depth $\mathcal O(2^b+\log w)$ and $\mathcal O(2^{2b}+\log w)$ gates, including alignment of the low address bits.
As $T_0=\Theta(w)$, summing over the levels gives
\[
 \sum_{b=0}^{\log_2T_0-1}(2^b+\log w)=\mathcal O(w),\qquad
 \sum_{b=0}^{\log_2T_0-1}(2^{2b}+\log w)=\mathcal O(w^2).
\]
The initial alignment and the reversible comparison in Eq.~\eqref{eq:blockhit} fit within the same bounds; all temporary controls use $\mathcal O(w)$ clean qubits.

\paragraph{Coefficient phases.}
Once the selectors are available, each must meet the phase fields for its table entry.
A fixed pair-swap network arranges all these meetings without repeatedly routing an entire word.
We pair selector $a_{g,i}$ with a clean bit to form $\mathsf A_i:=(a_{g,i},0)$; the phase-field pairs are $\mathsf B_0:=(S_g^\chi,0)$ and $\mathsf B_{\beta+1}:=(S_{g,\beta}^\omega,J_{g,\beta})$ for $0\le\beta<m$.
Thus each frequency pair contains both its phase qubit and the offset bit that controls whether its factor is present.
The $\mathsf A$ pairs belong to $\reg W_g$, while the non-work fields in the $\mathsf B$ pairs belong to $\reg R_g$.
The corresponding coefficients are abbreviated as $\chi_{g,i}:=\chi_{gT_0+i}$ and $z_{g,i,\beta}:=z_{gT_0+i,\beta}$.
A fixed permutation places all selector pairs first and the phase-field pairs after them, each in increasing index order. This uses existing registers and $\mathcal O(w)$ empty alignment positions.

Moving each $\mathsf B_j$ to the left of every $\mathsf A_i$ creates the required encounters.
The pair $(i,j)$ meets in round $T_0-1-i+j$ at left position $i+j$, so the encounters within a round are disjoint.
Immediately before the exchange with $\mathsf B_0$, the selector controls $\Dph(\chi_{g,i})$ on $S_g^\chi$.
At an encounter with $\mathsf B_{\beta+1}$, a Toffoli computes $a_{g,i}J_{g,\beta}$ into the clean bit, that bit controls $\Dph(z_{g,i,\beta})$, and a second Toffoli clears it.
Each encounter has constant cost on adjacent constant-width pairs.
The $T_0+m$ rounds contain $T_0(m+1)$ encounters in total.
After these phase operations, reversing only the pair exchanges and the initial permutation restores every field without removing the phases.
Including this restoration, the depth is $\mathcal O(w)$ and the gate count is $\mathcal O(w^2)$.
Figure~\ref{fig:crossing} depicts this fixed swap network~\cite{OGorman}.

\paragraph{Data operations.}
The same selector controls each nonidentity factor of $P_\lambda$.
With $S_{\rm ph}:=\operatorname{diag}(1,i)$, controlled $X$ is CNOT, controlled $Z$ is CNOT conjugated by target Hadamards, and controlled $Y$ is CNOT preceded by $S_{\rm ph}^\dagger$ and followed by $S_{\rm ph}$.
These identities preserve the exact phases.
Because each string has constantly many factors within a fixed radius of its assigned site, every controlled gate can be routed locally and its swaps reversed, even when the string crosses a subblock boundary.
A constant coloring resolves conflicts, giving insertion depth $\mathcal O(1)$ and $\mathcal O(w)$ gates per subblock.
The selectors are no longer needed after this use, so reversing the decoder and the initial comparison clears them.

A product-formula stage needs only the branch bit $b_g:=e(1-X)$, which is distributed along the subblock by a CNOT chain.
For $F_a=h_aP_a$, the required controlled operation is $e^{-i\alpha_kh_ahP_a}$.
Changing the Pauli factors to $Z$ and accumulating their parity in one data qubit reduces this to a controlled $R_z(2\alpha_kh_ah)$ rotation, where $R_z(\theta):=e^{-i\theta Z/2}$; the parity and basis changes are then reversed.
Thus the zero-control branch is exactly the identity.
A controlled $R_z$ rotation uses two CNOTs and two single-qubit rotations, whereas a scalar term $P_a=I$ contributes $\operatorname{diag}(1,e^{-i\alpha_kh_ah})$ on the control itself.
These are the standard controlled-rotation and controlled-phase constructions~\cite[Lemmas~5.2 and~5.4]{Barenco}, detailed in Appendix~\ref{app:gates}.

Each rotation has constant depth and gate count.
Disjoint term supports within a color allow parallel execution, with only a constant further coloring for the fixed routing neighborhoods.
For $D\ge2$, stage $V_k$ is completed in every subblock before $V_{k+1}$ begins, retaining the product-formula order.
For $D=1$, the staggered schedule in Appendix~\ref{app:1d} preserves the order of operations with overlapping supports and implements the same ordered product.
There are $\mathcal O(w)$ assigned terms and only $L=\mathcal O(1)$ stages, so the rotations cost $\mathcal O(w)$ gates.
Distributing and clearing $b_g$ adds depth and gate count $\mathcal O(w)$.
A separately invoked stage performs this cleanup itself; in a complete invocation the bit may instead be retained through all stages and cleared at the end.
Since $L$ is fixed and work is reused only after uncomputation, both choices satisfy the same bounds.

The stages of one local use have the following costs; the allocated $\mathcal O(w)$ work qubits are reused after each stage.
\begin{center}
\begin{tabular}{lcc}
\toprule
Operation & Depth & Gate count\\
\midrule
Selector preparation and cleanup & $\mathcal O(w)$ & $\mathcal O(w^2)$\\
Coefficient phases and restoration of pairs & $\mathcal O(w)$ & $\mathcal O(w^2)$\\
Selected Pauli string $P_\lambda$ & $\mathcal O(1)$ & $\mathcal O(w)$\\
Product-formula stages and control cleanup & $\mathcal O(w)$ & $\mathcal O(w)$\\
\bottomrule
\end{tabular}
\end{center}
Adding these contributions gives Eq.~\eqref{eq:localresources}.
Throughout the construction, every control keeps its basis value, every permutation is reversed, and all temporary registers are cleared.
Thus each local call has both the action and the resource bounds required by the distribution circuit.
\end{proof}

\begin{figure}[tbp]
\centering
\includegraphics[width=.90\linewidth]{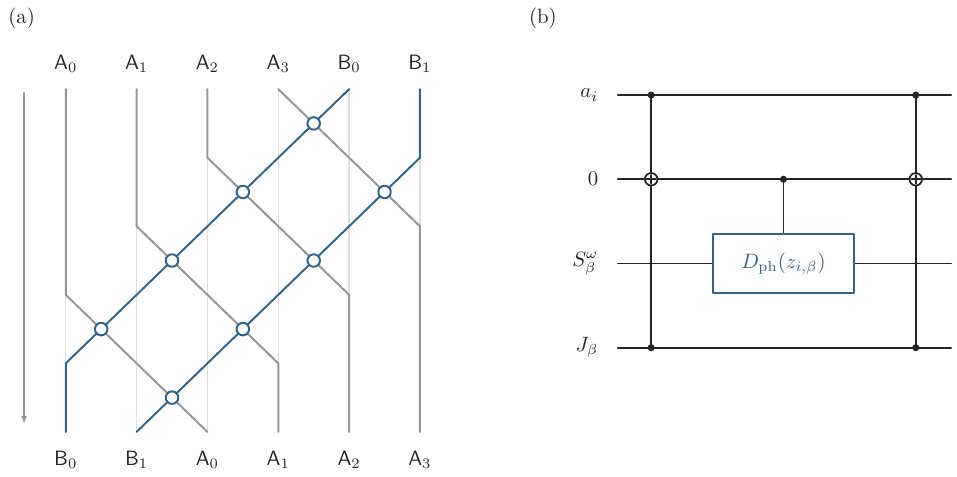}
\caption{Local phase operations in a fixed pair-swap network; the subblock index $g$ is suppressed. (a) Four selector pairs $\mathsf A_i=(a_i,0)$ cross the two phase-field pairs $\mathsf B_0=(S^\chi,0)$ and $\mathsf B_1=(S^\omega_0,J_0)$. Columns denote fixed pair positions and lines track their contents through five rounds. Each circle marks a phase operation immediately before a pair exchange. (b) At a frequency encounter, a Toffoli computes $a_iJ_\beta$ into the clean bit; that bit controls $D_{\rm ph}(z_{i,\beta})$, and a second Toffoli clears it. Reversing the pair exchanges restores every field while retaining the phases.}
\label{fig:crossing}
\end{figure}

For the complete local circuit $\mathcal L_g:=\mathcal L_{g,L}\cdots\mathcal L_{g,0}$, we therefore have
\begin{equation}
 \mathcal L_g\ket y_{\reg R_g}\ket0_{\reg W_g}\ket\psi_{\reg D}
 =e^{i\phi_g(y)}\ket y_{\reg R_g}\ket0_{\reg W_g}W_g(y)\ket\psi_{\reg D}.
 \label{eq:localsemantics}
\end{equation}
A single-stage call has the same identity with $\mathcal L_{g,k}$ and $D_{g,k}$ in place of $\mathcal L_g$ and $W_g$.
Only the insertion stage $k=0$ contributes $\phi_g$; the other stages contribute no such phase.

\paragraph{Block encoding.}
Taking the signal matrix element of the assembled circuit gives
\begin{equation}
 \iota_0^\dagger B\iota_0
 =\sum_z\pi(z)\left(\prod_u\zeta_u(z)\right)W_z
 =\widehat A_Q/2.
 \label{eq:blockencoding}
\end{equation}

\subsection{Extension to general local unitaries}\label{sec:compiler-interface}
The same construction selects ordered products of general bounded-support unitaries. The corollary below isolates this consequence; coefficient preparation remains specific to each LCU application.

\begin{corollary}[Selection of ordered local products]\label{cor:local-product-select}
Let a cube have side $\ell$ and $v=\ell^D$ sites, and let $(U_\lambda)_{0\le\lambda<A}$ be an explicit table of unitaries, each supported within a fixed radius of an assigned site, with a bounded number assigned to any site. The table is padded with identities to the binary address format of Eq.~\eqref{eq:address}. The $N\ge1$ address records occupy the source path of Section~\ref{sec:layout-overview} in their specified order, with width at most $w$ per record. Here $w$ is a power of two large enough for the address and a constant number of control bits, and
\[
 w\ge\log_2\ell,\qquad Nw+w\log\ell=\mathcal O(\ell).
\]
For all sufficiently large $\ell$, there is a nearest-neighbor circuit $\SELECT_{\rm prod}$ satisfying
\begin{equation}
 \SELECT_{\rm prod}\ket0_{\reg W}\ket{\boldsymbol\lambda}_{\reg S}\ket\psi_{\reg D}
 =
 \ket0_{\reg W}\ket{\boldsymbol\lambda}_{\reg S}
 U_{\lambda_{N-1}}\cdots U_{\lambda_0}\ket\psi_{\reg D},
 \label{eq:ordered-local-select}
\end{equation}
where $\boldsymbol\lambda:=(\lambda_0,\ldots,\lambda_{N-1})$. The depth is $\mathcal O(\ell)$, the gate count is $\mathcal O(v\ell)$, and the ancillary count is $\mathcal O(v)$.
\end{corollary}
\begin{proof}
The decoder can be reused by adjoining execution and insertion bits $e=X=1$ to each actual address record at the source.
The circuit creates these bits with $X$ gates and clears them with the same gates after return; padding records remain zero, so they are inactive even if $U_0\ne I$.
Replacing the controlled Pauli operation by controlled $U_\lambda$ changes the local cost by only a constant: a unitary on constantly many qubits, including its control, has an exact constant-size decomposition into one-qubit gates and CNOTs~\cite{Barenco,Shende06}, and its neighborhood has fixed radius.
With $\mathcal O(w)$ entries per subblock, decoding, selection, and uncomputation consequently cost depth $\mathcal O(w)$, $\mathcal O(w^2)$ gates, and $\mathcal O(w)$ work.
A constant coloring handles neighboring routing conflicts.

For each basis address only one subblock selects a nonidentity operation, so the group operators commute within the single data stage.
Lemma~\ref{lem:geometry} therefore applies for $D\ge2$, and its subblock-size condition holds for all sufficiently large $\ell$ under the stated assumptions.
For $D=1$, one event per address in Appendix~\ref{app:1d} suffices; bounded support supplies its overlap-order condition.

\end{proof}

For an LCU of such products, the corollary supplies the selection stage, including distribution and removal of its coherent controls~\cite{ChildsWiebe,QSVT}.
An application must still provide efficient coefficient preparation and a suitable normalization.

\FloatBarrier
\section{Proof of the simulation upper bound}\label{sec:error}
We prove the local simulation bound and its global application in three steps. First, amplitude amplification turns the block encoding into a simulation circuit with a bound on the entire output, including its ancillas. We then choose parameters to prove Theorem~\ref{thm:block}. Finally, we compose the block circuits in the layers of Fact~\ref{fact:spatial} and add their errors over all time intervals.

\subsection{Amplitude amplification and ancillary error}\label{sec:amplification}
One step of oblivious amplitude amplification removes normalization two for an exact unitary~\cite{BCCKS,Taylor}. Here $\widehat A_Q$ is approximate, so we must also bound the amplitude left in nonzero signal states. We use the construction of Refs.~\cite[Lemma~4.6]{BCCKS} and~\cite[Eqs.~(13)--(16)]{Taylor}, keeping this ancillary component in the calculation.

We continue to use the signal, work, and data registers $\reg S,\reg W,\reg D$ from Section~\ref{sec:lcu-overview}, and define the embeddings
\begin{equation}
 \iota_{\rm sig}\ket\varphi_{\reg S\reg D}
 :=\ket0_{\reg W}\ket\varphi_{\reg S\reg D},\qquad
 \iota_0\ket\psi:=\iota_{\rm sig}(\ket0_{\reg S}\ket\psi_{\reg D}).
 \label{eq:signalembedding}
\end{equation}

The zero-work subspace is therefore $\mathcal K:=\operatorname{ran}\iota_{\rm sig}$.

\begin{lemma}[Full-output bound for one-step oblivious amplitude amplification]\label{lem:amplify}
Let $B$ be a unitary preserving $\mathcal K$ and satisfying $\iota_0^\dagger B\iota_0=\widehat A_Q/2$. The signal reflection and amplified circuit are
\begin{equation}
 R_{\rm sig}:=I_{\reg W}\otimes(2\ket0\!\bra0_{\reg S}-I_{\reg S})\otimes I_{\reg D},
 \qquad \calU:=-BR_{\rm sig}B^\dagger R_{\rm sig}B.
 \label{eq:amplification}
\end{equation}
The full output satisfies
\begin{equation}
 \calU\iota_0=\iota_0\widehat A_Q+B\iota_0(I-\widehat A_Q^\dagger\widehat A_Q).
 \label{eq:fullcolumn}
\end{equation}
For every data unitary $U$ with $\delta:=\norm{\widehat A_Q-U}\le1$, the error obeys
\begin{equation}
 \norm{\calU\iota_0-\iota_0U}\le3\delta+\delta^2\le4\delta.
 \label{eq:amplificationnorm}
\end{equation}
\end{lemma}
\begin{proof}
Because $B$ is unitary, its invariant finite-dimensional subspace $\mathcal K$ is also invariant under $B^\dagger$.
Every intermediate state therefore stays in $\mathcal K$, where $R_{\rm sig}=2\Pi-I$ for $\Pi:=\iota_0\iota_0^\dagger$.
With $A_{\rm enc}:=\iota_0^\dagger B\iota_0=\widehat A_Q/2$, expanding the two reflections gives
\begin{align}
 \calU\iota_0
 &=-B(2\Pi-I)(2B^\dagger\iota_0A_{\rm enc}-\iota_0)\notag\\
 &=-4B\iota_0A_{\rm enc}^\dagger A_{\rm enc}+B\iota_0+2\iota_0A_{\rm enc}\notag\\
 &=\iota_0\widehat A_Q+B\iota_0(I-\widehat A_Q^\dagger\widehat A_Q).
 \label{eq:amplificationaction}
\end{align}
Unlike the projected block $\Pi\calU\iota_0$, this identity retains the component with nonzero signals.
Its size is controlled by the failure of $\widehat A_Q$ to be unitary.
Indeed, writing $\Delta:=\widehat A_Q-U$ and using $U^\dagger U=I$ yields
\[
 I-\widehat A_Q^\dagger\widehat A_Q
 =-U^\dagger\Delta-\Delta^\dagger U-\Delta^\dagger\Delta,
 \qquad \norm{I-\widehat A_Q^\dagger\widehat A_Q}\le2\delta+\delta^2.
\]
Since both $\iota_0$ and $B\iota_0$ are isometries, Eq.~\eqref{eq:fullcolumn} now bounds the full error by the approximation error $\delta$ plus this unitarity defect.
Thus $\norm{\calU\iota_0-\iota_0U}\le\delta+(2\delta+\delta^2)$, proving Eq.~\eqref{eq:amplificationnorm}.
\end{proof}

\paragraph{Application to the block circuit.}
For our circuit, the required invariance follows from the branchwise SELECT action in Eq.~\eqref{eq:selectinterface}:
\begin{equation}
 \SELECT\iota_{\rm sig}=\iota_{\rm sig}V_{\rm sig},\qquad
 V_{\rm sig}:=\sum_{z,s}e^{i\Phi(z,s)}\ket{z,s}\!\bra{z,s}\otimes W_z.
 \label{eq:globalclean}
\end{equation}
Each $W_z$ is unitary, and the orthogonal signal labels make $V_{\rm sig}$ unitary as well.
Consequently, SELECT maps all of $\mathcal K$ onto itself, not merely inputs with zero signals.
PREPARE acts only on $\reg S$, so $B$ and $B^\dagger$ inherit this invariance.
Together with Eq.~\eqref{eq:blockencoding}, this verifies the hypotheses of Lemma~\ref{lem:amplify}.

It remains to implement the signal reflection within the same geometry.
The standard reversible-predicate construction~\cite{BHMT,BCCKS,Barenco} reduces it to testing whether any signal bit is nonzero.
The signals occur along their $N$ consecutive slot cells in the order $s_0,\ldots,s_{q_S-1}$, with $q_S=\mathcal O(Nw)$.
Starting with $r_0:=0$, a reversible circuit computes their Boolean OR into fresh targets according to
\[
 r_{i+1}:=r_i\lor s_i=r_i\oplus s_i\oplus r_is_i.
\]
Here $\oplus$ denotes addition modulo two, so each update uses two CNOTs and one Toffoli.
Across a gap in the layout, the accumulator is carried to the next signal position.
A $Z$ gate on the final bit then gives phase $+1$ only when all signals are zero; reversing the chain clears its work without changing that phase.
The resulting reflection has depth, gate count, and clean-work requirement $\mathcal O(Nw)$.
Finally, a one-qubit $-I$ provides the overall minus sign in Eq.~\eqref{eq:amplification}.

This physical implementation need only agree with $R_{\rm sig}$ on zero-work inputs.
Equation~\eqref{eq:globalclean} guarantees precisely such an input at every reflection, even if the signals are nonzero.

\subsection{Parameter choice and proof of the block theorem}\label{sec:parameters}
The target scales are $r,Q,N=\Theta(M/\log M)$ and $w=\Theta(\log M)$. Then both sequential processing and coefficient lookup have depth $\mathcal O(M)$. We choose a degree weight $\xi=\Theta(M^\vartheta)$ so that the truncation error is exponentially small while the weighted coefficient sum remains bounded. The strict inequality $p>D$ makes these requirements compatible.

We first bound the padded table size before fixing the address width. Enclose the original rectangle in a cube of side $\ell$, its largest side length. The bounded aspect ratio gives
\begin{equation}
 v\le v_{\rm cube}=\ell^D\le C_{\rm rect}v.
 \label{eq:blockpaddingvolumes}
\end{equation}
For $D\ge2$, the further padding to $\ell_p=b2^K<2\ell$ multiplies the volume by at most $2^D$; the one-dimensional padding also changes it by a constant factor.
Because each site carries only a constant number of residual entries, including padded ones, these bounds give a computable estimate
\[
 A\le\overline A:=C_{\rm tab}v.
\]
Here $C_{\rm tab}$ includes all padding and is independent of $w$ when $b\le\ell$. The parameter scales below verify this condition for sufficiently large $M$.

Choose the positive exponent $\vartheta:=(p-D)/(2(p+1))$. To make the degree weight a power of two, set $k:=\max\{1,\lfloor\vartheta\log_2M\rfloor\}$ and $\xi:=2^k$. The cutoff is
\begin{equation}
 Q:=\max\left\{p+1,\left\lceil\frac{4M}{\log\xi}\right\rceil\right\}.
 \label{eq:parameters}
\end{equation}
The weight $\xi$ in Lemma~\ref{lem:tails} grows as a small power of $M$, while the choice of $Q$ ensures $\xi^{-Q}\le e^{-4M}$.
The hypotheses of that lemma are met by choosing the number of microsteps as
\begin{equation}
 r:=\left\lceil\max\left\{
 Q,\ 2\omega_*\tau\xi,\
 \left(\frac{16c_*\overline A\,\omega_*^p\tau^{p+1}\xi^{p+1}}{p!}\right)^{1/p}
 \right\}\right\rceil.
 \label{eq:microstepchoice}
\end{equation}
The first entry ensures $r\ge Q$; the second gives $t_\omega\xi\le1/2$. With $h=\tau/r$, the third bounds the weighted insertion sum:
\[
 r\calA(\xi)
 =\frac{c_*A\omega_*^p\tau^{p+1}\xi^{p+1}}
 {p!\,r^p(1-\omega_*\tau\xi/r)}\le\frac18.
\]
Thus Lemma~\ref{lem:tails} applies for every table size $A\le\overline A$ and also bounds the coefficient factors required by the phase gates.

With $r$ and $Q$ fixed, Eqs.~\eqref{eq:finiteparameters}--\eqref{eq:productdist} determine the remaining label parameters.
The address needs $b_{\rm addr}:=\lceil\log_2\overline A\rceil$ bits. The denominator uses $b_{\rm den}:=\lceil\log_2(Q+1)\rceil$ bits, and the formula-step counter uses $b_{\rm count}:=\lceil\log_2(N+1)\rceil$ bits.
Since $m=\lceil\log_2(Q-p)\rceil\le b_{\rm den}$, these widths also bound the offset fields.
In particular, $\reg O_u$ contains $e_u,X_u$, the address, the $m$ offset bits, the phase qubit $S_u^\chi$, and the $m$ frequency-phase qubits.
Consequently, its number of qubits is
\begin{equation}
 |\reg O_u|=\log_2A+2m+3
 \le b_{\rm addr}+2b_{\rm den}+3.
 \label{eq:transportwidth}
\end{equation}
Each denominator register $\reg B_u$ has $b_{\rm den}$ bits because $L_b=2^{b_{\rm den}}$, and each count register $\reg Z_u$ has $b_{\rm count}$ bits.
The subblock and within-subblock fields only partition the existing address bits; they do not require another address register.
Moreover, the arithmetic in Appendix~\ref{app:gates} uses work linear in the operand width.
A sufficiently large fixed $C_w$ therefore allows all fields and their work to fit in the power-of-two width
\begin{equation}
 w:=2^{\left\lceil\log_2\bigl(C_w(1+b_{\rm addr}+b_{\rm den}+b_{\rm count})\bigr)\right\rceil}
 =\Theta\bigl(1+\log\overline A+\log(Q+1)+\log(N+1)\bigr).
 \label{eq:width}
\end{equation}
This width accommodates the transported fields as well as the retained denominator, count, and work registers.
Since Eq.~\eqref{eq:blockpaddingvolumes} also gives $\log\overline A=D\log\ell+\mathcal O(1)$, increasing $C_w$ ensures the geometric requirement $w\ge\log_2\ell$.

These widths fix the subblocks and the padded table. Its entries then specify the gate angles in Eq.~\eqref{eq:factors}. Since the earlier choices used the bound $\overline A$, they remain valid for this table.

\begin{proof}[Proof of Theorem~\ref{thm:block}]
The parameter choice makes Lemma~\ref{lem:tails} applicable to the operator encoded in Eq.~\eqref{eq:blockencoding}.
Applying the full-output estimate of Lemma~\ref{lem:amplify} then yields
\begin{equation}
 \norm{\calU \iota_0-\iota_0e^{-iH_{\rm blk}\tau}}
 \le4(e^{1/8}+2)\xi^{-Q}
 \le4(e^{1/8}+2)e^{-4M}<e^{-M}
 \label{eq:finalblockbound}
\end{equation}
for $M\ge1$.
In particular, the unamplified error is below one, as required for that estimate.
It remains to show that the same parameters fit the geometry and give the claimed resources; increasing the fixed threshold $M_0$ will enforce all size conditions.

Because $p>D$, the exponent $\vartheta$ lies strictly between zero and one.
Thus $\xi=\Theta(M^\vartheta)$, $Q=\Theta(M/\log M)$, and $v=\Theta(M^D)$.
The main point is that, for large $M$, neither of the other two entries in Eq.~\eqref{eq:microstepchoice} exceeds $Q$.
To compare the third entry with $Q$, we evaluate the following ratio at $r=Q$. Uniformly over $0<\tau\le1$, the $p$th power of that entry divided by $Q^p$ is at most a constant times
\begin{equation}
 \frac{\overline A\,\xi^{p+1}}{r^p}
 =O\bigl(M^{D-p+\vartheta(p+1)}(\log M)^p\bigr)
 =O\bigl(M^{-(p-D)/2}(\log M)^p\bigr)=o(1).
 \label{eq:balance}
\end{equation}
The second entry is also $o(Q)$ because $\xi/Q=o(1)$.
Hence the maximum in Eq.~\eqref{eq:microstepchoice} equals $Q$ for all sufficiently large $M$, and the selected parameters obey
\begin{equation}
 r,Q,N=\Theta(M/\log M),\qquad
 w=\Theta(\log M),\qquad Nw,Qw=\Theta(M).
 \label{eq:scales}
\end{equation}
These scales imply all geometric size conditions.
For $D\ge2$, $b^D=\Theta(w)$ and $\ell=\Theta(M)$ give $b\le\ell$ and
\[
 Nw+w\log\ell=O\bigl(M+(\log M)^2\bigr)=\mathcal O(\ell).
\]
Together with $w\ge\log_2\ell$, this verifies the hypotheses of Lemma~\ref{lem:geometry} and the padding bound used to define $\overline A$.
For $D=1$, the event count is $F=(L+1)N$ with fixed $L$, so $Fw=\mathcal O(\ell)$ verifies the corresponding hypothesis in Appendix~\ref{app:1d}.

We can therefore use the local compilation bounds without additional size assumptions.
PREPARE has depth $\mathcal O(1)$ and uses $\mathcal O(Nw)$ gates; the prefix circuit $\calT$ and its inverse have depth $\mathcal O(Nw)$ and use $\mathcal O(Nw^2)$ gates.
The source tables in $\mathcal D_0$ contribute depth $\mathcal O(Qw)$ and $\mathcal O(NQw)$ gates.
Lemma~\ref{lem:geometry}, or Appendix~\ref{app:1d}, bounds distribution, address-dependent phases, ordered data operations, and return alignment together by depth $\mathcal O(M)$, $\mathcal O(vM)$ gates, and $\mathcal O(v)$ ancillas.
Finally, Section~\ref{sec:amplification} gives depth, gate count, and work $\mathcal O(Nw)$ for each signal reflection.

Equation~\eqref{eq:scales} makes every depth contribution $\mathcal O(M)$.
It also places all gate counts below $\mathcal O(vM)$: the two source terms are $Nw^2=\mathcal O(M\log M)$ and $NQw=\mathcal O(M^2/\log M)$, whereas $vM=\Theta(M^{D+1})$ for $D\ge1$.
The source signals, prefix outputs, scalar-table work, and reflection work occupy $\mathcal O(Nw)=\mathcal O(M)\le \mathcal O(v)$ qubits on the reserved source path.

Since amplification uses only constantly many calls to $B$, $B^\dagger$, and the signal reflection, these are also the depth, gate, and ancillary bounds for $\calU$.
\end{proof}

\subsection{Proof of the global simulation bound}\label{sec:global}
Applying the local simulator of Theorem~\ref{thm:block} to Fact~\ref{fact:spatial} gives the following global bound, with the side-length condition required by the block layout.

\begin{theorem}[Simulation on a large lattice]\label{thm:global-formal}
Let $H$ satisfy Eq.~\eqref{eq:model}, with $\mathcal T,M$ as in Eq.~\eqref{eq:globalM}. If every side of the lattice rectangle has length at least $c_1M$ for a sufficiently large constant $c_1$, then the simulation problem in Eq.~\eqref{eq:embedding} has a nearest-neighbor solution with
\begin{equation}
 \begin{aligned}
 \text{circuit depth}\quad\dcnt&=O(\mathcal TM),\\
 \text{gate count}\quad\gcnt&=O(n\mathcal TM),\\
 \text{ancillary qubits}\quad\acnt&=O(n).
 \end{aligned}
 \label{eq:globalresult}
\end{equation}
The ancillary qubits have bounded spatial density and are reused throughout the circuit.
\end{theorem}

The remaining issues are inverse block calls and coherent reuse of ancillas. Both are controlled by unitarity.

\begin{proof}[Proof of Theorem~\ref{thm:global-formal}]
When $T=0$, the identity circuit suffices.
Otherwise, Eqs.~\eqref{eq:normalizedrestriction}--\eqref{eq:normalizedevolution} express the target as $\bar U_\Lambda(\tau)^{\mathcal T}$, and we use the same layers from Fact~\ref{fact:spatial} in each normalized interval.
With $\ell_{\rm buf}=\lceil M/\mu\rceil$, Eq.~\eqref{eq:hhkl-error} bounds the spatial error by $Cn e^{-M}$ per interval.
For each factor $\bar U_Y(\sigma_{a,Y}\tau)$, we apply Theorem~\ref{thm:block} with $H_{\rm blk}=\bar H_Y$. The forward circuit has error $e^{-M}$; a negative sign calls for its inverse.

Taking the adjoint preserves this error in the required input subspace, even though the estimate is not a full-space comparison of the two unitaries.
More precisely, the hypothesis $\norm{V\iota_0-\iota_0U}\le\eta$ and unitarity give
\begin{equation}
 V^\dagger \iota_0-\iota_0U^\dagger
 =V^\dagger(\iota_0U-V\iota_0)U^\dagger,
 \qquad \norm{V^\dagger \iota_0-\iota_0U^\dagger}\le\eta.
 \label{eq:inverseerror}
\end{equation}

The actual circuits in execution order are denoted by $V_1,\ldots,V_s$, with adjoints where required, and $U_j$ denotes the corresponding restricted data unitary.
The embedding $\iota_0$ now initializes all ancillas, so every call satisfies $\norm{V_j\iota_0-\iota_0U_j}\le\eta_j$ with $\eta_j=e^{-M}$.
The appropriate comparison is the exact telescoping identity
\begin{equation}
 V_s\cdots V_1\iota_0-\iota_0U_s\cdots U_1
 =\sum_{j=1}^s V_s\cdots V_{j+1}
       (V_j\iota_0-\iota_0U_j)U_{j-1}\cdots U_1.
 \label{eq:telescoping}
\end{equation}
In its $j$th summand, the later actual circuits multiply on the left and the earlier ideal evolutions on the right.
Both factors are unitary, so no assumption about the actual ancillary state before call $j$ is needed.
The triangle inequality consequently gives
\begin{equation}
 \norm{V_s\cdots V_1\iota_0-\iota_0U_s\cdots U_1}
 \le\sum_{j=1}^s\norm{V_j\iota_0-\iota_0U_j}
 \le\sum_{j=1}^s\eta_j.
 \label{eq:composederror}
\end{equation}
The embeddings in Eq.~\eqref{eq:telescoping} describe the ideal comparison. The implemented product has no intermediate measurement or reset; Eq.~\eqref{eq:composederror} includes the ancillary error left by earlier calls.

There are at most $C_1n\mathcal T$ calls, since the block volumes sum to $\mathcal O(n)$ in each interval and every block has positive volume.
Adding their errors to the decomposition error from Fact~\ref{fact:spatial} therefore gives
\[
 C_1n\mathcal T e^{-M}+C_2n\mathcal T e^{-\mu\ell_{\rm buf}}
 \le (C_1+C_2)n\mathcal T e^{-M}.
\]
Choosing the fixed constant $c_0$ in Eq.~\eqref{eq:globalM} sufficiently large makes this sum at most $\epsilon$ and also ensures $M\ge M_0$ for the block-side comparison constants.

Within each group from Section~\ref{sec:block-reduction}, the block circuits run in parallel with depth $\mathcal O(M)$.
The volumes sum to $\mathcal O(n)$ over all groups in an interval, giving $\mathcal O(nM)$ gates and $\mathcal O(n)$ simultaneous ancillas.
A fixed number of ancillary qubits per original site suffices: each group has disjoint enlarged rectangles, and there are only constantly many groups.
Reusing this placement and the same gate lists over $\mathcal T$ intervals gives depth $\mathcal O(\mathcal TM)$ and $\mathcal O(n\mathcal TM)$ gates, without increasing the $\mathcal O(n)$ ancillary count.

\end{proof}

\FloatBarrier
\section{Proof of the circuit-depth lower bound}\label{sec:lower}
We now prove Theorem~\ref{thm:depth-lower-bound}. A shallow circuit has a strict causal cone, whereas XXZ dynamics can send an excitation beyond it with a small positive probability. As explained in Section~\ref{sec:insight}, many separated probes amplify this difference into a detectable event. A state-dependent comparison with independent reference walks controls the effect of interactions between the probes.

To match the proof notation to Eq.~\eqref{eq:lb-hamiltonian}, set $\Delta_{\rm XXZ}=\Delta$ and $\Delta_0=|\Delta|$. More generally, the bound is uniform for $|\Delta_{\rm XXZ}|\le\Delta_0$ at fixed $J_{\rm hop}=J/[2D(1+\Delta_0)]$. Write $\ell_{\rm sys}$ for the cube side, so $n=\ell_{\rm sys}^D$, and use the time count $\mathcal T$ from Eq.~\eqref{eq:globalM}. The precise claim is as follows.

\begin{theorem}[Worst-case geometric depth lower bound]\label{thm:depth-lower-bound-formal}
Fix $D\ge1$ and $\Delta_0\ge0$. There are constants $c_{\rm time},c_{\rm lb},\epsilon_{\rm lb}>0$ and $\ell_0$, depending only on these fixed choices, such that the following holds. If
\begin{equation}
 \ell_{\rm sys}\ge\ell_0,\qquad
 2D(1+\Delta_0)\le J|T|\le c_{\rm time}\ell_{\rm sys},
 \qquad 0<\epsilon\le\epsilon_{\rm lb},
 \label{eq:lb-regime}
\end{equation}
then every circuit satisfying Eq.~\eqref{eq:embedding} for $H=H_{\rm XXZ}$ has depth
\begin{equation}
 \dcnt\ge c_{\rm lb}\min\!\left\{
 \ell_{\rm sys},\quad
 \mathcal T+
 \frac{\log(n/\epsilon)}
 {\log\!\bigl(e+\log(n/\epsilon)/\mathcal T\bigr)}
 \right\}.
 \label{eq:depth-lower-bound}
\end{equation}

\end{theorem}

We first derive the causal distinguishability test and estimate the arrival signal of separated probes. A relative bound on the boundary error transfers that signal to the interacting model. A separate ballistic estimate supplies the linear dependence on time. The coefficients are explicitly known throughout; this is a circuit-depth obstruction rather than a query-access bound~\cite{BACS,BCK}.

Throughout this argument, $\Lambda=\{1,\ldots,\ell_{\rm sys}\}^D$ and $n=\ell_{\rm sys}^D$, with $H_{\rm XXZ}$ as in Eq.~\eqref{eq:lb-hamiltonian}.
We express propagation in units of the hopping strength by setting
\begin{equation}
 \kappa_{\rm XXZ}:=2D(1+\Delta_0),\qquad
 \theta_{\rm hop}:=J_{\rm hop}|T|=J|T|/\kappa_{\rm XXZ}.
 \label{eq:lb-normalization}
\end{equation}
It suffices initially to work in the range $1\le\theta_{\rm hop}\le c\ell_{\rm sys}$ for a sufficiently small constant $c>0$.

The estimates are uniform in the signs of $T$ and $\Delta_{\rm XXZ}$: for the real generators and computational-basis sources used here, reversing time complex-conjugates the state and preserves all arrival probabilities.

\subsection{Causal cones and distinguishability}\label{sec:lb-causal}
In one circuit layer, an input can influence qubits at distance at most $\rho_{\rm gate}$, a fixed radius that includes the lattice refinement and ancillary paths~\cite{BHV,Rosenbaum}. A depth-$\dcnt$ circuit therefore has causal radius
\begin{equation}
 a_{\rm cone}:=\rho_{\rm gate}\dcnt.
 \label{eq:lb-cone-radius}
\end{equation}
Let $\ket\Omega=\ket0^{\otimes n}$ be an eigenstate of $H$, and let $\ket S=\prod_{x\in S}X_x\ket\Omega$ place excitations at source sites $S$. For a receiver region $B$, denote by $E_B$ the event of detecting at least one excitation. If the sources lie outside its backward causal cone, the circuit gives the same receiver statistics on both inputs. The true vacuum has no arrivals, so an accurate simulator can reproduce only a small arrival probability from $\ket S$.

\begin{samepage}
\begin{lemma}[Causal distinguishability]\label{lem:lb-causal}
Suppose $H\ket\Omega=E_\Omega\ket\Omega$ and a depth-$\dcnt$ circuit satisfies Eq.~\eqref{eq:embedding} for $H$. If $\dist(S,B)>a_{\rm cone}$, then the true evolution obeys
\begin{equation}
 \Pr_H(E_B\mid S)\le2\epsilon.
 \label{eq:lb-causal-test}
\end{equation}
Here the probability is evaluated after evolution for time $T$ from $\ket S$.
\end{lemma}
\end{samepage}
\begin{proof}
Outside $S$, the two inputs and their initialized ancillas agree.
The separation condition places $S$ outside the backward cone of $B$, so the circuit produces identical reduced states on $B$ for these inputs.

For normalized pure states $\ket a,\ket b$, the trace distance satisfies~\cite{FuchsVanDeGraaf99}
\begin{equation}
 D_{\rm tr}(a,b):=\tfrac12\bigl\|\ket a\!\bra a-\ket b\!\bra b\bigr\|_1
 =\sqrt{1-|\braket{a}{b}|^2}\le\|\ket a-\ket b\|.
 \label{eq:lb-trace-distance}
\end{equation}
Here $\|A\|_1:=\operatorname{Tr}\sqrt{A^\dagger A}$ is the trace norm.
Trace distance bounds the change in any event probability, including one measured only on $B$.
Thus Eq.~\eqref{eq:embedding} changes the probability of $E_B$ by at most $\epsilon$ on each input.
The true vacuum probability is zero by the eigenstate assumption, and the two circuit probabilities agree.
Comparing through those identical probabilities gives the bound $2\epsilon$.
\end{proof}

A single arrival can be too rare to detect at constant error. We combine $k$ probes with independent reference events, choosing $k$ to be at most the reciprocal of the single-probe probability. The resulting signal grows as $kp_0$, so a comparison error of order $k\eta$ is tolerable when $\eta\ll p_0$. The lemma requires reference states for every selected subset because this choice of $k$ may use only part of the available packing.

\begin{lemma}[Separated propagation probes]\label{lem:separated-probes}
Suppose $H\ket\Omega=E_\Omega\ket\Omega$. We are given $K\ge1$ distinct source sites $x_1,\ldots,x_K$ and disjoint receiver regions $B_1,\ldots,B_K$ with $\dist(x_i,B_j)>a_{\rm cone}$ for all $i,j$. Fix $0<p_0\le1/4$ and $0\le\eta\le p_0/100$. For each nonempty subset $\mathcal I\subseteq\{1,\ldots,K\}$, write $S_{\mathcal I}:=\{x_i:i\in\mathcal I\}$ and $B_{\mathcal I}:=\bigcup_{i\in\mathcal I}B_i$. Suppose there is a normalized reference state $\ket{\phi_{\mathcal I}}$ for which the events $E_{B_i}$, $i\in\mathcal I$, are independent, each has probability at least $p_0$, and
\begin{equation}
 \left\|e^{-iT(H-E_\Omega I)}\ket{S_{\mathcal I}}-\ket{\phi_{\mathcal I}}\right\|
 \le |\mathcal I|\eta.
 \label{eq:probe-comparison-assumption}
\end{equation}
Every depth-$\dcnt$ simulation of $H$ with error $\epsilon$ then satisfies
\begin{equation}
 \epsilon\ge\frac18\min\{1,Kp_0\}.
 \label{eq:separated-probe-bound}
\end{equation}
\end{lemma}
\begin{proof}
We use $k:=\min\{K,\lfloor p_0^{-1}\rfloor\}$ probes, so their total guaranteed probability is $u:=kp_0$.
If all probes are used, then $u=Kp_0$; otherwise $u\ge1-p_0\ge3/4$.
Hence $\tfrac34\min\{1,Kp_0\}\le u\le1$.
Independence gives a reference union probability at least $1-(1-p_0)^k\ge1-e^{-u}\ge u/2$.
The vector comparison and Eq.~\eqref{eq:lb-trace-distance} change this probability by at most $k\eta\le u/100$, leaving a true probability of at least $49u/100$.
Lemma~\ref{lem:lb-causal} bounds the same probability by $2\epsilon$, so
\[
 \epsilon\ge\frac{49u}{200}
 \ge\frac{147}{800}\min\{1,Kp_0\}
 \ge\frac18\min\{1,Kp_0\}.
\]
The separation assumption applies to every selected subset, so the causal lemma is valid for the choice just made.
Subtracting $E_\Omega$ changes only a global phase and leaves these event probabilities unchanged.
\end{proof}

\subsection{Excitation propagation and boundary effects}\label{sec:lb-propagation}
We obtain independent reference walks by cutting interactions at probe boundaries~\cite{FarhiGutmann98,Bose03,Konno05}. The operator norm of the removed interactions is too large relative to the small arrival signal. Instead, we bound their action on the reference state, whose occupation near the cuts is small. The interaction-picture comparison of HHKL~\cite[Lemma~4]{HHKL} gives both estimates below.

\begin{lemma}[Evolution comparison]\label{lem:evolution-comparison}
Let $A(t)$ and $B(t)$ be continuous Hermitian generators on a finite-dimensional Hilbert space, with propagators $U_A(t)$ and $U_B(t)$ starting at the identity. For $t\ge0$, define $U_A(t,u):=U_A(t)U_A(u)^\dagger$. The interaction-picture operator $E(t):=U_B(t)^\dagger U_A(t)$ satisfies
\begin{equation}
 iE'(t)=U_B(t)^\dagger(A(t)-B(t))U_B(t)E(t),\qquad E(0)=I.
 \label{eq:comparison-generator}
\end{equation}
The two propagators also satisfy
\begin{equation}
 U_A(t)-U_B(t)=-i\int_0^t U_A(t,u)(A(u)-B(u))U_B(u)\,du.
 \label{eq:comparison-duhamel}
\end{equation}
For any unit vector $\ket\psi$, these identities give the estimates
\begin{align}
 \norm{(U_A(t)-U_B(t))\ket\psi}
 &\le\int_0^t\norm{(A(u)-B(u))U_B(u)\ket\psi}\,du,\label{eq:comparison-state}\\
 \norm{U_A(t)-U_B(t)}
 &\le\int_0^t\norm{A(u)-B(u)}\,du.\label{eq:comparison-norm}
\end{align}
\end{lemma}
\begin{proof}
Differentiating $U_B(t)^\dagger U_A(t)$ gives Eq.~\eqref{eq:comparison-generator}.
For the additive comparison, differentiation of $U_A(t,u)U_B(u)$ gives $iU_A(t,u)(A(u)-B(u))U_B(u)$; integrating between the endpoints yields Eq.~\eqref{eq:comparison-duhamel}.
Applying this identity to $\ket\psi$ before taking the norm retains the reference-state factor in Eq.~\eqref{eq:comparison-state}.
Taking the operator norm instead gives Eq.~\eqref{eq:comparison-norm}, since the remaining propagators are unitary.
\end{proof}

For a negative final time, the comparison is applied on $[0,|t|]$ to generators $-A(-u)$ and $-B(-u)$.
For time-independent generators, this simply negates $A$ and $B$.
We will use the state-dependent estimate with an explicitly known reference evolution on this nonnegative interval.

In the one-excitation sector, spin transport is a continuous-time walk~\cite{FarhiGutmann98,Bose03}.
Nonuniform engineered couplings can give perfect transfer~\cite{Christandl04}, but uniform couplings are more suitable here because the hard Hamiltonian must remain spatially uniform.
We only need a lower bound on their finite-distance arrival probability.
To compare the resulting isolated walks with XXZ dynamics, the one-body terms at artificial boundaries must first be identified correctly.

The occupation at site $x$ is $\nu_x:=\ket1\!\bra1_x$, the total excitation number is $N_{\rm exc}:=\sum_{x\in\Lambda}\nu_x$, and exchange across an edge is represented by $Q_{xy}:=(X_xX_y+Y_xY_y)/2$.
We denote the nearest-neighbor edge set by $\mathcal E_\Lambda$ and the degree of $x$ by $\deg_\Lambda(x)$.
Substituting $Z_x=I-2\nu_x$ into the XXZ interaction gives
\begin{equation}
 \begin{aligned}
 H_{\rm XXZ}-E_{\rm vac}I
 &=\sum_{\langle x,y\rangle\in\mathcal E_\Lambda}v_{xy}
       -4D\Delta_{\rm XXZ}J_{\rm hop}N_{\rm exc}+B_{\partial},\\
 v_{xy}&:=J_{\rm hop}(Q_{xy}+4\Delta_{\rm XXZ}\nu_x\nu_y),\qquad
 E_{\rm vac}:=J_{\rm hop}\Delta_{\rm XXZ}|\mathcal E_\Lambda|,\\
 B_{\partial}&:=2\Delta_{\rm XXZ}J_{\rm hop}
       \sum_x(2D-\deg_\Lambda(x))\nu_x.
 \end{aligned}
 \label{eq:lb-rewrite}
\end{equation}
Only the physical boundary supports $B_{\partial}$.
Our reference removes the crossing terms $v_{xy}$ but retains both one-body terms in Eq.~\eqref{eq:lb-rewrite}.
This choice matters: removing the original XXZ edge terms would also change the potential along each artificial boundary.
By removing $v_{xy}$ instead, an interior probe cube retains a uniform potential, and its one-excitation evolution is a hopping walk up to a scalar phase.

For an integer $b\ge1$ and a real $x\ge0$, we define the exponential-series tail by
\begin{equation}
 F_{\rm tail}(b,x):=\sum_{k\ge b}\frac{x^k}{k!}.
 \label{eq:lb-tail-definition}
\end{equation}
If $b\ge2x$, consecutive terms beyond $b$ have ratio at most $1/2$. Together with the Stirling bound $b!\ge(b/e)^b$~\cite{Robbins55}, this gives
\begin{equation}
 F_{\rm tail}(b,x)\le2(ex/b)^b.
 \label{eq:lb-factorial-tail}
\end{equation}
The same estimate holds at $x=0$.
For a local walk, short powers of the generator cannot reach a distant boundary, so the difference between two evolutions begins only in this factorial tail.
The next lemma records the two vector forms of this graph-distance principle~\cite[Section~3.5]{BenziRazouk07} needed below.

\begin{lemma}[Propagation bounds from an exponential series]\label{lem:series-locality}
Let $b\ge1$ be an integer, $t\in\mathbb R$, and $K,\widetilde K$ bounded self-adjoint operators on a common Hilbert space, with norms at most $\Lambda_*$. Let $\ket\psi$ be a unit vector. If $K^j\ket\psi=\widetilde K^j\ket\psi$ for $0\le j<b$, the two evolutions satisfy
\begin{equation}
 \norm{(e^{-itK}-e^{-it\widetilde K})\ket\psi}
 \le2F_{\rm tail}(b,\Lambda_*|t|).
 \label{eq:series-boundary}
\end{equation}
If $P$ is an orthogonal projector and $PK^j\ket\psi=0$ for $0\le j<b$, the projected evolution satisfies
\begin{equation}
 \norm{Pe^{-itK}\ket\psi}\le F_{\rm tail}(b,\Lambda_*|t|).
 \label{eq:series-projection}
\end{equation}
\end{lemma}
\begin{proof}
In either exponential series, the hypotheses remove all terms below degree $b$.
The remaining terms are bounded by $2\Lambda_*^j|t|^j/j!$ for the difference and by $\Lambda_*^j|t|^j/j!$ for the projection.
Summing gives the two inequalities.
Boundedness of the operators guarantees norm convergence even on an infinite lattice, so the same argument applies to the reference walks.
\end{proof}

A nearest-neighbor generator moves support by at most one site per multiplication.
Consequently, agreement of the generators near the source verifies the first low-power hypothesis, while separation from the projector support verifies the second.

\begin{samepage}
\begin{lemma}[Arrival probability in a finite interval]\label{lem:lb-arrival}
There is an absolute integer $A_1$ with the following property. For $\theta_{\rm hop}\ge1$, an integer $\ell_{\rm pr}\ge2\theta_{\rm hop}$, and an integer $b_{\rm pr}\ge A_1\ell_{\rm pr}$, let $A_I$ be the adjacency matrix of the nearest-neighbor interval $I:=\{-b_{\rm pr},\ldots,b_{\rm pr}\}$. The arrival probability at displacement $\ell_{\rm pr}$ satisfies
\begin{equation}
 p_{\rm pr}:=\left|\bra{\ell_{\rm pr}}
       e^{-i\theta_{\rm hop}A_I}\ket0\right|^2
 \ge p_*:=\frac14
       \left(\frac{\theta_{\rm hop}}{2\ell_{\rm pr}}\right)^{2\ell_{\rm pr}}.
 \label{eq:lb-single-signal}
\end{equation}
Let $A_{I^D}$ be the nearest-neighbor adjacency matrix of the Cartesian cube $I^D$. For an excitation starting at its center $\mathbf 0$, the evolution $e^{-i\theta_{\rm hop}A_{I^D}}$ has the same probability $p_{\rm pr}$ of reaching the receiver slice $\{x\in I^D:x_1=\ell_{\rm pr}\}$.
\end{lemma}
\end{samepage}
\begin{proof}
The finite-interval bound follows by comparing with the explicit infinite-chain amplitude.
Because that amplitude can be exponentially small, the boundary correction must be bounded relative to it, not merely by a small constant.
Let $\mathrm J_j(z)$ be the integer-order Bessel function of the first kind, with power series and differential equation as in~\cite[Eqs.~10.2.1 and~10.2.2]{DLMF}.
For the infinite-chain adjacency matrix $A_{\mathbb Z}$, Fourier transformation gives
\begin{equation}
 \bra{\ell_{\rm pr}}e^{-i\theta_{\rm hop}A_{\mathbb Z}}\ket0
 =(-i)^{\ell_{\rm pr}}\mathrm{J}_{\ell_{\rm pr}}(2\theta_{\rm hop}).
 \label{eq:lb-walk-bessel}
\end{equation}
This identity uses the integral representation in~\cite[Eq.~10.9.2]{DLMF}.
Positive-argument comparisons for Bessel functions are classical~\cite{Paris84}; see also~\cite[Eq.~10.14.7]{DLMF}.
The coarse inequality sufficient here follows directly from their differential equation:
\begin{equation}
 \mathrm{J}_{\ell_{\rm pr}}(2\theta_{\rm hop})
 \ge\left(\frac{\theta_{\rm hop}}{2\ell_{\rm pr}}\right)^{\ell_{\rm pr}}
 \qquad(\ell_{\rm pr}\ge2\theta_{\rm hop}).
 \label{eq:lb-bessel}
\end{equation}
We set $j:=\ell_{\rm pr}$ and define
$f(z):=j!(2/z)^j\mathrm{J}_j(z)$, extended by $f(0):=1$.
The differential equation and the expansion at zero give
\[
 f''+\frac{2j+1}{z}f'+f=0,\qquad
 q(z):=-\frac{f'(z)}{f(z)}=\frac{z}{2j+2}+\mathcal O(z^3).
\]
Before the first positive zero of $f$, its logarithmic derivative satisfies $q'=q^2-(2j+1)q/z+1$.
The expansion at zero places $q$ below $z/j$.
If it first crossed upward in $0<z\le j$, the derivative of $q-z/j$ at the crossing would be $z^2/j^2-1-2/j<0$, contradicting that crossing.
Thus $q(z)<z/j$ while $f$ stays positive.
Integrating $-(\log f)'=q$ gives $f(z)\ge e^{-z^2/(2j)}$, which in turn excludes a first zero in $[0,j]$.
The inequality therefore holds throughout that interval, and at $z=2\theta_{\rm hop}$ yields
\begin{align}
 \mathrm{J}_j(2\theta_{\rm hop})
 &\ge \frac{\theta_{\rm hop}^j}{j!}e^{-2\theta_{\rm hop}^2/j}
 \ge \left(\frac{\theta_{\rm hop}}{j}\right)^j e^{-j/2}
 \ge \left(\frac{\theta_{\rm hop}}{2j}\right)^j.
 \label{eq:lb-bessel-comparison}
\end{align}
The second step uses $j!\le j^j$ and $2\theta_{\rm hop}\le j$; the last uses $e^{-1/2}>1/2$.

To compare this amplitude with the finite walk, we extend $A_I$ by zero outside $I$ on $\ell_2(\mathbb Z)$.
Both adjacency operators have norm at most two, and their powers on $\ket0$ agree below $b_{\rm pr}$ because such walks cannot reach the boundary.
Lemma~\ref{lem:series-locality} consequently gives
\begin{equation}
 \left\|\bigl(e^{-i\theta_{\rm hop}A_I}
             -e^{-i\theta_{\rm hop}A_{\mathbb Z}}\bigr)\ket0\right\|
 \le2F_{\rm tail}(b_{\rm pr},2\theta_{\rm hop}).
 \label{eq:lb-finite-interval}
\end{equation}
A fixed, sufficiently large buffer ratio makes this error less than half the infinite-chain lower bound.
Indeed, when $b_{\rm pr}=A\ell_{\rm pr}$ for an integer $A\ge A_1$, Eq.~\eqref{eq:lb-factorial-tail} bounds the ratio of the vector error to the amplitude in Eq.~\eqref{eq:lb-bessel} by
\[
 4\left[2\left(\frac{2e}{A}\right)^A
       \left(\frac{\theta_{\rm hop}}{\ell_{\rm pr}}\right)^{A-1}
       \right]^{\ell_{\rm pr}}.
\]
Since $\theta_{\rm hop}/\ell_{\rm pr}\le1/2$ and $\ell_{\rm pr}\ge2$, this ratio is uniformly below $1/2$ for a sufficiently large fixed $A$.
For a general $b_{\rm pr}\ge A_1\ell_{\rm pr}$, the same estimate follows with $A=\lfloor b_{\rm pr}/\ell_{\rm pr}\rfloor$ by monotonicity of the tail in its threshold.
The finite amplitude is then at least half the lower bound in Eq.~\eqref{eq:lb-bessel}; squaring proves Eq.~\eqref{eq:lb-single-signal}.

The higher-dimensional statement needs no new propagation estimate.
The cube adjacency is a sum of commuting one-coordinate adjacencies, so the propagator factorizes on the central source.
Summing over all transverse coordinates therefore gives
\[
 \begin{gathered}
 \sum_{x_2,\ldots,x_D\in I}
 \left|\bra{(\ell_{\rm pr},x_2,\ldots,x_D)}
 e^{-i\theta_{\rm hop}A_{I^D}}\ket{\mathbf 0}\right|^2\\
 =p_{\rm pr}\prod_{r=2}^{D}
 \sum_{x_r\in I}\left|\bra{x_r}e^{-i\theta_{\rm hop}A_I}\ket0\right|^2
 =p_{\rm pr}.
 \end{gathered}
\]
Each factor in the product equals one by unitarity; when $D=1$, the product is empty and equals one.
\end{proof}

\subsection{Amplifying the signal with separated probes}\label{sec:lb-packing}
We now apply Lemma~\ref{lem:separated-probes} to XXZ dynamics. Packing many single-excitation walks gives the system-size dependence. To preserve their combined signal in the interacting model, the buffer around each probe must make its boundary error a small fraction of its arrival probability.

\begin{samepage}
\begin{proposition}[Error bound from separated probes]\label{prop:lb-certificate}
There are constants $\alpha,c_{\rm sig}>0$ and $\ell_0$ such that, for every integer propagation distance satisfying
\begin{equation}
 2\theta_{\rm hop}\le\ell_{\rm pr}\le\alpha\ell_{\rm sys},
 \qquad \theta_{\rm hop}\ge1,\qquad \ell_{\rm sys}\ge\ell_0,
 \label{eq:lb-probe-range}
\end{equation}
a simulator with $a_{\rm cone}<\ell_{\rm pr}$ must have error
\begin{equation}
 \epsilon\ge c_{\rm sig}\min\!\left\{1,
 \frac{n}{\ell_{\rm pr}^D}
 \left(\frac{\theta_{\rm hop}}{2\ell_{\rm pr}}\right)^{2\ell_{\rm pr}}
 \right\}.
 \label{eq:lb-error-certificate}
\end{equation}
\end{proposition}
\end{samepage}
\begin{proof}
We verify the probe lemma for every subset of a packed family.
The buffer ratio will be fixed independently of the instance so that the error contributed by each selected cube is a small fraction of its arrival probability.

\paragraph{Reference probes.}
We take an integer buffer ratio $A_{\rm buf}\ge A_1$ and set $b_{\rm pr}:=A_{\rm buf}\ell_{\rm pr}$. The error comparison below will fix the value of $A_{\rm buf}$.
A regular grid in the middle half of the lattice gives disjoint cubes of side $2b_{\rm pr}+1$, separated by at least one unused layer.
For $\alpha$ sufficiently small relative to $A_{\rm buf}$, their number satisfies
\begin{equation}
 N_{\rm pr}^{\max}\ge c_{\rm pack}\frac{n}{\ell_{\rm pr}^D},
 \label{eq:lb-packing}
\end{equation}
where $c_{\rm pack}>0$ is constant. Rounding a grid of spacing $2b_{\rm pr}+2$ in a central box proves this estimate after increasing $\ell_0$.

Each selected cube has a source at its center and a receiver consisting of the transverse slice at first-coordinate displacement $\ell_{\rm pr}$.
The separation from its own source is $\ell_{\rm pr}$.
For another cube, the centers differ by at least $2b_{\rm pr}+2$ in some coordinate, while the receiver extends by at most $b_{\rm pr}$ transversely and only $\ell_{\rm pr}$ longitudinally.
Thus every source is at least $\ell_{\rm pr}$ from every receiver. Denote the selected centers by $S$ and the receiver union by $B$; the construction applies to any subset of the packing.

Let $\mathcal E_{\rm cut}$ denote the edges crossing a selected cube boundary. We define the reference Hamiltonian by
\begin{equation}
 H_{\rm ref}:=H_{\rm XXZ}-E_{\rm vac}I
       -\sum_{\langle x,y\rangle\in\mathcal E_{\rm cut}}v_{xy}.
 \label{eq:lb-reference}
\end{equation}
The one-body terms of Eq.~\eqref{eq:lb-rewrite}, including $B_{\partial}$, are retained in this definition.
Under the reference dynamics, each selected cube contains exactly one excitation and the exterior stays in the vacuum.
The density interactions $\nu_x\nu_y$ vanish on this invariant subspace.
Moreover, the cubes are strictly inside the physical boundary, so $B_{\partial}$ vanishes there and the uniform one-body term contributes only a phase.
Thus the reference state factorizes into one-excitation hopping states, with independent receiver events of probability $p_{\rm pr}\ge p_*$ by Lemma~\ref{lem:lb-arrival}.
Figure~\ref{fig:lower-probes} shows the probes and the removed edges.

\begin{figure}[tbp]
 \centering
 \includegraphics[width=.90\linewidth]{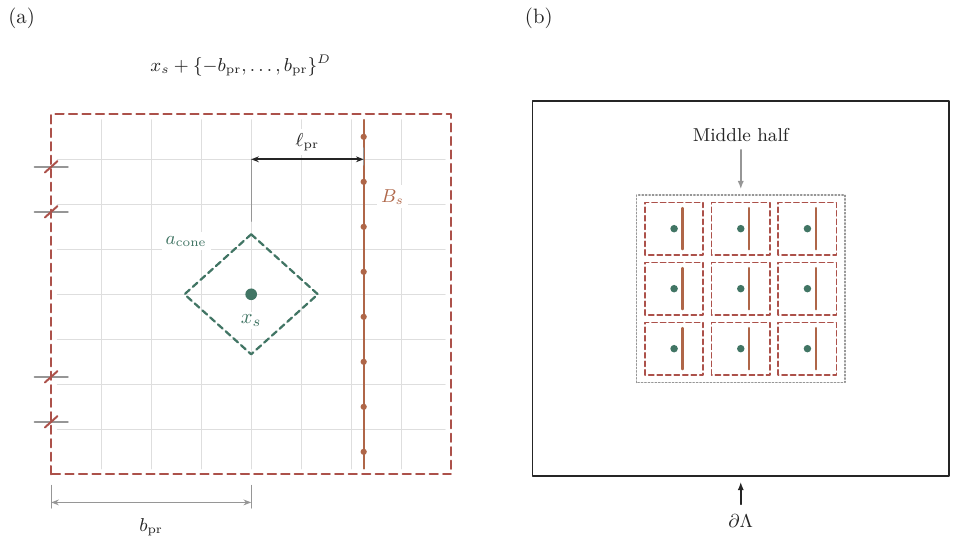}
 \caption{Separated probes in two dimensions. (a) A source is centered in a cube of radius $b_{\rm pr}$; its receiver is a slice at distance $\ell_{\rm pr}>a_{\rm cone}$. The reference deletes the crossing $v_{xy}$ terms and retains the one-body terms. (b) Probes lie inside the physical boundary. Their reference output factorizes, and Eq.~\eqref{eq:lb-cut-error} controls the difference from the interacting evolution. The fixed ratio $b_{\rm pr}/\ell_{\rm pr}=A_{\rm buf}$ makes this error small relative to the arrival signal. Distances are schematic.}
\label{fig:lower-probes}
\end{figure}

\paragraph{Error caused by the cuts.}
The cut error is determined by the reference states $\ket{\psi_{\rm ref}(u)}:=e^{-i\sigma uH_{\rm ref}}\ket S$, where $\sigma:=\operatorname{sgn}(T)$ and $0\le u\le|T|$.
An interior endpoint of a cut edge is at least $b_{\rm pr}$ hopping steps from the cube center.
Because the one-particle adjacency norm is at most $2D$, the projection bound of Lemma~\ref{lem:series-locality} makes its occupation at most $F_{\rm tail}(b_{\rm pr},2D J_{\rm hop}u)^2$.
The exterior endpoint has zero occupation.
The following two-qubit inequality converts these occupation bounds into a bound on the cut interaction:
\begin{equation}
 v_{xy}^2\le J_{\rm hop}^2\max\{1,16\Delta_0^2\}
                     (\nu_x+\nu_y).
 \label{eq:lb-edge-square}
\end{equation}
Taking the expectation of this inequality in the reference state and then its square root gives
\begin{equation}
 \begin{aligned}
 \|v_{xy}\ket{\psi_{\rm ref}(u)}\|
 &\le J_{\rm hop}\max\{1,4\Delta_0\}
 \sqrt{\bra{\psi_{\rm ref}(u)}(\nu_x+\nu_y)
                   \ket{\psi_{\rm ref}(u)}}\\
 &\le J_{\rm hop}\max\{1,4\Delta_0\}
             F_{\rm tail}(b_{\rm pr},2D J_{\rm hop}u).
 \end{aligned}
 \label{eq:lb-cut-action}
\end{equation}
There are $2D(2b_{\rm pr}+1)^{D-1}$ cut edges per cube.
Using $2b_{\rm pr}+1\le3b_{\rm pr}$ and integrating the monotone tail up to $|T|$ leads to the per-cube bound
\begin{equation}
 C_{\rm cut}:=2D\,3^{D-1}\max\{1,4\Delta_0\},\qquad
 \eta_{\rm cut}:=C_{\rm cut}\theta_{\rm hop}b_{\rm pr}^{D-1}
 F_{\rm tail}(b_{\rm pr},2D\theta_{\rm hop}).
 \label{eq:lb-cut-majorant}
\end{equation}
For the complete comparison, we denote the vacuum-shifted Hamiltonian by $H_\circ:=H_{\rm XXZ}-E_{\rm vac}I$ and the sum of cut interactions by $V_{\rm cut}:=\sum_{\langle x,y\rangle\in\mathcal E_{\rm cut}}v_{xy}$.
Since $H_\circ-H_{\rm ref}=V_{\rm cut}$, the state-dependent estimate for $\sigma H_\circ$ and $\sigma H_{\rm ref}$ gives
\begin{align}
 \delta_{\rm cut}
 &:=\left\|\left(e^{-iT(H_{\rm XXZ}-E_{\rm vac}I)}
                  -e^{-iTH_{\rm ref}}\right)\ket S\right\|\notag\\
 &\le\int_0^{|T|}\sum_{\langle x,y\rangle\in\mathcal E_{\rm cut}}
       \bigl\|v_{xy}e^{-i\sigma uH_{\rm ref}}\ket S\bigr\|\,du
 \le N_{\rm pr}\eta_{\rm cut}.
 \label{eq:lb-cut-error}
\end{align}
Unitarity removes the interacting propagator from the integrand. The bound therefore uses only the known boundary occupation of the reference state and applies without assuming independent probes in the true dynamics.

\paragraph{A uniform buffer and a detectable event.}
We choose the buffer so that $\eta_{\rm cut}\le p_*/100$ uniformly. Write $\phi_{\rm pr}=\ell_{\rm pr}\log(2\ell_{\rm pr}/\theta_{\rm hop})$.
For $A_{\rm buf}\ge4eD$, Eq.~\eqref{eq:lb-factorial-tail} implies
\[
 F_{\rm tail}(b_{\rm pr},2D\theta_{\rm hop})
 \le2e^{-A_{\rm buf}\phi_{\rm pr}}.
\]
The signal has $p_*=e^{-2\phi_{\rm pr}}/4$, whereas the cut tail gains the larger exponent $A_{\rm buf}\phi_{\rm pr}$.
Since $\theta_{\rm hop}\le\ell_{\rm pr}/2$, their ratio is therefore bounded, for a constant $C$ depending only on $D,\Delta_0$, by
\begin{equation}
 \frac{\eta_{\rm cut}}{p_*}
 \le C A_{\rm buf}^{D-1}\ell_{\rm pr}^{D}
              e^{-(A_{\rm buf}-2)\phi_{\rm pr}}
 \le C A_{\rm buf}^{D-1}\ell_{\rm pr}^{D}
              4^{-(A_{\rm buf}-2)\ell_{\rm pr}}.
 \label{eq:lb-buffer-comparison}
\end{equation}
For $(A_{\rm buf}-2)\log4>D/2$, the logarithmic derivative of $\ell^D4^{-(A_{\rm buf}-2)\ell}$ is negative on $\ell\ge2$.
Its maximum is consequently at $\ell=2$, where the remaining bound is a polynomial in $A_{\rm buf}$ times $4^{-2(A_{\rm buf}-2)}$.
Choosing $A_{\rm buf}$ large enough makes the ratio at most $1/100$ uniformly in $\ell_{\rm pr}$.
We fix this value first, and then choose $\alpha$, $c_{\rm pack}$, and $\ell_0$ for the packing.
Because $A_{\rm buf}\ge A_1$, the finite-interval arrival bound remains valid.

Lemma~\ref{lem:separated-probes} now applies with $K=N_{\rm pr}^{\max}$, $p_0=p_*$, and $\eta=\eta_{\rm cut}$, giving
\begin{equation}
 \epsilon\ge\frac18\min\{1,N_{\rm pr}^{\max}p_*\}.
 \label{eq:lb-probe-application}
\end{equation}
Substituting the packing count and $p_*$ proves Eq.~\eqref{eq:lb-error-certificate}, for example with $c_{\rm sig}:=\min\{1,c_{\rm pack}\}/100$.
The bound also covers packings larger than $p_*^{-1}$: the probe lemma uses only enough cubes to obtain a constant signal, keeping the comparison error controlled.
\end{proof}

\subsection{Time dependence and the final depth bound}\label{sec:lb-conclusion}
The probe estimate supplies the size and precision terms at distances beyond $2\theta_{\rm hop}$. We also need a time-dependent lower bound on the circuit cone. A single walk has a constant probability of traveling a distance proportional to time, as its second and fourth moments show.

\begin{samepage}
\begin{lemma}[Depth required by ballistic propagation]\label{lem:lb-ballistic}
There are constants $c>0$ and $\ell_0$ such that, for $\ell_{\rm sys}\ge\ell_0$ and $1\le\theta_{\rm hop}\le c\ell_{\rm sys}$, every simulator of $H_{\rm XXZ}$ with error $\epsilon\le1/128$ satisfies
\begin{equation}
 a_{\rm cone}\ge\theta_{\rm hop}.
 \label{eq:lb-time-term}
\end{equation}
\end{lemma}
\end{samepage}
\begin{proof}
On the infinite chain, let $Y_{\rm pos}$ be the displacement of one excitation after hopping time $\theta_{\rm hop}$.
The characteristic function of this continuous-time walk~\cite{Konno05} is
\begin{align}
 \mathbb E e^{i\varphi Y_{\rm pos}}
 &=\frac1{2\pi}\int_{-\pi}^{\pi}
 e^{2i\theta_{\rm hop}(\cos(q+\varphi)-\cos q)}\,dq\notag\\
 &=\mathrm{J}_0\bigl(4\theta_{\rm hop}\sin(\varphi/2)\bigr).
 \label{eq:lb-characteristic}
\end{align}
The Bessel integral representation gives the second equality~\cite[Eq.~10.9.1]{DLMF}.
Expanding at $\varphi=0$ then yields the moments
\begin{equation}
 \mathbb E Y_{\rm pos}^2=2\theta_{\rm hop}^2,\qquad
 \mathbb E Y_{\rm pos}^4=2\theta_{\rm hop}^2+6\theta_{\rm hop}^4.
 \label{eq:lb-moments}
\end{equation}
The second moment alone could be supported by a very rare large displacement.
The fourth moment excludes that possibility: applying Paley--Zygmund~\cite{PaleyZygmund32} to $Z:=Y_{\rm pos}^2$ gives
\begin{equation}
 \Pr(|Y_{\rm pos}|\ge\theta_{\rm hop})
 \ge\frac{(\mathbb EZ)^2}{4\mathbb EZ^2}
 =\frac{\theta_{\rm hop}^4}
        {2\theta_{\rm hop}^2+6\theta_{\rm hop}^4}
 \ge\frac18.
 \label{eq:lb-ballistic-probability}
\end{equation}
The inequality follows from
$\mathbb EZ\le\tfrac12\mathbb EZ+
(\mathbb EZ^2)^{1/2}\Pr(Z\ge\tfrac12\mathbb EZ)^{1/2}$,
by splitting the expectation and applying Cauchy--Schwarz. Coordinate factorization gives the same first-coordinate marginal on $\mathbb Z^D$.

To transfer this constant signal to the finite system, we place the source at a central site $o$.
In its one-excitation sector, subtracting the vacuum energy and the uniform one-body scalar gives the dimensionless Hamiltonian
\begin{equation}
 K_{\rm 1p}:=A_\Lambda+
 2\Delta_{\rm XXZ}\operatorname{diag}(2D-\deg_\Lambda(x)).
 \label{eq:lb-one-particle}
\end{equation}
Its norm and the norm of $A_{\mathbb Z^D}$ are at most
$C_{\rm 1p}:=2D+4D\Delta_0$.
We extend $K_{\rm 1p}$ by zero to $\ell_2(\mathbb Z^D)$. The powers of the two operators applied to $\ket o$ agree for every order below
$b_{\rm phys}:=\lfloor(\ell_{\rm sys}-1)/2\rfloor$,
because all shorter walks stay in the common interior. Lemma~\ref{lem:series-locality} therefore gives
\begin{equation}
 \left\|\bigl(e^{-i\theta_{\rm hop}K_{\rm 1p}}
             -e^{-i\theta_{\rm hop}A_{\mathbb Z^D}}\bigr)\ket o\right\|
 \le2F_{\rm tail}(b_{\rm phys},C_{\rm 1p}\theta_{\rm hop}).
 \label{eq:lb-physical-boundary}
\end{equation}
We choose $c\le1/(12eC_{\rm 1p})$ and $\ell_0\ge9$ to make the boundary error uniformly small.
These choices ensure $b_{\rm phys}\ge\ell_{\rm sys}/3$ and $C_{\rm 1p}\theta_{\rm hop}/b_{\rm phys}\le1/(4e)$ throughout the stated time range.
The factorial-tail estimate consequently makes the vector error at most $4\cdot4^{-b_{\rm phys}}\le1/32$.

Both states can now be tested with the same projector on $\ell_2(\mathbb Z^D)$, namely displacement $|x_1-o_1|\ge\theta_{\rm hop}$.
Since the finite state is supported on $\Lambda$, this projector is exactly its receiver event for
$B:=\{x\in\Lambda:|x_1-o_1|\ge\theta_{\rm hop}\}$.
Combining Eqs.~\eqref{eq:lb-ballistic-probability} and~\eqref{eq:lb-physical-boundary} gives a finite-system arrival probability of at least $1/8-1/32\ge1/16$.
If $a_{\rm cone}<\theta_{\rm hop}$, however, Lemma~\ref{lem:lb-causal} would bound it by $2\epsilon\le1/64$.
This contradiction proves the ballistic depth constraint.
\end{proof}

\begin{proof}[Proof of Theorem~\ref{thm:depth-lower-bound-formal}]
Choose the model-dependent time range and minimum side length so that both Proposition~\ref{prop:lb-certificate} and Lemma~\ref{lem:lb-ballistic} apply. Take $\epsilon_{\rm lb}\le1/128$ and $\epsilon_{\rm lb}<c_{\rm sig}$. Then $a_{\rm cone}\ge\theta_{\rm hop}\ge1$. We apply the probe estimate just beyond this cone unless its radius is already of the order of the system size.

If $a_{\rm cone}\ge\alpha\ell_{\rm sys}/4$, then $\dcnt=\Omega(\ell_{\rm sys})$, which proves the theorem because its stated minimum is at most $\ell_{\rm sys}$.
In the remaining case, we use the integer distance
\begin{equation}
 \ell_{\rm pr}:=\left\lceil\max\{2\theta_{\rm hop},a_{\rm cone}+1\}\right\rceil.
 \label{eq:lb-distance-choice}
\end{equation}
The choice obeys $a_{\rm cone}<\ell_{\rm pr}\le4a_{\rm cone}<\alpha\ell_{\rm sys}$ and $\ell_{\rm pr}\ge2\theta_{\rm hop}$, so Proposition~\ref{prop:lb-certificate} applies.
Since $\epsilon<c_{\rm sig}$, its minimum cannot be saturated.
Taking logarithms and writing $m_{\rm lb}:=\log(n/\epsilon)$ therefore gives
\begin{align}
 m_{\rm lb}
 &\le2\ell_{\rm pr}\log(2\ell_{\rm pr}/\theta_{\rm hop})
       +D\log\ell_{\rm pr}+\log(1/c_{\rm sig})\notag\\
 &\le C'\ell_{\rm pr}\log(2\ell_{\rm pr}/\theta_{\rm hop}).
 \label{eq:lb-log-certificate}
\end{align}
Here $\ell_{\rm pr}\ge2$ and $2\ell_{\rm pr}/\theta_{\rm hop}\ge4$, so a fixed $C'$ suffices.

If $\ell_{\rm pr}\ge m_{\rm lb}$, then
$\ell_{\rm pr}\ge m_{\rm lb}/\log(e+m_{\rm lb}/\theta_{\rm hop})$.
Otherwise, the logarithm in Eq.~\eqref{eq:lb-log-certificate} is at most
$2\log(e+m_{\rm lb}/\theta_{\rm hop})$.
In either case, $\ell_{\rm pr}\le4a_{\rm cone}$ gives
\[
 a_{\rm cone}\ge c_{\rm prec}
 \frac{m_{\rm lb}}{\log(e+m_{\rm lb}/\theta_{\rm hop})},
 \qquad c_{\rm prec}:=\min\left\{\frac14,\frac1{8C'}\right\}>0.
\]
The ballistic estimate $a_{\rm cone}\ge\theta_{\rm hop}$ holds for the same circuit.
Since $c_{\rm prec}\le1$, taking the maximum of these two bounds gives
\begin{equation}
 a_{\rm cone}\ge\frac{c_{\rm prec}}2
 \left[\theta_{\rm hop}+
 \frac{m_{\rm lb}}{\log(e+m_{\rm lb}/\theta_{\rm hop})}\right].
 \label{eq:lb-max-to-sum}
\end{equation}
Dividing by $\rho_{\rm gate}$ and combining with the earlier case
$a_{\rm cone}\ge\alpha\ell_{\rm sys}/4$ yields
\begin{equation}
 \dcnt\ge c'\min\!\left\{\ell_{\rm sys},\quad
 \theta_{\rm hop}+\frac{m_{\rm lb}}
 {\log(e+m_{\rm lb}/\theta_{\rm hop})}\right\}.
 \label{eq:lb-intermediate-bound}
\end{equation}
One may take $c':=\rho_{\rm gate}^{-1}\min\{\alpha/4,c_{\rm prec}/2\}>0$,
which covers both cases and depends only on the fixed model.

To return to the time normalization of the theorem, we use the lower time limit in Eq.~\eqref{eq:lb-regime}, which ensures $\mathcal T=\lceil J|T|\rceil$. Equation~\eqref{eq:lb-normalization} consequently gives
\begin{equation}
 \kappa_{\rm XXZ}\theta_{\rm hop}\le\mathcal T
 \le(\kappa_{\rm XXZ}+1)\theta_{\rm hop}.
 \label{eq:lb-time-comparison}
\end{equation}
Because both logarithms are at least one and differ by at most $\log(\kappa_{\rm XXZ}+1)$, replacing $\theta_{\rm hop}$ by $\mathcal T$ changes the bound only by constants depending on $D,\Delta_0$.
Substitution into Eq.~\eqref{eq:lb-intermediate-bound} therefore proves Eq.~\eqref{eq:depth-lower-bound}.
The choice $c_{\rm time}:=\kappa_{\rm XXZ}c$ expresses the permitted time range in the same normalization.
\end{proof}

\FloatBarrier
\section{Discussion}\label{sec:discussion}
The upper-bound proof uses locality both to bound the coefficient sum of the product-formula correction and to distribute its controls. On the blocks of Fact~\ref{fact:spatial}, a fixed formula order makes the residual small enough for exponential precision, while concurrent control distribution keeps the depth proportional to the block diameter. Together they prove Theorem~\ref{thm:block}; combining this local simulator with Fact~\ref{fact:spatial} gives Theorem~\ref{thm:global-formal}. Corollary~\ref{cor:local-product-select} also makes the selection circuit available to other LCUs of ordered local products.

Theorem~\ref{thm:depth-lower-bound} shows why the bounded typical propagation distance at fixed time does not imply bounded simulation depth at global accuracy. Many weak arrival events produce a detectable signal, and the state-dependent boundary estimate preserves it in the interacting XXZ model. The same argument applies whenever another model satisfies the relative comparison in Lemma~\ref{lem:separated-probes}.

When time grows with the system size, the comparison retains a different form.
With $m_{\epsilon}:=\log(n/\epsilon)$, we denote the worst-case simulation depth over the Hamiltonian class on the same cube by $\dcnt_{\rm worst}$.
Under the shared hypotheses, we obtain
\begin{equation}
 \Omega\!\left(\min\!\left\{\ell_{\rm sys},\,
 \mathcal T+\frac{m_{\epsilon}}{\log(e+m_{\epsilon}/\mathcal T)}\right\}\right)
 \ \le\ \dcnt_{\rm worst}
 \ \le\ O\!\left(\mathcal T\log(2+n\mathcal T/\epsilon)\right).
 \label{eq:joint-depth-comparison}
\end{equation}
Here the common regime requires $\ell_{\rm sys}\ge C\log(2+n\mathcal T/\epsilon)$ and the time and error conditions of Eq.~\eqref{eq:lb-regime}, with sufficiently large fixed $C$.

A remaining question is the joint dependence on time and precision in Eq.~\eqref{eq:joint-depth-comparison}.
The fixed-time gap is doubly logarithmic, but as time grows the lower bound adds the two scales while the upper bound multiplies time by a logarithmic overhead.
A related question is which algebraic structures allow fast-forwarding~\cite{AtiaAharonov17,GuSomma21} while retaining an efficient geometric implementation.
For this construction, finite gate sets add the synthesis cost of Appendix~\ref{app:gates}; short lattice sides and time-dependent interactions would require further control-placement and residual-representation arguments, respectively.

\subsection*{AI Usage}
The core ideas and overall approach of this work were developed by the authors. The proofs and manuscript were developed through multiple rounds of interaction between the authors and GPT-5.6 and GPT-6, which assisted with calculations, proof development, and writing. The authors take full responsibility for the correctness of the proofs and conclusions presented in this paper. The assistance concerned the development and presentation of proofs throughout Sections 2–6 and Appendices A–D.

\appendix
\section{Control distribution on a lattice}\label{app:geometry}
This appendix proves the control-distribution bound in Lemma~\ref{lem:geometry}, used in the upper-bound circuit of Section~\ref{sec:local}. For $D\ge2$, we construct a bounded-density tree with equal root-to-leaf delays and separate return paths. The geometry, timing, and return circuit are given below. The containing cube has volume $v=\ell^D$, within a constant factor of the original block volume.

\paragraph{A hierarchical lattice layout.}
Pad the block to side $\ell_p=b2^K$, with the smallest $K\ge0$ for which $\ell_p\ge\ell$. Since $b\le\ell$, this increases the side length by a factor of less than two; added sites have zero interaction and count toward the resources. The parameters from Section~\ref{sec:geometry} satisfy $b=2s$ and $b^D=\Theta(w)$, leaving room for the control registers and endpoint routing.

Place each data coordinate $\mathbf x$ at coarse coordinate $2\mathbf x$, making space for auxiliary paths. Multiple qubits at one coarse site are separated into physical sites by the fixed refinement in Eq.~\eqref{eq:microcell}. The smallest cubes are the subblocks, each of volume $v_0=b^D$.

A dyadic cube is indexed by $\nu:=(k,\mathbf m)$, where $0\le k\le K$ is its level and $0\le m_j<2^{K-k}$ are its coordinate indices.
Its original data coordinates are
\[
 \prod_{j=1}^D[m_jb2^k,(m_j+1)b2^k),
\]
of side $b2^k$: level zero consists of the leaf subblocks, and level $K$ is the whole padded cube. Its center in the doubled coordinates and its auxiliary workspace are
\begin{equation}
 c_j(\nu):=(2m_j+1)b2^k,\qquad
 B_\nu:=\mathbf c(\nu)+\{0,\ldots,s-1\}^D.
 \label{eq:nodecoords}
\end{equation}
A node at level $k\ge1$ has children $(k-1,2\mathbf m+\boldsymbol\epsilon)$ indexed by $\boldsymbol\epsilon\in\{0,1\}^D$. The child centers therefore satisfy
\begin{equation}
 \mathbf c(\mathrm{child})-\mathbf c(\nu)
 =b2^{k-1}(2\boldsymbol\epsilon-\mathbf1).
 \label{eq:childcoords}
\end{equation}
Here $\mathbf1$ is the all-one vector. Centers differ by at least $b$ in one coordinate, while workspaces have side $b/2$, so the workspaces are disjoint across levels. Each has $\Theta(w)$ sites ordered along a snake, with room for the endpoint words and routing. A constant number of registers share each site. Lemma~\ref{lem:path-permutation} aligns an arriving word in depth $\mathcal O(w)$ using $\mathcal O(w^2)$ gates. Figure~\ref{fig:hierarchy} shows this layout.

\begin{figure}[tbp]
\centering
\includegraphics[width=.90\linewidth]{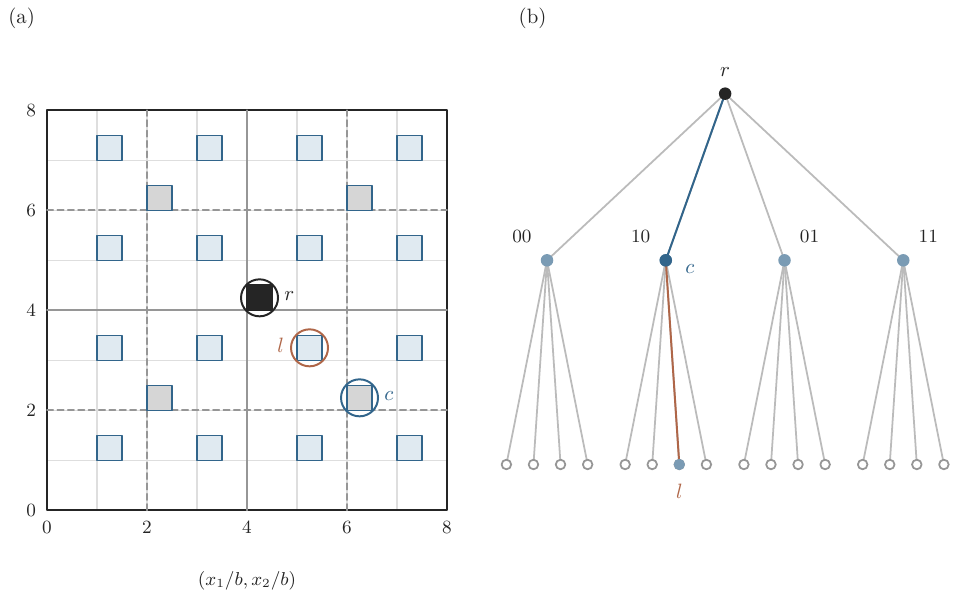}
\caption{Hierarchical control distribution for $D=2$, $K=2$. (a) Coordinates are in units of $b$ after doubling the data coordinates. The dark, gray, and pale blue squares are the disjoint workspaces of levels $2,1,0$. (b) The tree has child indices $(k-1,2\mathbf m+\boldsymbol\epsilon)$, with $\boldsymbol\epsilon\in\{0,1\}^2$. The labels $r,c,l$ identify the same root, child, and leaf in both panels; their coarse centers are $(4b,4b)$, $(6b,2b)$, and $(5b,3b)$.}
\label{fig:hierarchy}
\end{figure}

An internal node receives a word in $\reg I_\nu$ and distributes it into child words $\reg C_{\nu,\epsilon}$, one for each $\epsilon\in\{0,1\}^D$.
Separate words handle the reverse direction, so return traffic does not occupy the forward registers.
At the root, the original-slot array supplies the input.
With $0$ denoting the all-zero child index, distribution is implemented by
\begin{equation}
 \SWAP(\reg I_\nu,\reg C_{\nu,0});\qquad
 \CX(\reg C_{\nu,0};\reg C_{\nu,\epsilon})
 \quad(\epsilon\ne0).
 \label{eq:split}
\end{equation}
These gates move a basis word into the children as $\ket y^{\otimes2^D}$ and empty the input. Since local uses preserve these values, the inverse combines equal returned copies and clears all but one. Alignment within a workspace costs depth $\mathcal O(w)$ and $\mathcal O(w^2)$ gates; emptied positions are reusable in the next period.
Figure~\ref{fig:dyadic} relates this bitwise circuit to the geometric paths.

\begin{figure}[tbp]
\centering
\includegraphics[width=.90\linewidth]{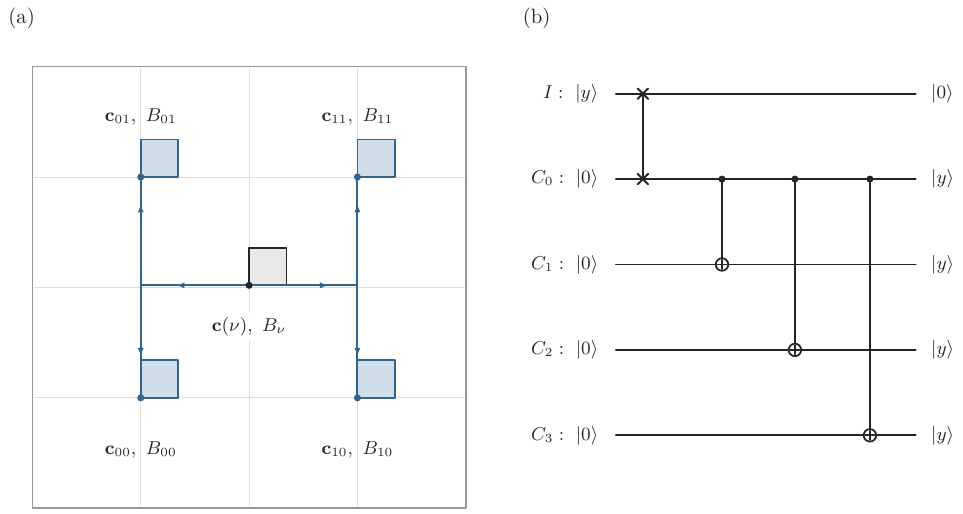}
\caption{Distribution to four children. (a) Each base path changes the first coordinate and then the second, giving length $m_k=Db2^{k-1}$. Coarse path overlaps use distinct qubits. (b) A SWAP moves the input to $C_0$; three CNOTs create correlated basis copies in the children and empty the input. The inverse combines the returned copies. Labels $0,1,2,3$ abbreviate the two-bit child indices.}
\label{fig:dyadic}
\end{figure}

\paragraph{Routing between levels.}
A parent--child connection changes coordinates $1,\ldots,D$ in order, using the signed displacements of Eq.~\eqref{eq:childcoords}.
This gives a coarse nearest-neighbor path of center-to-center length $m_k$.
For word transport, its length must be a multiple of $w$ and must accommodate endpoint segments, so we extend it to a half-walk of length $d_kw$, where
\begin{equation}
 m_k:=Db2^{k-1},\qquad d_k:=\lceil m_k/w\rceil+4.
 \label{eq:edgedelay}
\end{equation}
The additive four provides the endpoint segments and padding. One period advances a word by $w$ bit positions, so $d_k$ is the connection delay. Forward and reverse traffic use separate arrays.
Let $e$ denote a directed connection.
We implement its transport by a cycle of $2d_kw$ bit positions, indexed modulo $2d_kw$, with one half-walk from source to destination and another closing the cycle.
This closure permits the same reversible shift at every position.
Each position has distinct qubits $\reg A_{e,q}$ for the record and $\reg B_{e,q}$ for the initially zero shift channel.
Positions $0,\ldots,w-1$ form the source word and positions $d_kw,\ldots,d_kw+w-1$ the destination word.

Locality requires consecutive array positions to lie on a closed coarse walk.
Our placement starts at the source center, follows its workspace snake for $w-1$ edges, and retraces those edges; the first $w$ positions consequently hold the source word.
To obtain the prescribed connection length, the walk includes a further out-and-back segment within this workspace, of length
\begin{equation}
 \rho_k:=d_kw-m_k-2(w-1),\qquad
 2w+2\le\rho_k\le3w+1.
 \label{eq:localpadding}
\end{equation}
The segment follows $\lfloor\rho_k/2\rfloor$ snake edges in each direction, with a same-site step between distinct qubits if $\rho_k$ is odd. It fits because $s^D\ge4w+8$. Adding the base path gives length $d_kw$; the analogous destination segment and reversed base path close the cycle. Repeated visits use distinct qubits, except for the final identification with the starting position.
Figure~\ref{fig:transport} shows how the endpoint segments, padding, and base paths fit together.

\begin{figure}[tbp]
\centering
\includegraphics[width=.90\linewidth]{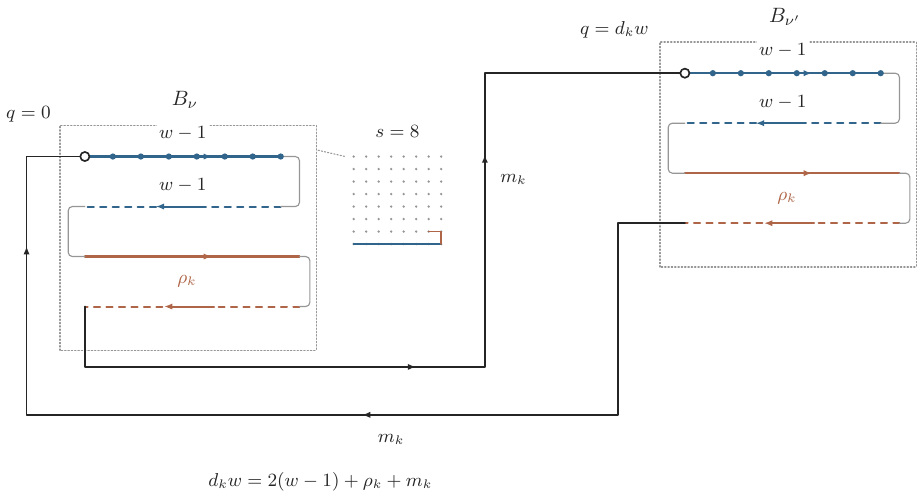}
\caption{Visit order in a closed transport array. Dotted boxes mark the endpoint workspaces. At each endpoint, the blue outward and retracing strokes contribute $2(w-1)$ steps, and the brown excursion contributes $\rho_k$. Thin gray joins add no step. The black base path and its reverse each have length $m_k$, so the full cycle has $2d_kw$ positions. The open circles mark the first bits of the endpoint words, at $q=0$ and $q=d_kw$. Repeated visits use distinct qubits $(\mathsf A_{e,q},\mathsf B_{e,q})$ even at the same coarse site. The inset projects the endpoint excursions onto the workspace snake. Here $D=2$, $w=s=8$, $b=16$, and $k=1$, giving $m_k=32$, $d_k=8$, and $\rho_k=18$. Equation~\eqref{eq:localpadding} also covers odd padding lengths. The drawing is not to scale.}
\label{fig:transport}
\end{figure}

One bit shift consists of the two matchings
\begin{equation}
 \prod_q\SWAP(\reg A_{e,q},\reg B_{e,q});\qquad
 \prod_q\SWAP(\reg B_{e,q},\reg A_{e,q+1}).
 \label{eq:cyclicshift}
\end{equation}
The matchings give $A'_q=A_{q-1}$ and $B'_q=B_{q+1}$ on arbitrary inputs. Thus a zero $B$ channel stays zero. Repeating the shift $w$ times advances one word in depth $\mathcal O(w)$, using $\mathcal O(w)$ gates per bit position.
Figure~\ref{fig:shift} shows the matchings.
A word inserted at zero reaches the destination after $d_k$ word shifts and is extracted there.
As no occupied word passes the destination, the return half stays empty and the source is zero before the next insertion, even while later records occupy the forward half.

\begin{figure}[tbp]
\centering
\includegraphics[width=.90\linewidth]{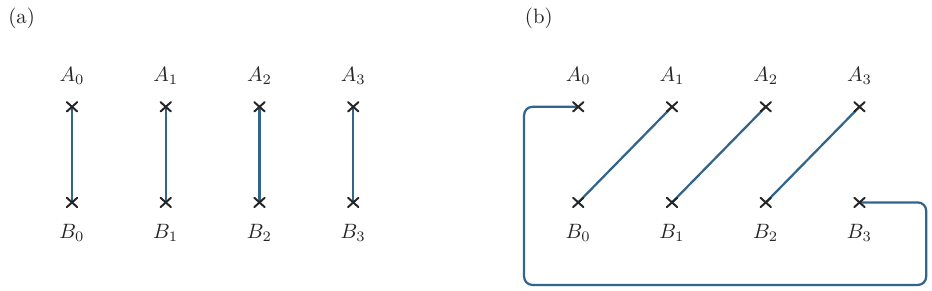}
\caption{A reversible cyclic bit shift from two SWAP matchings. Exchanging $A_q$ with $B_q$ and then $B_q$ with $A_{q+1}$ gives $A'_q=A_{q-1}$ and $B'_q=B_{q+1}$, with cyclic indices. Hence a zero $B$ array returns to zero while the $A$ contents advance. Repeating the operation $w$ times shifts a whole word, and the closed-walk placement keeps the wraparound SWAP local.}
\label{fig:shift}
\end{figure}

\begin{lemma}[Bounded physical density]\label{lem:density}
The node workspaces $B_\nu$, directed connection arrays $(\reg A_{e,q},\reg B_{e,q})$, and leaf registers $(\reg R_g,\reg W_g)$ above use $\mathcal O(v)$ qubits, with a constant number of auxiliary qubits per coarse site.
\end{lemma}
\begin{proof}
Bounded density is not implied by the total number of array qubits; we must also bound how many paths share a site.
On the direction-$j$ segment of a level-$k$ connection, coordinates before $j$ equal the child center, while coordinates after $j$ still equal the parent center.
By Eq.~\eqref{eq:nodecoords}, dividing a parent coordinate by $b$ gives an odd integer times $2^k$, and dividing a child coordinate gives an odd integer times $2^{k-1}$.
These powers distinguish paths at different levels.

The level can be recovered from a coordinate that stays fixed along a segment.
For a positive integer $a$, let $\nu_2(a)$ be the exponent of the largest power of two dividing $a$.
At a coarse site $\mathbf x$ on a segment in direction $j$, this gives
\[
 j=1:\quad k=\nu_2(x_2/b),\qquad
 j\ge2:\quad k=1+\nu_2(x_1/b).
\]
The displayed ratios are positive integers because they are center coordinates.
The use of $x_2$ when $j=1$ is the reason for assuming $D\ge2$.
Thus a fixed site and segment direction determine at most one possible level, preventing a logarithmic accumulation of path qubits across levels.

Within a level, every base path stays in the central interior of its parent cube, and distinct parent interiors are disjoint.
A parent has $2^D$ children, so at most $2^D$ paths in a given direction can pass through a site.
Summing over directions gives multiplicity at most $D2^D=\mathcal O(1)$, independently of the tree height.

Endpoint visits do not change this conclusion.
The workspaces $B_\nu$ are disjoint, each meets at most $2^D+1$ connections, and each endpoint extension traverses its workspace snake only a bounded number of times.
Accounting for both directions, both half-walks, and the two array qubits changes the bound only by a constant.
The node registers and leaf computations likewise have bounded occupancy.

Combining the path and endpoint counts therefore bounds the number of ancillary qubits at every coarse site independently of $K$.
In particular, geometric path overlaps require only constantly many distinct qubits at a site.

For the total number of qubits, level $k$ has $\mathcal O((\ell_p/(b2^k))^D)$ connections, each of size $\mathcal O(b2^k+w)$. The relation $\ell_p^D=\Theta(v)$ therefore gives
\begin{align}
 \sum_{k=1}^K O\left(\frac{v}{b^D2^{Dk}}(b2^k+w)\right)
 &=O\left(\frac{v}{b^{D-1}}\sum_{k\ge1}2^{-(D-1)k}
 +\frac{vw}{b^D}\sum_{k\ge1}2^{-Dk}\right)
 =\mathcal O(v).
 \label{eq:spacelevels}
\end{align}
Here $D\ge2$ and $b^D=\Theta(w)$. There are $\mathcal O(v/w)$ node and leaf workspaces, each using $\mathcal O(w)$ qubits, so these workspaces together also use $\mathcal O(v)$ qubits.
\end{proof}

\paragraph{Physical nearest-neighbor gates.}
The density bound permits a fixed ordering $a=0,1,\ldots$ of the qubits at each coarse site.
A constant $c_{\rm phys}$ with $c_{\rm phys}^D$ larger than their maximum number then leaves a distinct physical position for every qubit, given by
\begin{equation}
 \operatorname{pos}(\mathbf x,a):=c_{\rm phys}\mathbf x+\operatorname{digits}_{c_{\rm phys},D}(a).
 \label{eq:microcell}
\end{equation}
Here $\operatorname{digits}_{c_{\rm phys},D}(a)\in\{0,\ldots,c_{\rm phys}-1\}^D$ gives the $D$ base-$c_{\rm phys}$ digits, including leading zeros.
This is an injective placement in a constant-factor enlargement, and gates between equal or adjacent coarse sites have constant physical range.
A CNOT between the endpoints of a path $p_0,\ldots,p_d$ is implemented by moving the control to $p_{d-1}$ with SWAPs, applying the adjacent CNOT, and reversing those SWAPs.
 An exchange of the endpoints instead takes $2d-1$ adjacent SWAPs, each consisting of three CNOTs.
After the constant-size controlled gates are decomposed as in Appendix~\ref{app:gates}, the same routing applies to them.

The physical paths implementing a coarse layer have bounded radius.
Together with bounded qubit density, this implies that each path intersects only constantly many others.
A fixed greedy coloring consequently splits the layer into constantly many disjoint path families.
This costs only constant overhead.

\paragraph{Timing and operation order.}
Each period extracts arriving words, combines return copies, distributes forward copies, and applies the leaf operations. At all leaves, stage $k$ is completed before stage $k+1$. The words then enter their outgoing arrays and advance by one word shift. Separate forward and reverse registers allow simultaneous traffic, and empty records act trivially. The period has depth $\mathcal O(w)$.

Because the root-to-leaf delays agree, they can be represented by a single $T_\downarrow$.
Let $\delta_{\rm geo}$ also include the return trip and the leaf processing in its arrival period.
Measured in periods, the two delays are
\begin{equation}
 T_\downarrow:=\sum_{k=1}^K d_k,\qquad \delta_{\rm geo}:=2T_\downarrow.
 \label{eq:latency}
\end{equation}
Record $u$ enters the root at boundary $u$.
Every level-$k$ connection contributes exactly $d_k$ periods, independently of the chosen child.
Induction down and back up the tree therefore gives
\begin{equation}
 t_{\rm leaf}(u)=u+T_\downarrow,\qquad
 t_{\rm return}(u)=u+\delta_{\rm geo}.
 \label{eq:recordtimes}
\end{equation}
All leaves use a record in the same period, then pass it to the return array and empty their local register. Successive records occupy distinct positions. For $K=0$, the root is the only leaf and processes and returns a record in the same period, so both delays are zero.

\paragraph{Returning and clearing the control registers.}
The source must retain all $N$ records until their copies return after $\delta_{\rm geo}$ periods.
A buffer length $R_c:=N+\delta_{\rm geo}+3$ provides the required storage and three empty positions. The cyclic word arrays $\reg F[i]$ and $\reg Q[i]$, for $0\le i<R_c$, hold forward inputs and align returning copies, respectively, with position zero near the root.
For the bitwise shift circuit, each array consists of $R_cw$ positions $A_q$ and a second, initially zero channel $B_q$.
Executing the second matching of Eq.~\eqref{eq:cyclicshift} before the first implements the inverse shift, $A'_q=A_{q+1}$ and $B'_q=B_{q-1}$ modulo $R_cw$.
Repeating it $w$ times moves the words one position toward lower indices and restores the zero $B$ channel.
Thus a word shift has depth $\mathcal O(w)$ and uses $\mathcal O(R_cw^2)$ gates, while the two arrays occupy $\mathcal O(R_cw)$ qubits.

As with the connections, these source arrays need a local closed placement.
The walk starts at the root, follows its workspace snake for $w-1$ edges, and retraces them. An out-and-back coordinate segment supplies the remaining length up to $R_cw$, with one same-site step between distinct qubits when needed for parity.
Every visit except the closure uses a different qubit, so the first $w$ positions form the root word without identifying repeated visits.
The original records $\reg O_u$ and their retained $\mathcal O(w)$ work occupy separate qubits along the same folded path; they are never identified with $\reg F$ or $\reg Q$.
Substituting the buffer length and connection delays bounds this path by
\begin{align}
 R_cw
 &=Nw+2w\sum_{k=1}^K\bigl(\lceil Db2^{k-1}/w\rceil+4\bigr)+3w\notag\\
 &\le Nw+2Db(2^K-1)+10wK+3w\notag\\
 &=\mathcal O(Nw+\ell+wK)=\mathcal O(\ell).
 \label{eq:sourcepathlength}
\end{align}
Here $b2^K=\ell_p=\Theta(\ell)$ and $K=\mathcal O(\log\ell)$; the last equality uses $Nw+w\log\ell=\mathcal O(\ell)$, with $N\ge1$ absorbing the $\mathcal O(w)$ term.
The root workspace and an out-and-back segment of length $\mathcal O(\ell)$ therefore fit inside a constant-factor enlargement of the processing cube.
The snake has only a fixed number of visits, and the coordinate segment is traversed once in each direction, so no site acquires unboundedly many qubits.
The resulting path through the source and array registers has length $\mathcal O(R_cw)$.

To initialize, fixed path permutations align each retained field of $\reg O_u$ with its target in $\reg F[u]$. CNOTs create the copies, and inverse permutations restore the placement; other words remain zero. Lemma~\ref{lem:path-permutation} gives depth $\mathcal O(R_cw)$ and $\mathcal O((R_cw)^2)$ gates.

At each period boundary $t$, the root's forward operation takes the word from $\reg F[0]$, leaving that register empty, while the returning word enters the zero register $\reg Q[0]$.
Both arrays shift toward lower indices at the end of each period, for $t=0,\ldots,R_c-1$.
A record returning at boundary $u+\delta_{\rm geo}$ therefore undergoes $R_c-(u+\delta_{\rm geo})$ further shifts, which puts record $u$ at index
\begin{equation}
 0-[R_c-(u+\delta_{\rm geo})]
 \equiv u+\delta_{\rm geo}\pmod{R_c}.
 \label{eq:returnindex}
\end{equation}
By boundary $N-1+\delta_{\rm geo}<R_c$, all records have returned and both the tree and $\reg F$ are empty.
Equation~\eqref{eq:returnindex} gives the remaining displacement in $\reg Q$ explicitly.
Because it is fixed, one prescribed nearest-neighbor permutation aligns each returning field with its original in $\reg O_u$.
After alignment, CNOTs from the unchanged original fields clear the returning copies, and the inverse permutation restores the layout.

Thus alignment and clearing have
\begin{equation}
 \begin{aligned}
 \text{depth }\dcnt_{\rm align}&=\mathcal O(R_cw),\\
 \text{gate count }\gcnt_{\rm align}&=\mathcal O((R_cw)^2),\\
 \text{ancillas }\acnt_{\rm align}&=\mathcal O(R_cw),
 \end{aligned}
 \label{eq:alignmentresources}
\end{equation}

\begin{proof}[Proof of Lemma~\ref{lem:geometry}]
For each basis input, the circuit distributes correlated control values and preserves them during local use. Equation~\eqref{eq:recordtimes} synchronizes the leaves and orders the records; the stage schedule then gives Eq.~\eqref{eq:distributed-product}, including operations across subblock boundaries.

The depth is determined by the sum of connection delays along a root-to-leaf path.
Equation~\eqref{eq:edgedelay} gives
\begin{equation}
 T_\downarrow\le\sum_{k=1}^K(m_k/w+5)
 =D(\ell_p-b)/w+5K=\mathcal O(\ell/w+K).
 \label{eq:latencybound}
\end{equation}
The complete schedule has $R_c=N+2T_\downarrow+3$ periods, each of depth $\mathcal O(w)$ for local processing, alignments, and shifts.
Initialization and final alignment contribute the same order as these periods combined.
Hence the total depth is
\[
 \dcnt=\mathcal O(R_cw)
 =\mathcal O(Nw+\ell+wK)
 =\mathcal O(Nw+\ell+w\log\ell)
 =\mathcal O(\ell),
\]
where the last step uses the size assumption of Lemma~\ref{lem:geometry}.

For the gate count, each period shifts $\mathcal O(v)$ array positions by $w$ bits, costing $\mathcal O(vw)$ gates.
The $\mathcal O(v/w)$ node and leaf workspaces each use $\mathcal O(w^2)$ gates, so processing contributes another $\mathcal O(vw)$.
The source path has $\mathcal O(R_cw)=\mathcal O(\ell)\le \mathcal O(v)$ positions and obeys the same per-period bound.
Finally, the two full-path alignments cost $\mathcal O((R_cw)^2)$; because $R_cw=\mathcal O(\ell)\le \mathcal O(v)$, this is at most $\mathcal O(vR_cw)$.
Adding the contributions yields
\[
 \gcnt=\mathcal O(R_cvw+(R_cw)^2)
 =\mathcal O(vR_cw)=O\bigl(v(Nw+\ell+wK)\bigr)=\mathcal O(v\ell).
\]
The last equality uses the geometric size assumption.
All original records, arithmetic work, and source arrays fit along the path of length $\mathcal O(Nw+\ell+wK)=\mathcal O(\ell)$ from Eq.~\eqref{eq:sourcepathlength}, with constantly many qubits per site.
Together with Lemma~\ref{lem:density}, this gives $\mathcal O(v)$ ancillas.

\end{proof}

\FloatBarrier
\section{Control distribution in one dimension}\label{app:1d}
This appendix supplies the one-dimensional control-distribution circuit for the upper bound. Records reach different subblocks at different times. We separate each record into its insertion and product-formula stages, spacing them to preserve the order of operations on overlapping supports. Inverse transport then returns the unchanged controls for uncomputation.

\paragraph{Registers and their placement.}
Use the width $w$ from Section~\ref{sec:parameters} and pad with zero-interaction sites to the smallest interval of length $\ell=w2^{K_1}$ that covers the block. Its original length is at least $w$, so the padding factor is less than two.

The interval is partitioned into $G_{\rm blk}:=\ell/w$ subblocks.
For $0\le g<G_{\rm blk}$, subblock $\Omega_g$ contains the constant-width data registers $\reg D_x$ at $gw\le x<(g+1)w$.
As in Lemma~\ref{lem:residual}, each site is assigned $q_{\rm rec}$ residual entries, with zero coefficients providing padding. Each subblock consequently has $T_0:=q_{\rm rec}w=\Theta(w)$ entries.
The address $\lambda=(g,j,\mu)$ uses Eq.~\eqref{eq:address}, with $j$ now the left-to-right site index.
Each selector $\reg a_{g,i}$ is placed near its assigned site, including for padded entries; the remaining work is distributed along the subblock at constant density.
This clean local work is denoted by $\reg W_g$.

At each arrival, Lemma~\ref{lem:local} applies one insertion or formula stage in depth $\mathcal O(w)$ with $\mathcal O(w^2)$ gates, restoring the control fields and clearing local work before transport resumes.

\paragraph{Scheduling the controlled operations.}
To assign a transport time to every stage, we label the events of slot $u\in\{0,\ldots,N-1\}$ by
\begin{equation}
 f:=(L+1)u+\alpha,\qquad 0\le\alpha\le L,
 \qquad F:=(L+1)N.
 \label{eq:1devents}
\end{equation}
The event with $\alpha=0$ applies the insertion phases $\chi_\lambda,z_{\lambda,\beta}$ and the selected Pauli under $eX$.
For $\alpha=k\ge1$, it applies only $V_k(h)$ under $e(1-X)$.

Bounded support gives a fixed $\rho$ such that operations on intersecting declared supports are assigned to subblocks at index distance at most $\rho$. Choose the event spacing and buffer length, in words, as
\begin{equation}
 \sigma:=\rho+2,\qquad P:=\sigma(F-1)+1.
 \label{eq:1dbuffer}
\end{equation}
With this spacing, the time difference between successive events exceeds any arrival-time difference between overlapping supports.
The buffer holds all $F$ input events on the left and their completed words on the right; below we also verify the empty boundary needed for a finite reversible shift.

To accommodate both buffers, we extend the coarse line to $x_{\min}:=-Pw$ and $x_{\max}:=\ell+Pw-1$.
Each site has distinct transport qubits $\reg A_x,\reg B_x$ and a fixed number of work qubits, while data and selectors occupy $0\le x<\ell$.
The $A$ channel carries the event words and the $B$ channel makes their shift reversible.
The initial word for event $f$ occupies
\begin{equation}
 (\reg A_{(-\sigma f-1)w+j})_{0\le j<w}.
 \label{eq:1dplacement}
\end{equation}
All remaining $A$ qubits and all $B$ qubits are zero.
Event zero therefore starts immediately to the left of the data, with successive events $\sigma$ words apart.

The fixed values of $L$ and $\sigma$ place all copies of one slot within $\mathcal O(w)$ sites. Retained originals occupy separate qubits in the same interval. Fixed permutations and CNOTs initialize the $L+1$ copies in depth $\mathcal O(w)$ with $\mathcal O(w^2)$ gates per slot. The slots are initialized in parallel and retain their originals for final clearing.

\paragraph{Transport along the chain.}
The following SWAP matchings implement transport on the finite auxiliary line, including its endpoints:
\begin{equation}
 \prod_{x=x_{\min}}^{x_{\max}}\SWAP(\reg A_x,\reg B_x);\qquad
 \prod_{x=x_{\min}}^{x_{\max}-1}\SWAP(\reg B_x,\reg A_{x+1}).
 \label{eq:openshift}
\end{equation}

Tracking arbitrary input values gives
\begin{equation}
 \begin{aligned}
 A'_{x_{\min}}&=B_{x_{\min}},& A'_x&=A_{x-1}\quad(x>x_{\min}),\\
 B'_x&=B_{x+1}\quad(x<x_{\max}),& B'_{x_{\max}}&=A_{x_{\max}}.
 \end{aligned}
 \label{eq:openboundary}
\end{equation}
Consequently, if the $B$ channel and the rightmost $A$ bit are zero, the $A$ contents move right by one position and the $B$ channel returns to zero.
The chosen buffer guarantees this boundary condition throughout forward transport; its verification is part of the proof below.

A word step consists of $w$ bit shifts, followed by the local operations for words aligned with data subblocks.
The step counter $t$ starts at zero before transport and counts completed word steps.
Because event $f$ begins at word position $-\sigma f-1$, it reaches subblock $g$ at
\begin{equation}
 t(f,g):=\sigma f+g+1.
 \label{eq:1dtime}
\end{equation}
At each such time, transport pauses while the local circuits process exactly the pairs $(f,g)$ satisfying Eq.~\eqref{eq:1dtime}.
The aligned $A$ fields act as $\reg R_g$; after the insertion or stage, the local circuit restores them before transport resumes.
Lemma~\ref{lem:local}, including its fixed coloring of physical conflicts, bounds this processing by depth $\mathcal O(w)$.
The $w$ bit shifts have the same depth, so every word step costs $\mathcal O(w)$.

\begin{lemma}[One-dimensional record compilation]\label{prop:1d}
Consider $N\ge1$ original records $\reg O_u$ of width at most $w$ on a one-dimensional data block padded to length $\ell=w2^{K_1}$. Assume $w$ is a power of two sufficient for the local fields and work, the assigned residual entries have bounded density and bounded support, and the product formula has a fixed number $L$ of stages, each composed of disjoint bounded-support gates as in Lemma~\ref{lem:local}. With $F=(L+1)N$, suppose $Fw=\mathcal O(\ell)$. The schedule in Eqs.~\eqref{eq:1devents}--\eqref{eq:1dtime} implements the ordered data operator $W_z$ and the address-dependent diagonal phases prescribed by the original records. Its resources are
\begin{equation}
 \begin{aligned}
 \text{depth }\dcnt_{1}&=\mathcal O(\ell+Fw)=\mathcal O(\ell),\\
 \text{gate count }\gcnt_{1}&=O\bigl(\ell(\ell+Fw)\bigr)=\mathcal O(\ell^2),\\
 \text{ancillas }\acnt_{1}&=\mathcal O(\ell+Fw)=\mathcal O(\ell).
 \end{aligned}
 \label{eq:1dresources}
\end{equation}

\end{lemma}
\begin{proof}
To verify the data order, we compare two events $f<f'$ whose operations in subblocks $g,g'$ have intersecting data supports.
Their assigned subblocks satisfy $|g-g'|\le\rho$, so their arrival times obey
\begin{equation}
 t(f',g')-t(f,g)
 =\sigma(f'-f)+g'-g
 \ge\sigma-\rho\ge2.
 \label{eq:1dorder}
\end{equation}
Thus overlapping operations retain their order, while operations on disjoint supports commute.

Contributions within a single event commute as well.
For a product-formula stage this follows from disjoint supports; for an insertion, only one residual entry is selected on each basis branch, so every other contribution is the identity.
The physical schedule therefore has the same data action and phase as completing the events in increasing order $f=0,\ldots,F-1$.
By Eq.~\eqref{eq:1devents}, that order is the insertion followed by $V_1,\ldots,V_L$ in each slot, exactly as in Eq.~\eqref{eq:recordword}.
A constant coloring of the bounded restoration paths resolves any simultaneous physical conflicts without changing these completed operations.

Next, the final event reaches the last subblock at
\begin{equation}
 T_f:=t(F-1,G_{\rm blk}-1)=\sigma(F-1)+G_{\rm blk}=P+G_{\rm blk}-1.
 \label{eq:1dfinaltime}
\end{equation}
The initially rightmost occupied word is event zero, at position $-1$.
After $T_f$ word shifts, it is at $P+G_{\rm blk}-2$, whereas the last word of the extended chain is at $P+G_{\rm blk}-1$.
An entire zero word therefore remains to its right even at the final time, and occupied positions are farther left at every earlier bit shift.
Thus the rightmost $A$ bit is always zero when Eq.~\eqref{eq:openshift} is applied.
Equation~\eqref{eq:openboundary} now proves inductively that the whole $B$ channel is restored to zero after every forward shift.

After all events, reversing only the shifts returns the preserved control values to their initial positions. Reversing initialization clears the copies against the originals without undoing data operations or phases.

For the resources, the extended chain contains $\ell+2Pw=\mathcal O(\ell+Fw)$ coarse sites with constant local density.
There are $T_f=\mathcal O(F+\ell/w)$ forward steps, each of depth $\mathcal O(w)$, and reversing their shifts has the same asymptotic cost.
Initialization and clearing add only $\mathcal O(w)$ depth, giving total depth $\mathcal O(\ell+Fw)$.

Each bit shift uses $\mathcal O(\ell+Fw)$ swaps, so forward and inverse transport together require
\[
 O\bigl((F+\ell/w)w(\ell+Fw)\bigr)=O\bigl((\ell+Fw)^2\bigr)
\]
gates.
The $FG_{\rm blk}$ local uses contribute $\mathcal O(\ell Fw)$ gates, and copying and clearing contribute $\mathcal O(Nw^2)$.
Under $Fw=\mathcal O(\ell)$ and $w\le\ell$, their sum is $\mathcal O(\ell(\ell+Fw))$.
Constant local density gives $\mathcal O(\ell+Fw)$ ancillas.

\end{proof}

\FloatBarrier
\section{Elementary gates, arithmetic, and gate synthesis}\label{app:gates}
\subsection{Controlled local operations}
Let $R_z(\theta):=e^{-i\theta Z/2}$ for real $\theta$.
A controlled rotation must preserve its basis control $c$, applying the rotation to target $s$ only when $c=1$.
The following physical gate order achieves this:
\begin{equation}
 R_z(\theta/2)_s\ ;\ \CX(c;s)\ ;\ R_z(-\theta/2)_s\ ;\ \CX(c;s).
 \label{eq:crz}
\end{equation}
For $c=0$, the two rotations cancel.
For $c=1$, $XZX=-Z$ makes their time-ordered product $XR_z(-\theta/2)XR_z(\theta/2)=R_z(\theta)$.
This identity also implements the rotation in the phase-average gate, whose decomposition is
\[
 \Dph(z)=e^{i\phi}R_z(-2\theta),\qquad
 \phi=\arg z,\quad \theta=\arccos|z|.
\]
Its controlled version therefore combines a controlled $R_z(-2\theta)$ with $\operatorname{diag}(1,e^{i\phi})$ on the control.
The latter is essential: the scalar $e^{i\phi}$ is a relative phase between the selected and unselected branches, not a disposable global phase.

Through Eq.~\eqref{eq:phaseaverage}, this implements the source factors $f(J_u)$ and $\gamma(\widetilde b_u)$ and the address-dependent factors $\chi_\lambda$ and $z_{\lambda,\beta}$.

The Hadamard gate is denoted by $H_{\rm Had}$, and $S_{\rm ph}:=\operatorname{diag}(1,i)$ is the phase gate.
Conjugating CNOT gives controlled $Y$ in the order $S_{\rm ph}^\dagger;\CX;S_{\rm ph}$ and controlled $Z$ in the order $H_{\rm Had};\CX;H_{\rm Had}$.
For a general Pauli rotation $e^{-i\theta P}$, we use the tensor-product basis change $F_P$ with factors $H_{\rm Had}$ for $X$, $H_{\rm Had}S_{\rm ph}^\dagger$ for $Y$, and $I$ for $Z$.
The conjugated operator $F_PPF_P^\dagger$ is a product of $Z$ operators.
After applying $F_P$, CNOTs accumulate their parity in one data qubit; Eq.~\eqref{eq:crz} with angle $2\theta$ applies the controlled phase, and reversing the parity computation and basis change restores the original coordinates.
This gives the required controlled rotation while preserving the control.
For $P=I$, the same operation is simply $\operatorname{diag}(1,e^{-i\theta})$ on that control.
Because every support has bounded size and diameter, these standard controlled-gate constructions~\cite{Barenco} have constant gate count and routing depth.
They cover both the selected Pauli insertions and the product-formula rotations, including their scalar phases.

For the decoder, Toffoli and controlled-SWAP gates are also needed.
A Toffoli is obtained by conjugating controlled-controlled-$Z$ with $H_{\rm Had}$ on the target.
In the elementary phase-decomposition method~\cite{Barenco}, the integer identity
\[
 4xyz=x+y+z-(x\oplus y)-(x\oplus z)-(y\oplus z)
          +(x\oplus y\oplus z)
\]
for $x,y,z\in\{0,1\}$ expresses the phase $e^{i\pi xyz}$ as $\pi/4$ and $-\pi/4$ phases on bits and parities.
Each parity can be computed into one clean work qubit, phased, and uncomputed; the same work qubit is reused for the fixed number of parities.
These elementary gates implement a controlled swap of $a,b$ with control $c$ in the order
\[
 \CX(b;a)\ ;\ \CCX(c,a;b)\ ;\ \CX(b;a).
\]
The CNOTs cancel on $c=0$, while the three gates exchange $a$ and $b$ on $c=1$.
Conjugating a control wire by $X$ gives a negative control.
Thus all constant-size gates used by the decoder and phase-encounter network have explicit elementary implementations.

The table lookups further require a comparison with a fixed address.
For a $q$-bit string $x$ and a fixed classical string $x_*$, let $\operatorname{EQ}_{x_*}$ have action
\[
 \operatorname{EQ}_{x_*}\ket{x}\ket{0}_{c}\ket{0}_{\rm work}
 =\ket{x}\ket{\ind{x=x_*}}_{c}\ket{0}_{\rm work}.
\]
Conjugating the positions where $x_*$ is zero by $X$ reduces this comparison to an AND of the resulting bits.
The AND is computed along fresh targets, copied to $c$, and uncomputed; undoing the initial $X$ gates then restores $x$, leaving only the comparison output.
Interleaving operands and work on a path gives depth, gate count, and work $\mathcal O(q)$.
After $c$ has controlled its intended gate, the same XOR circuit clears it.
A control used at several adjacent positions can likewise be copied along a CNOT chain and then cleared in reverse.
This is valid for coherent inputs because the intervening gates preserve the copied basis values.

\subsection{Reversible addition and comparison}
The prefix calculation retains its operands so that additions and comparison outputs can later be erased. We use reversible carry propagation~\cite{VBE} and compute--copy--uncompute~\cite{Bennett} to obtain the required linear cost on a path.
Let $x=(x_0,\ldots,x_{q-1})$ and $y=(y_0,\ldots,y_{q-1})$ be retained $q$-bit inputs, indexed from the least significant bit.
The addition circuit first writes the sum into a fresh output:
\begin{equation}
 \operatorname{ADD}_q\ket{x}\ket{y}\ket0_{\rm sum}\ket0_{\rm work}
 =\ket{x}\ket{y}\ket{(x+y)\bmod2^q}_{\rm sum}\ket0_{\rm work}.
 \label{eq:retainedadder}
\end{equation}
Starting with $c_0:=0$, the carry and sum are written into fresh targets according to
\begin{equation}
 c_{i+1}:=x_iy_i\oplus x_ic_i\oplus y_ic_i,
 \qquad s_i:=x_i\oplus y_i\oplus c_i
 \quad(0\le i<q).
 \label{eq:carry}
\end{equation}
The carry is the majority of $x_i,y_i,c_i$, so three Toffolis write it into a zero target, while CNOTs write the sum bit.
Once all sum bits are present, reversing the carry gates from high to low significance clears their work without changing either the sum or the retained inputs.
A constant number of qubits at each bit position keeps $x_i,y_i,c_i,c_{i+1},s_i$ within a fixed distance of one another, giving depth, gate count, and work $\mathcal O(q)$.
For modular addition, only the low $q$ sum bits are retained; the final overflow carry is cleared by the same reverse sweep.

The same circuit gives the in-place update $y\mapsto y+x\pmod{2^q}$, but its old output must be erased.
Computing the sum into zero work and exchanging that register with $y$ leaves the new value in $y$ and the old value in work.
Because the retained $x$ recovers the old value as $(y_{\rm new}-x)\bmod2^q$, a subtraction into separate work, an XOR into the old-value register, and inverse subtraction clear that register.
For subtraction $x-y$, the required borrow recurrence starts at $t_0:=0$ and is
\[
 t_{i+1}:=(1-x_i)y_i\oplus(1-x_i)t_i\oplus y_it_i,
 \qquad d_i:=x_i\oplus y_i\oplus t_i.
\]
Here the $d_i$ are the difference bits modulo $2^q$.
As in addition, reversing the borrow computation retains that difference and clears the carries.
For comparison, the required output is instead the final borrow $t_q=\ind{x<y}$. Copying it to a condition bit permits the entire subtraction, including its difference output, to be reversed.
The resulting circuit therefore has the clean action
\begin{equation}
 \operatorname{LT}_q\ket{x}\ket{y}\ket0_c\ket0_{\rm work}
 =\ket{x}\ket{y}\ket{\ind{x<y}}_c\ket0_{\rm work}.
 \label{eq:cleancomparison}
\end{equation}
All three operations use depth, gate count, and work $\mathcal O(q)$ in the aligned layout.
The in-place update needs a fixed number of extra length-$q$ registers, so its cleanup does not change these bounds.

Controlled addition $\ket c\ket x\ket y\mapsto\ket c\ket x\ket{(y+cx)\bmod2^q}$ reduces to the same update with operand $cx$.
A clean CNOT chain distributes the retained control next to $x$, where Toffolis write the bits $cx_i$ into work. The in-place adder acts on these bits and $y$; afterward, the retained $c$ and $x$ clear the product bits and control copies.
Each step has linear depth, gate count, and work, establishing the bounds used for $\calT$ and Eq.~\eqref{eq:width}.
Space-efficient ripple adders~\cite{Cuccaro} provide an alternative implementation with the same required interface.

For the denominator recurrence in Eq.~\eqref{eq:moddenom}, the preceding value remains in $\reg B_{u-1}$ while the new value is written into $\reg B_u$.
Controlled on $X_u=1$, the circuit copies the preceding value and adds $J_u$ and $p+1$ modulo $L_b$; on $X_u=0$, the fresh target remains zero.
At $u=0$, the copying step is absent because there is no preceding value.
Since $L_b=2^{b_{\rm den}}$, these are exactly the modular operations above.
The $b_{\rm count}$-bit counter has $2^{b_{\rm count}}>N$, so incrementing it by $1-X_u$ never overflows.
Before the update, comparison with the classical value $r$ writes $e_u=\ind{j_u<r}$ into $\reg E_u$.
During $\calT^\dagger$, reversing the swap and increment recovers $j_u$ before the inverse comparison clears $\reg E_u$.
The retained inputs therefore suffice to erase every stored prefix output and all arithmetic work in the prescribed reverse order.

\subsection{Numerical precision and finite gate sets}
The complete simulation contains $G_{\rm one}$ one-qubit gates. We allocate an error budget $\epsilon_{\rm syn}>0$ to their evaluation and synthesis.
If gate $j$ has evaluation error $\eta_j^{\rm num}$ and synthesis error $\eta_j^{\rm syn}$, its total error is at most their sum.
Telescoping the full product of unitary gates then gives
\begin{equation}
 \norm{\widetilde{\calU}\iota_0-\calU \iota_0}
 \le\norm{\widetilde{\calU}-\calU}
 \le\sum_j(\eta_j^{\rm num}+\eta_j^{\rm syn}).
 \label{eq:synthesisbudget}
\end{equation}
Here $\calU$ uses the specified exact rotations and $\widetilde{\calU}$ uses their numerical implementations.

This norm also retains the overall phase, which must be accounted for when using synthesis in $SU(2)$.
Each one-qubit gate has the form $G_j=e^{i\phi_j}V_j$ with $V_j\in SU(2)$, so the accumulated scalar phase is $\Phi_{\rm syn}:=\sum_{j=1}^{G_{\rm one}}\phi_j$.
Replacing $G_j$ by $V_j$ removes the scalar $e^{i\Phi_{\rm syn}}$ from the complete circuit.
To restore that phase, we reserve one additional local qubit, untouched by all other gates, and apply $R_z(-2\Phi_{\rm syn})$ to its initial state $\ket0$.
This contributes precisely $e^{i\Phi_{\rm syn}}$ while leaving its basis state unchanged.
Thus $G_{\rm one}+1$ single-qubit $SU(2)$ rotations and one extra ancilla suffice to preserve the target isometry, including its phase.

For a finite inverse-closed single-qubit gate set dense in $SU(2)$, together with exact CNOT gates, the Solovay--Kitaev theorem~\cite{SK} approximates each rotation to error $\eta$ with $\mathcal O(\log^{c_{\rm SK}}(1/\eta))$ gates.
We choose the evaluation and synthesis tolerances so that every final approximation has error at most $\epsilon_{\rm syn}/(G_{\rm one}+1)$.
Equation~\eqref{eq:synthesisbudget} then bounds the added error by $\epsilon_{\rm syn}$, while gate count and worst-case depth acquire a factor $\mathcal O(\log^{c_{\rm SK}}((G_{\rm one}+1)/\epsilon_{\rm syn}))$.

For example, dividing the final error budget equally between simulation and synthesis gives a finite-gate-set implementation with this additional precision dependence.

For the Clifford+$T$ gate set, generated by $H_{\rm Had}$, CNOT, and $T_{\rm gate}:=\operatorname{diag}(1,e^{i\pi/4})$, specialized arithmetic methods are also available.
Kliuchnikov, Maslov, and Mosca characterize and synthesize exactly representable single-qubit unitaries~\cite{KMM}, while Ross and Selinger approximate arbitrary $z$ rotations without ancillas~\cite{RossSelinger}.
These methods can replace the general synthesis step at the same sites; the resource guarantee stated above uses Solovay--Kitaev.

\section{\texorpdfstring{Proof of Fact~\ref{fact:spatial}}{Proof of Fact 1}}\label{app:spatial-proof}
We prove the precise form of Fact~\ref{fact:spatial} stated in Lemma~\ref{lem:spatial}, which reduces the simulation upper bound to short-time block evolutions. The construction follows HHKL~\cite[Lemmas~5 and~6, Appendix~C.4]{HHKL}. A Lieb--Robinson restriction estimate~\cite{NS,HHKL} localizes the effect of each removed strip boundary to a narrow slab. Applying the resulting correction in each coordinate gives a constant number of layers of restricted evolutions.

Throughout this appendix, $H$ denotes the normalized $\bar H$ of Eq.~\eqref{eq:normalizedrestriction}, and $U_Y(t)=e^{-iH_Yt}$ uses its declared-support restriction.
The corresponding Heisenberg evolution is
\[
 \tau_t^Y(O):=U_Y(t)^\dagger O U_Y(t).
\]
The following restriction estimate is stated so that it can also be applied after earlier strip boundaries have been removed.

\begin{lemma}[Restriction of a local evolution]\label{lem:restriction}
Suppose the local strength is at most one. Let $\nu\ge1$ be a fixed upper bound on every interaction support size $|Z|$, and let $R_{\rm LR}:=\max\{1,R\}$. If $O_X$ is supported on a nonempty set $X\subseteq Y$ and $d:=\dist(X,\Lambda\setminus Y)$, the evolution difference satisfies the following bound for $|t|\le1$ and any fixed $\mu>0$:
\begin{equation}
 \norm{\tau_t^\Lambda(O_X)-\tau_t^Y(O_X)}
 \le |X|\norm{O_X}\sum_{q\ge\lceil d/R_{\rm LR}\rceil}
             \frac{(2\nu|t|)^q}{q!}
 \le C_{\rm LR}|X|\norm{O_X}e^{-\mu d},
 \label{eq:LR}
\end{equation}
where $C_{\rm LR}:=\exp(2\nu e^{\mu R_{\rm LR}})$. If $Y=\Lambda$, the difference and the tail are zero. The same estimates hold after removing any collection of interaction terms.
\end{lemma}
\begin{proof}
It suffices to consider $t\ge0$, since negating the Hamiltonian gives the negative-time case.
For a Hamiltonian $G=\sum_Zg_Z$ obtained by retaining any subset of the interactions, its Heisenberg evolution is denoted by $\tau_t^G(O):=e^{iGt}Oe^{-iGt}$.
We measure how an observable initially supported on $X$ can affect a fixed $B$ through the normalized commutator
\[
 C_B(X,t):=\sup_{\substack{A\ne0\\\supp(A)\subseteq X}}
       \frac{\norm{[\tau_t^G(A),B]}}{\norm A}.
\]
Only terms meeting $X$ contribute to $[G,A]$.
Thus, with $G_{\cap X}:=\sum_{Z:Z\cap X\ne\varnothing}g_Z$ and $f(t):=[\tau_t^G(A),B]$, differentiation and the Jacobi identity give
\[
 f'(t)=i[\tau_t^G(G_{\cap X}),f(t)]
       -i[\tau_t^G(A),[\tau_t^G(G_{\cap X}),B]].
\]
The first term generates norm-preserving unitary conjugation; only the second can increase the commutator norm.
Variation of constants, followed by $\norm{[A,B]}\le2\norm A\norm B$ and the supremum over $A$, therefore gives the standard recursion
\begin{equation}
 C_B(X,t)\le C_B(X,0)
   +2\sum_{Z:Z\cap X\ne\varnothing}\norm{g_Z}
       \int_0^t C_B(Z,s)\,ds.
 \label{eq:LRrecursion}
\end{equation}
Iterating this recursion produces chains with $X\cap Z_1\ne\varnothing$ and $Z_j\cap Z_{j+1}\ne\varnothing$.
The local strength bound controls the first support by $\sum_{Z_1\cap X\ne\varnothing}\norm{g_{Z_1}}\le|X|$.
At every subsequent step, it gives
\[
 \sum_{Z':Z'\cap Z_j\ne\varnothing}\norm{g_{Z'}}
 \le\sum_{x\in Z_j}\sum_{Z'\ni x}\norm{g_{Z'}}\le|Z_j|\le\nu.
\]
Hence the total weight of chains of length $q\ge1$ is at most $|X|\nu^{q-1}$.
The remainder after $q$ recursion steps is bounded by $2\norm B\,|X|\nu^{q-1}(2t)^q/q!$ and tends to zero.
If $B$ is supported on $Z_B$, the initial bound $C_B(X,0)\le2\norm B\,\ind{X\cap Z_B\ne\varnothing}$ restricts the iterated sum to chains whose last support meets $Z_B$.
Together with the vanishing remainder, this gives
\begin{equation}
 C_B(X,t)\le2\norm B\sum_{q=0}^\infty\frac{(2t)^q}{q!}
 \sum_{\substack{X\cap Z_1\ne\varnothing,\ Z_j\cap Z_{j+1}\ne\varnothing\\
                  Z_q\cap Z_B\ne\varnothing}}
        \prod_{j=1}^q\norm{g_{Z_j}}.
 \label{eq:LRchains}
\end{equation}

In particular, each step follows from overlap with the preceding interaction support, rather than with the union of all earlier supports; this is why the chain weight grows only geometrically in $q$.

To turn the commutator bound into a restriction error, we differentiate $\tau_{t-s}^\Lambda(\tau_s^Y(O_X))$ and integrate between its endpoints.
Only the omitted boundary interactions contribute, giving
\begin{equation}
 \norm{\tau_t^\Lambda(O_X)-\tau_t^Y(O_X)}
 \le\int_0^t\sum_{Z_B\in\partial\mathcal I_Y}
          \norm{[\tau_s^Y(O_X),h_{Z_B}]}\,ds,
 \label{eq:LRduhamel}
\end{equation}
where $\partial\mathcal I_Y:=\{Z_B:Z_B\cap Y\ne\varnothing,\ Z_B\setminus Y\ne\varnothing\}$.
An omitted term wholly outside $Y$ commutes with $\tau_s^Y(O_X)$, which explains why no other terms appear.

We apply Eq.~\eqref{eq:LRchains} with $G=H_Y$ and $B=h_{Z_B}$ and sum over the crossing supports $Z_B$ in Eq.~\eqref{eq:LRduhamel}.
A term with $q$ retained interactions corresponds to
\[
 X\cap Z_1\ne\varnothing,\quad Z_1\cap Z_2\ne\varnothing,\quad\ldots,\quad
 Z_q\cap Z_B\ne\varnothing,\qquad Z_B\setminus Y\ne\varnothing.
\]
For $q=0$, this condition is simply $X\cap Z_B\ne\varnothing$.
In either case, the $q+1$ interaction supports connect $X$ to the complement of $Y$.
Since each has diameter at most $R_{\rm LR}$, no such term exists unless $(q+1)R_{\rm LR}\ge d$.
The boundary interaction is included in the same local-strength sum, so
\[
 \sum_{\substack{Z_1,\ldots,Z_q\subseteq Y,\ Z_B\in\partial\mathcal I_Y\\
                  \text{the above intersections hold}}}
       \left(\prod_{j=1}^q\norm{h_{Z_j}}\right)\norm{h_{Z_B}}
 \le |X|\nu^q.
\]
Substituting the distance constraint and this weight bound into Eq.~\eqref{eq:LRduhamel} gives
\begin{align*}
 2|X|\norm{O_X}
 \sum_{\substack{q\ge0\\(q+1)R_{\rm LR}\ge d}}
      \nu^q\int_0^t\frac{(2s)^q}{q!}\,ds
 &=|X|\norm{O_X}
 \sum_{\substack{k\ge1\\kR_{\rm LR}\ge d}}
       \frac{2^k\nu^{k-1}t^k}{k!}\\
 &\le |X|\norm{O_X}
 \sum_{k\ge\lceil d/R_{\rm LR}\rceil}\frac{(2\nu t)^k}{k!}.
\end{align*}
This is the first inequality of Eq.~\eqref{eq:LR}.
For the exponential form, every remaining index satisfies $kR_{\rm LR}\ge d$, so multiplying by the corresponding exponential weight bounds the tail by
\[
 \sum_{k\ge\lceil d/R_{\rm LR}\rceil}\frac{(2\nu|t|)^k}{k!}
 \le e^{-\mu d}\sum_{k\ge0}
       \frac{(2\nu e^{\mu R_{\rm LR}}|t|)^k}{k!}
 \le C_{\rm LR}e^{-\mu d}.
\]
All steps used only range, support size, and local strength, none of which increases when interactions are removed.
The same bound is therefore available for every restricted Hamiltonian encountered in the decomposition.
\end{proof}

\begin{proof}[Proof of Lemma~\ref{lem:spatial}]
Again, it suffices to treat $0\le t\le1$; a sign change gives negative time.
Our first step is to separate the lattice into wide strips and express the omitted interactions as corrections near their boundaries.

We partition the first coordinate into strips $B_i$ of widths between $4\ell_{\rm buf}$ and $8\ell_{\rm buf}$.
When its length is $L_1\ge8\ell_{\rm buf}$, distributing the sites as evenly as possible among $\lfloor L_1/(4\ell_{\rm buf})\rfloor$ intervals gives such a partition.
The Hamiltonian with all interstrip interactions removed is
\[
 H_0:=\sum_iH_{B_i},\qquad H-H_0=\sum_iV_i,
\]
where $V_i$ contains the terms crossing between $B_i$ and $B_{i+1}$.
The threshold is chosen with $\ell_0>R$, ensuring that no interaction crosses two strip boundaries.
Around boundary $i$, a slab $P_i$ extends $\ell_{\rm buf}$ sites on each side, with halves $P_i^-$ and $P_i^+$.
The slabs are disjoint and contain every support in $V_i$, which gives the identities
\begin{equation}
 H_{P_i}=H_{P_i^-}+H_{P_i^+}+V_i,
 \qquad (H_0)_{P_i}=H_{P_i^-}+H_{P_i^+}.
 \label{eq:boundaryhamiltonian}
\end{equation}
The second identity follows because $H_0$ has already removed all interactions between the two halves.
These identities let a correction supported on $P_i$ replace the full interaction-picture evolution of $V_i$.

The evolution under the decoupled Hamiltonian is $U_0(s):=e^{-iH_0s}=\prod_iU_{B_i}(s)$.
The remaining correction $W(s):=U_0(s)^\dagger U_\Lambda(s)$ obeys
\begin{align}
 W'(s)&=iU_0(s)^\dagger H_0U_\Lambda(s)
            -iU_0(s)^\dagger HU_\Lambda(s)\notag\\
 &=-iK_I(s)W(s),\qquad
 K_I(s):=U_0(s)^\dagger(H-H_0)U_0(s)
       =\sum_iU_0(s)^\dagger V_iU_0(s).
 \label{eq:interactiongenerator}
\end{align}
Every term of $V_i$ is at distance at least $\ell_{\rm buf}-R$ from the complement of $P_i$.
The restriction estimate therefore allows its conjugation by $U_0(s)$ to be replaced by conjugation inside $P_i$.
Equation~\eqref{eq:boundaryhamiltonian} identifies the restricted unitaries and the resulting approximate generator as
\[
 U_{i,0}(s):=U_{P_i^-}(s)U_{P_i^+}(s),\qquad
 \widetilde K_I(s):=\sum_iU_{i,0}(s)^\dagger V_iU_{i,0}(s).
\]
Each support has $|Z|\le\nu$, while the local strength bound gives $\sum_i\sum_{Z\text{ in }V_i}\norm{h_Z}=\mathcal O(n)$.
Summing Eq.~\eqref{eq:LR} over these terms consequently yields
\begin{equation}
 \norm{K_I(s)-\widetilde K_I(s)}
 \le C_{\rm LR}\nu e^{\mu R}e^{-\mu\ell_{\rm buf}}
               \sum_i\sum_{Z\text{ in }V_i}\norm{h_Z}
 \le Cn e^{-\mu\ell_{\rm buf}}.
 \label{eq:boundarygeneratorerror}
\end{equation}
Here $C$ depends only on fixed model constants, not on the lattice size or buffer width.
Thus localization introduces an error exponentially small in the slab width, even after all boundaries are included.

The local replacement also has an explicit expression in restricted evolutions.
For slab $i$, the correction is $E_i(t):=U_{i,0}(t)^\dagger U_{P_i}(t)$.
By Eq.~\eqref{eq:boundaryhamiltonian}, its derivative is
\begin{align}
 E_i'(t)
 &=iU_{i,0}(t)^\dagger(H_{P_i^-}+H_{P_i^+})U_{P_i}(t)
       -iU_{i,0}(t)^\dagger H_{P_i}U_{P_i}(t)\notag\\
 &=-iU_{i,0}(t)^\dagger V_iU_{i,0}(t)E_i(t).
 \label{eq:boundarycorrectionderivative}
\end{align}
This is exactly the evolution generated by the $i$th summand of $\widetilde K_I$.
Since the slabs are disjoint, their corrections commute, and $\widetilde W(t):=\prod_iE_i(t)$ satisfies $\widetilde W'=-i\widetilde K_I\widetilde W$ with $\widetilde W(0)=I$.
We can therefore compare $W$ and $\widetilde W$ directly through their generators.
Writing $W(t,s):=W(t)W(s)^\dagger$, Duhamel's formula gives
\begin{equation}
 W(t)-\widetilde W(t)
 =-i\int_0^t W(t,s)\bigl(K_I(s)-\widetilde K_I(s)\bigr)
                   \widetilde W(s)\,ds.
 \label{eq:spatialduhamel}
\end{equation}
The propagators on either side of the generator difference are unitary.
Combining Eqs.~\eqref{eq:boundarygeneratorerror} and~\eqref{eq:spatialduhamel} therefore bounds the correction error by
\[
 \norm{U_0(t)^\dagger U_\Lambda(t)-\prod_iE_i(t)}
 \le |t|Cn e^{-\mu\ell_{\rm buf}}
 \le Cn e^{-\mu\ell_{\rm buf}}.
\]
Multiplying by $U_0(t)$ returns to the original evolution without changing the norm error.
Substituting the expressions for the slab corrections then gives
\begin{equation}
 \widetilde U_\Lambda(t):=
 \left(\prod_iU_{B_i}(t)\right)
 \left(\prod_iU_{P_i^-}(t)^\dagger U_{P_i^+}(t)^\dagger\right)
 \left(\prod_iU_{P_i}(t)\right),
 \label{eq:spatialformula}
\end{equation}
with $\norm{U_\Lambda(t)-\widetilde U_\Lambda(t)}\le Cn e^{-\mu\ell_{\rm buf}}$.
In forming the last two parenthesized groups, only factors in distinct disjoint slabs were exchanged.
Thus each parenthesized group is a layer of disjoint restricted evolutions.
The strips have total volume $|\Lambda|$, and the slabs and their halves each have total volume at most $|\Lambda|$, so the three layers together have volume at most $3|\Lambda|$.

The remaining coordinate directions are treated by applying the same construction to each current factor.
A direction not yet processed still has its original side length, so the strip partition applies; processed directions already have widths in $[\ell_{\rm buf},8\ell_{\rm buf}]$.
For an inverse factor, we decompose its forward evolution and take the adjoint, which reverses the three layers.
Each direction therefore multiplies both the layer count and total region volume by at most three.
After $D$ directions, there are at most $3^D$ layers, every side lies in $[\ell_{\rm buf},8\ell_{\rm buf}]$, and the volumes sum to at most $3^D|\Lambda|$.
Telescoping the unitary products bounds the accumulated error by
\[
 C|\Lambda|e^{-\mu\ell_{\rm buf}}\sum_{j=0}^{D-1}3^j
 \le C'n e^{-\mu\ell_{\rm buf}}.
\]
Because $D$ is fixed, the geometric sum is absorbed into the constant, proving the required error bound and completing the decomposition.

\end{proof}

\begingroup
\small
\raggedright
\let\standardthebibliography\thebibliography
\renewcommand{\thebibliography}[1]{%
  \standardthebibliography{#1}%
  \setlength{\itemsep}{2pt plus 1pt minus 1pt}}
\bibliographystyle{alphaurl}
\bibliography{references}
\endgroup
\end{document}